\documentclass[12pt]{article}
\usepackage[utf8]{inputenc}

\usepackage{amsmath, amssymb, amsfonts}
\usepackage{mathtools}
\usepackage{color}
\usepackage{bbm}
\usepackage{subcaption}

\usepackage{graphicx}
\usepackage{geometry}
\usepackage{booktabs}

\usepackage{amsthm}
\usepackage{aliascnt}

\usepackage{algorithm}
\usepackage{algpseudocode}

\usepackage[hidelinks]{hyperref}
\usepackage[nameinlink]{cleveref}

\usepackage{footmisc}
\usepackage[round,authoryear]{natbib}
\usepackage{authblk}

\theoremstyle{plain}
\newtheorem{theorem}{Theorem}

\newaliascnt{lemma}{theorem}
\newtheorem{lemma}[lemma]{Lemma}
\aliascntresetthe{lemma}

\newaliascnt{proposition}{theorem}
\newtheorem{proposition}[proposition]{Proposition}
\aliascntresetthe{proposition}

\newaliascnt{corollary}{theorem}
\newtheorem{corollary}[corollary]{Corollary}
\aliascntresetthe{corollary}

\theoremstyle{definition}
\newaliascnt{definition}{theorem}
\newtheorem{definition}[definition]{Definition}
\aliascntresetthe{definition}

\theoremstyle{remark}
\newaliascnt{remark}{theorem}
\newtheorem{remark}[remark]{Remark}
\aliascntresetthe{remark}

\crefname{theorem}{theorem}{theorems}
\Crefname{theorem}{Theorem}{Theorems}

\crefname{lemma}{lemma}{lemmas}
\Crefname{lemma}{Lemma}{Lemmas}

\crefname{proposition}{proposition}{propositions}
\Crefname{proposition}{Proposition}{Propositions}

\crefname{corollary}{corollary}{corollaries}
\Crefname{corollary}{Corollary}{Corollaries}

\crefname{definition}{definition}{definitions}
\Crefname{definition}{Definition}{Definitions}

\crefname{remark}{remark}{remarks}
\Crefname{remark}{Remark}{Remarks}

\newcommand{\R}{\mathbb{R}}
\newcommand{\E}{\mathbb{E}}

\newcommand{\Mcal}{\mathcal{M}}

\title{From Arbitrage Removal to Density Extraction: A Model-Free Framework for Short-Dated Options}

\author{
Aaron Wizman\thanks{
CEREMADE, Universit\'e Paris Dauphine--PSL, Pl. du Mar\'echal de Lattre de Tassigny, 75016 Paris, France.
Email: \texttt{aaron.wizman@dauphine.psl.eu}.
}
\qquad
Gabriel Turinici\thanks{
CEREMADE, Universit\'e Paris Dauphine--PSL, Pl. du Mar\'echal de Lattre de Tassigny, 75016 Paris, France. 
Email: \texttt{gabriel.turinici@dauphine.fr}.
}
\qquad
Gregory Merran\thanks{
DRM--Finance, Universit\'e Paris Dauphine--PSL, Pl. du Mar\'echal de Lattre de Tassigny, 75016 Paris, France.
Email: \texttt{gregory.merran@dauphine.eu}.
}
}

\date{May 2026}

\begin{document}

\maketitle
\begin{abstract}
We study risk-neutral density extraction from short-dated option
chains. As expiry approaches, option premia decline and bid--ask
spreads can be large relative to prices, making mid quotes
particularly uninformative. Stale or asynchronous quotes may also
generate potential static arbitrages, rendering standard procedures infeasible or unstable. We develop a model-free pipeline that treats
bid--ask quotes as the primitive market constraint. The pipeline consists of two steps. First, a procedure called ``Arbitrage Removal Iterative Executable Strategy'' (ARIES) filters
executable static arbitrage at quoted bid and ask prices under
market-depth constraints. Second, the ``Smooth Entropic Density EXtraction'' (SEDEx) then recovers the density through a criterion leveraging
smoothness and entropy under bid--ask constraints. We test the pipeline on synthetic Heston panels and short-dated SPX option data, sampled from a few hours to one week before expiry. Computation is fast and returns robust densities across various market conditions, including
scheduled macroeconomic announcements. As an empirical application, we use the recovered densities to construct short dated implied-volatility
smiles.
\end{abstract}

\noindent\textbf{Keywords:} 
risk-neutral density;
short-dated options;
0DTE options;
static arbitrage;
bid--ask quotes;
implied volatility smile.

\vspace{0.5em}

\noindent\textbf{JEL Classification:} C61, C63, G12, G13.

\section{Introduction}
\label{sec:introduction}
The risk-neutral density (hereafter RND) is the market-implied
distribution of an asset's future price. It is a central object in derivatives research, with applications to option pricing, implied-volatility surface construction, tail-risk management \citep{ait2000nonparametric}, and the study of market beliefs and preferences \citep{beber2006effect}, among others. Recovering the RND from observed option quotes is a
nontrivial inverse problem. Market data are discrete and finite in
strike coverage, may be noisy or affected by arbitrage violations,
and are observed through bid--ask quotes rather than exact prices.

The difficulty is most acute for short-dated options, defined here
as contracts with at most one week to expiry; we use $n$DTE to
denote contracts expiring exactly $n$ days after the trading date.
According to market statistics provided by Cboe Global
Markets,\footnote{We thank Cboe Global Markets for providing the
market statistics reported in this paragraph.} U.S. listed options
averaged approximately $68.6$ million contracts per day in the first
quarter of $2026$, with 0DTE volume alone reaching $18.55$ million.
Contracts expiring within one week now exceed those with longer
maturities. The phenomenon is most visible in S\&P~500 index
options: SPX averaged $4.9$ million contracts per day, with
roughly $58\%$ in 0DTE. Originally confined to a handful of index
products, short-dated cycles have since spread across asset classes,
including single names, rates, metals, crypto-linked products, and
selected commodities. They are used by both retail and institutional
investors, with applications ranging from short-horizon volatility
exposure and intraday hedging to low-premium directional views and
event-risk trading.

The recent literature on short-dated options highlights the
specificity of the very short tenor and motivates studying RND
extraction in this regime as a problem of its own.
\citet{andersen2017short} show that weekly S\&P~500 options carry
distinctive information about short-horizon volatility and jump risk,
and develop a semi-nonparametric procedure tailored to that setting.
\citet{todorov2024intraday} exploit the ratio of risk-neutral
variance measures from 0DTE and 1DTE S\&P~500 options to estimate
intraday volatility patterns nonparametrically. Their identification
rests on a short-time asymptotic regime in which the stochastic
component of volatility cancels. \citet{bandi20230dte} develop
local-in-time pricing expansions designed for 0DTE options, with
closed-form skewness and kurtosis adjustments, and document gains in
pricing and hedging.

Short-dated options are also where the standard treatment of option
data becomes most fragile. As expiry approaches, premia decline, and
for deep out-of-the-money contracts the bid--ask spread is of the
same order of magnitude as the premium itself. Distinct strikes may
then share identical mid quotes, making the corresponding
observations uninformative. At very short maturities, the bid--ask
interval is therefore the natural primitive market observation. Quoted
panels may also contain potential static arbitrage opportunities arising from
asynchronous or stale quotes, which can make the subsequent inverse
problem infeasible or numerically unstable.

These observations motivate our approach whose main goal is to  recover
risk-neutral densities from short-dated option chains under bid--ask
pricing constraints, after filtering executable static arbitrage. Our method
has two autonomous components. First, the ``Arbitrage Removal Iterative Executable Strategy'' (ARIES) is a filtering procedure that detects and removes static arbitrage opportunities from
bid--ask quotes. It operates directly on executable quotes subject to
market-depth constraints. When multiple quotes are simultaneously
candidates for removal, ARIES selects the one with the smallest
available size, a criterion that retains as much market information
as possible. Second, the ``Smooth
Entropic Density EXtraction'' (SEDEx) operates on the cleaned panel. It
recovers the RND by maximizing entropy and enforcing smoothness
under the bid--ask constraints. The two procedures are conceptually independent
and can be used separately. When combined, they yield a natural
pipeline for short-dated option data. Both operate on a
single-maturity option cross-section; the joint treatment of the
option surface across maturities lies beyond the scope of this
paper.

The proposed methodology connects to two parts of the literature.

The first concerns no-arbitrage constraints on option panels and the
operational treatment of arbitrage in observed data. In this context, \citet{carr2005note} provide sufficient conditions for the absence of
static arbitrage on regular call price grids, expressed through
adjacent call-spread, butterfly, and calendar inequalities.
\citet{davis2007range} extend the analysis to finite option panels.
They clarify the distinction between consistency with an
arbitrage-free martingale model, model-independent arbitrage, and
weak arbitrage. \citet{Cousot2006} is the closest theoretical
reference for our setting. He derives necessary inequalities for the
absence of arbitrage formulated directly on bid and ask
quotes, and shows that under these inequalities there exists a
martingale measure whose call prices lie inside the observed
intervals. \citet{cohen2020detecting} formulate arbitrage repair as a linear program. No-arbitrage
relations are imposed as hard constraints on a single repaired price
vector, while bid--ask spreads enter as soft bounds in the objective.
The repaired vector is arbitrage-free by construction, but may lie
outside the original bid--ask intervals. ARIES departs from this repair paradigm by filtering executable violations at bid and ask prices, subject to market-depth constraints.

The second part is the literature on RND extraction. This literature
splits into two routes. Price-first and volatility-first methods fit a
call price function or implied volatility surface before applying the
seminal identity of \citet{BreedenLitzenberger1978}. Examples include
parametric volatility-smile smoothing \citep{shimko1993} and
nonparametric kernel estimation of the call price function
\citep{ait1998nonparametric}. The latter highlights the slower
convergence of derivatives relative to the estimated function itself,
the so-called curse of differentiation. By contrast, density-first
methods target the terminal distribution directly.

Among density-first methods, parametric and semi-parametric approaches
impose structure directly on the terminal distribution.
\citet{jarrow1982approximate} approximate the risk-neutral
distribution by an expansion around a reference lognormal, with
corrections in higher moments. \citet{bahra1997implied}
represents the RND as a mixture of lognormal components fitted to the
observed cross-section. The flexibility of these specifications is
controlled by a finite set of parameters rather than by an
unrestricted functional class. Parametric models such as
\citet{BlackScholes1973} and \citet{heston1993closed} provide a
related benchmark in which the RND is induced by a specified
risk-neutral dynamic.

SEDEx belongs to the class of fully nonparametric density-first
methods, which relax these restrictions on the terminal distribution.
Among such approaches, \citet{jackwerth1996recovering} select
terminal probabilities by minimizing a curvature criterion.
SEDEx also regularizes the recovered density, but rather includes a
first-order smoothness penalty: the $L^{2}$ norm of its first
derivative. This penalty controls local variation in the density and
acts as a first-order Tikhonov regularization of the discretized
inverse problem. Such regularization is necessary because the pricing operator that maps a high-dimensional
vector of terminal probabilities into a smaller set of option-price
constraints is rank-deficient by construction \citep{hansen1998rank}.

The criterion also includes an entropic component, in the spirit of
the maximum-entropy approach of \citet{buchen1996maximum}. Entropy is
therefore used as a complementary selection principle, rather than as
the sole objective. \citet{bondarenko2003estimation} provides another
nonparametric route by recovering the RND within a positive
convolution class.

A final distinction concerns the treatment of quotes. In most
applications, bid--ask intervals are collapsed into mid prices.
\citet{rubinstein1994implied} is a notable exception: he recovers
terminal probabilities under bid--ask constraints on a recombining
binomial tree and around a lognormal prior. SEDEx retains this
bid--ask-aware formulation, but recovers a continuous density without
a tree construction or a prior density.

As an empirical application of the recovered densities, we construct
implied-volatility smiles and compare them to SVI \citep{Gatheral2014}. The SVI parametrization has become a standard market reference for its
parsimony and flexibility, and works well in most standard
settings, except for ultra short-dated
expirations. This motivates a method specifically designed for
the short-dated regime.

Against this background, the paper makes two methodological
contributions. First, it introduces ARIES, an executable
static-arbitrage filter that treats bid--ask quotes as the primitive
market constraint. Second, it introduces SEDEx, a model-free RND recovery
method based on a hybrid criterion encouraging smoothness and maximizing entropy under bid--ask
constraints. We evaluate the resulting pipeline on synthetic and SPX
option data, showing that it is fast and stable across contrasting
market conditions, including scheduled events, and across short-dated
horizons from a few hours to one week before expiry. As an empirical
application, we use the recovered densities to construct
implied-volatility smiles in this short-dated setting.

The remainder of the paper is organized as follows.
Section~\ref{sec:Notation and Market Framework} introduces the market
framework and the two methodologies. Section~\ref{sec:theoretical_results}
establishes the theoretical guarantees for ARIES and SEDEx.
Section~\ref{sec:numerical_results} presents the numerical results on
synthetic and market data, including the implied-volatility smile
application. Section~\ref{sec:conclusion} concludes. Detailed proofs and additional information are deferred to the
appendices and supplementary material. 

\section{Market Framework and Methodology}
\label{sec:Notation and Market Framework}
\subsection{Market Setup}
We consider a market $\Mcal$  with no transaction costs that contains $N$ European call options with strikes $0<K_1 < \dots < K_N$, all expiring at time $T$. Their best bid and ask prices at $t=0$ are denoted by $C_i^{\text{bid}}$ and $C_i^{\text{ask}}$ satisfying $C_i^{\text{bid}} \leq C_i^{\text{ask}}$, with corresponding finite bid and ask sizes $Q_i^{\text{bid}}$ and $Q_i^{\text{ask}}$. The continuously compounded risk-free rate and dividend yield are $r$ and $q$, both assumed constant and deterministic.

The underlying asset is denoted by $S=(S_t)_{t \in [0,T]}$. In particular its price at $t=0$ is $S_0$; to account for possible dividends distributed before $T$ we discount its price by $q$ and thus we can assume that one unit of $S$ costs $S_0 e^{-qT}$ at time $t=0$. Note that we can see $S$ as a call with strike $K_0=0$.\\

We formalize the class of portfolios constructible from instruments in $\mathcal{M}$.
A portfolio is defined by a position vector
$w=(w_{K_0},\dots,w_{K_N})\in\mathbb{R}^{N+1}$, where $w_{K_i}$ denotes the number of call options held at strike $K_i$ for $i=0,\dots,N$, with $K_0=0$ corresponding to the underlying asset. Portfolios are assumed to be \emph{static}, with positions chosen at time~0 and held until maturity~$T$; short positions are permitted. The resulting terminal payoff is
\begin{equation}
V_T(w)
:=\sum_{i=0}^N w_{K_i}\,(S_T-K_i)_{+}.
\label{eq:VT-def}
\end{equation}

Because options are traded at bid and ask quotes, the time--$0$ cost of a given net position $w$
depends on its execution.

\begin{definition}[Execution cost at time $0$]
\label{def:pi0}
For any execution quantities $q^{\mathrm{ask}},q^{\mathrm{bid}}\in\mathbb{R}_+^{N}$ and an underlying position
$u\in\mathbb{R}$ satisfying the depth constraints
\begin{equation}
0 \le q_i^{\mathrm{bid}} \le Q_i^{\mathrm{bid}},
\qquad
0 \le q_i^{\mathrm{ask}} \le Q_i^{\mathrm{ask}},
\qquad i=1,\dots,N,
\label{eq:size-constraints-q-market}
\end{equation}
we define the associated time--$0$ execution cost by
\begin{equation}
\pi_0(q^{\mathrm{ask}}, q^{\mathrm{bid}}, u)
:=
\sum_{i=1}^{N} \bigl( q_i^{\mathrm{ask}} C_i^{\mathrm{ask}} - q_i^{\mathrm{bid}} C_i^{\mathrm{bid}} \bigr)
+ u\,S_0 e^{-qT}.
\label{eq:pi0-market}
\end{equation}
\end{definition}

\begin{remark}[Induced net weights]
\label{rem:w-from-q}
Any triple $(q^{\mathrm{ask}},q^{\mathrm{bid}},u)$ induces net weights $w=(w_{K_0},\dots,w_{K_N})$ via
\begin{equation}
w_{K_0}:=u,
\qquad
w_{K_i}:=q_i^{\mathrm{ask}}-q_i^{\mathrm{bid}},
\qquad i=1,\dots,N.
\label{eq:w-from-q-market}
\end{equation}
\end{remark}

\begin{definition}[Time--$0$ value]
\label{def:V0}
Assume that the net weights satisfy the depth constraints
\begin{equation}
-\,Q_i^{\mathrm{bid}} \le w_{K_i} \le Q_i^{\mathrm{ask}},
\qquad i=1,\dots,N.
\label{eq:w-box}
\end{equation}
We define the time--$0$ value of the static portfolio $w=(w_{K_0},\dots,w_{K_N})$ by
\begin{equation}
V_0(w)
:=
w_{K_0}\,S_0 e^{-qT}
+\sum_{i=1}^{N}\Bigl( w_{K_i}^{+}\,C_i^{\mathrm{ask}} - w_{K_i}^{-}\,C_i^{\mathrm{bid}} \Bigr),
\label{eq:V0-def}
\end{equation}
where $x^{+}=\max(x,0)$ and $x^{-}=\max(-x,0)$.
\end{definition}
\begin{remark}[Minimal-cost implementation]
\label{rem:V0-execution}
For a given net position $w$, the pair of execution quantities $(q^{\mathrm{ask}},q^{\mathrm{bid}})$ satisfying
\eqref{eq:w-from-q-market} and \eqref{eq:size-constraints-q-market} is not unique. The value $V_0(w)$ in
\Cref{def:V0} corresponds to the least costly implementation in the sense of \Cref{def:pi0}.
Equivalently, $V_0(w)$ is attained by the canonical
choice $q_i^{\mathrm{ask}}=w_{K_i}^{+}$ and $q_i^{\mathrm{bid}}=w_{K_i}^{-}$. In particular, any trade 
at the same strike
with $q_i^{\mathrm{ask}}>0$, $q_i^{\mathrm{bid}}>0$,
$q_i^{\mathrm{ask}}-q_i^{\mathrm{bid}}= w_{K_i}$,
will leave $w_{K_i}$ unchanged while increasing the cost by
$\min(q_i^{\mathrm{ask}},q_i^{\mathrm{bid}})\,(C_i^{\mathrm{ask}}-C_i^{\mathrm{bid}})\ge 0$.\\
\end{remark}
\noindent 
In the following sections, we employ two option strategies defined as follows.
\begin{definition}[Vertical spread and butterfly on a non-uniform strike grid]
\ \vspace{-0.5cm}\\
\begin{itemize}
\item \textbf{Vertical spread.}
For any indices \(0\le i<j\le N\),
the vertical spread with strikes \((K_i,K_j)\) is the portfolio
 with weights $\mathrm{VS}_{i,j}$:
\begin{align}
(\mathrm{VS}_{i,j})_{K_i} = 1, (\mathrm{VS}_{i,j})_{K_j} = -1,\qquad (\mathrm{VS}_{i,j})_{K_\ell}=0\ \text{for } \ell\notin\{i,j\}.
\end{align}
Note that the terminal payoff  of such a spread is:
\begin{align}
V_T(\mathrm{VS}_{i,j})=(S_T-K_i)_+-(S_T-K_j)_+ \ge 0.  
\end{align}

\item \textbf{Butterfly spread.}
For $a,b \in \{0, ..., N\}$, $a\neq b$ denote
\begin{align}
\alpha_{a,b}:=\frac{1}{K_b-K_a}.
\label{eq:definition_alphaij}
\end{align}
Consider indices \(0\le i<j<k\le N\). The butterfly spread supported on \((K_i,K_j,K_k)\) is the portfolio with weights 
$\mathrm{BF}_{i,j,k}$:
\begin{equation}
\begin{aligned}
(\mathrm{BF}_{i,j,k})_{K_i} &= \alpha_{i,j}, \\
(\mathrm{BF}_{i,j,k})_{K_j} &= -(\alpha_{i,j} + \alpha_{j,k}), \\
(\mathrm{BF}_{i,j,k})_{K_k} &= \alpha_{j,k},
\end{aligned}
\qquad
(\mathrm{BF}_{i,j,k})_{K_\ell} = 0 \ \text{for } \ell \notin \{i,j,k\}.
\end{equation}
Its terminal payoff is
\begin{equation}
\begin{aligned}
V_T(\mathrm{BF}_{i,j,k})
&= \alpha_{i,j}(S_T - K_i)_+
- (\alpha_{i,j} + \alpha_{j,k})(S_T - K_j)_+ \\
&\quad + \alpha_{j,k}(S_T - K_k)_+ \ge 0.
\end{aligned}
\end{equation}
\end{itemize}
\label{def:vs-bf}
\end{definition}
\subsection{Arbitrage Definitions}
We next define the notion of arbitrage for static portfolios built from the instruments available in the market $\Mcal$, see comments in \Cref{subsec:cousot_cond} for further motivation.
\begin{definition}[Weak Arbitrage, Strong Arbitrage, No Arbitrage]
\label{def:arbitrage}
\ \\ \vspace{-0.5cm}
\begin{enumerate}
    \item Let $(\Omega,\mathcal{F},\mathbb{P})$ be a probability space. We say that a static portfolio $V(w)$ together with a cash position $c$ generates a \emph{weak arbitrage} (with respect to $\mathbb{P}$) if
\begin{equation}
\begin{cases}
& V_{0}(w) + c = 0, \\ 
&  \mathbb{P}\left( \left\{  \boldsymbol{\omega} \in \Omega : V_T(w)(\boldsymbol{\omega})+ c e^{rT} \geq 0\right\} \right)=1,\\ 
& \mathbb{P}\left( \left\{  \boldsymbol{\omega} \in \Omega : V_T(w)(\boldsymbol{\omega})+ c e^{rT}>0\right\} \right) >0 .
\end{cases}
\label{eq:weak arbitrage}
\end{equation}

This notion depends on the choice of the measure $\mathbb{P}$.
\item 
A static portfolio with a cash position $c$ generates a \emph{strong arbitrage} if
\begin{equation} \label{eq:strong-arbitrage}V_{0}(w) + c=0,\ \  \text{ for all } \boldsymbol{\omega}\in\Omega: V_{T}(w)(\boldsymbol{\omega})+c e^{rT} > 0.\end{equation}
This yields a sure profit and is therefore an arbitrage under every probability measure on $(\Omega,\mathcal{F})$.
\item We say that no-arbitrage holds for the market $\mathcal M$
if for any probability measure $\mathbb P$ on
$(\mathbb R_+,\mathcal B(\mathbb R_+))$ that is absolutely continuous with
respect to the Lebesgue measure, no weak arbitrage exists.
\label{item:non_arb}
\end{enumerate}
\end{definition}
\noindent
We also accept the following assumption.
\begin{align}{\bf \textbf{(Hyp-Spot)}}\quad &
\parbox[t]{0.7\linewidth}{
The underlying asset can be traded frictionlessly, with zero bid--ask spread and no position limits.
Its spot price is consistent with the absence of static arbitrage relative to the option market.}
\end{align}

\begin{remark}
This standard assumption reflects the fact that spot markets typically display substantially higher liquidity and more efficient price formation than derivatives markets. 
As a result, potential violations of no-arbitrage conditions are attributed exclusively to the option quotes, which may incorporate genuine inconsistencies or residual data imperfections such as asynchronicity.   
\label{rem:arb_consistency}
\end{remark}
We formalize this by introducing a fictitious call option with strike $K_0=0$ whose bid and ask prices coincide with the discounted forward price:
\begin{equation}
    C_0^{\mathrm{bid}} = C_0^{\mathrm{ask}} = e^{-rT} F_0^T,
    \label{eq:call-strike-zero}
\end{equation}
where $F_0^T := S_0 e^{(r-q)T}$ denotes the forward price of the underlying asset for maturity $T$.

\subsection{Methodology}
\subsubsection{Optimization Framework for Arbitrage Detection}
\label{subsubsec:ARIES}
This subsection introduces the mathematical framework used to detect arbitrage opportunities in observed option quotes, excluding those related to the term structure. 
We formulate the problem as an optimization program over static portfolios. The admissible portfolios must satisfy bid–ask constraints as well as market depth restrictions. The resulting optimization problem allows us to identify violations of the arbitrage conditions introduced in \Cref{def:arbitrage}.

The theoretical analysis is presented in \Cref{subsec:arbitrage_removal}, while the numerical implementation is discussed in \Cref{sec:numerical_results}.\\

We consider static portfolios formed from the underlying, the quoted call options, and a cash account accruing at the
risk-free rate $r$. For each strike $K_i$, let $q_i^{\mathrm{ask}}$ (respectively $q_i^{\mathrm{bid}}$) denote the number of
contracts bought at the ask (respectively sold at the bid), subject to the depth constraints \eqref{eq:size-constraints-q-market}.
Let $u,c\in\mathbb{R}$ denote the positions in the underlying and the cash account, respectively.

The time--$0$ execution cost is $\pi_0(q^{\mathrm{ask}},q^{\mathrm{bid}},u)$ as defined in \Cref{def:pi0}, and we impose the
zero-cost condition
\begin{equation}
\pi_0(q^{\mathrm{ask}}, q^{\mathrm{bid}}, u) + c = 0.
\label{eq:C0zero}
\end{equation}
The induced net weights are $w_{K_0}=u$ and, for $i=1,\dots,N$,
\begin{equation}
w_{K_i}:=q_i^{\mathrm{ask}}-q_i^{\mathrm{bid}}.
\label{eq:w-from-q-P1}
\end{equation}
For a realized terminal spot level $s\ge 0$, the total terminal payoff (including cash) is
\begin{equation}
\Pi(q^{\mathrm{ask}},q^{\mathrm{bid}},u,c,s)
:= V_T(w)(s) + c e^{rT}
= \sum_{i=1}^{N} (q_i^{\mathrm{ask}} - q_i^{\mathrm{bid}}) (s - K_i)_{+}
    + u s + c e^{rT}.
\label{eq:Pi}
\end{equation}

We therefore consider the following max–min optimization problem:
\begin{equation}
\label{eq:P1}
\begin{aligned}
\text{(P1)}\qquad
&
\max_{\substack{
    q^{\mathrm{ask}},\, q^{\mathrm{bid}},\, u,\, c \\[4pt]
    \pi_0(q^{\mathrm{ask}}, q^{\mathrm{bid}}, u) \,+\, c \;=\; 0 \\[2pt]
    0 \;\le\; q^{\mathrm{bid}} \;\le\; Q^{\mathrm{bid}} \,,\;\; 0 \;\le\; q^{\mathrm{ask}} \;\le\; Q^{\mathrm{ask}}
}}
\min_{s \ge 0}\;\Pi(q^{\mathrm{ask}},q^{\mathrm{bid}},u,c,s).
\end{aligned}
\end{equation}

Although the Problem \eqref{eq:P1} is formulated in execution variables $(q^{\mathrm{ask}},q^{\mathrm{bid}},u,c)$, any feasible
solution induces a static portfolio $w$ through \eqref{eq:w-from-q-P1}. By \Cref{rem:V0-execution}, an optimal solution necessarily satisfies the canonical implementation $q_i^{\mathrm{ask}}=w_{K_i}^{+}$ and $q_i^{\mathrm{bid}}=w_{K_i}^{-}$, so that the budget constraint $\pi_0(q^{\mathrm{ask}},q^{\mathrm{bid}},u)+c=0$ is equivalent to $V_0(w)+c=0$.\\
The value of Problem \eqref{eq:P1} determines whether an arbitrage opportunity exists among admissible portfolios initiated at zero cost. Depending on the payoff $\Pi$ across all possible future states $\mathbf{s}$, two distinct cases arise. If the optimal function $\Pi$ is not identically zero, the strategy constitutes a \emph{strong arbitrage} when its minimum is strictly positive, and a \emph{weak arbitrage} when it is null. Conversely, if $\Pi$ is identically zero, the \emph{absence of any arbitrage opportunity} is guaranteed. \\

In line with the no-arbitrage concept in \Cref{def:arbitrage} \Cref{item:non_arb}, we will eliminate all weak arbitrages upfront via the ARIES procedure detailed in Algorithm~\ref{alg:arb_filter}.

ARIES identifies static arbitrage portfolios supported by available market depth by repeatedly solving the executable arbitrage detection problem \eqref{eq:P3prime}. Each detected arbitrage triggers the removal of the option associated with the smallest saturated bid or ask size, iterating until all arbitrages are eliminated.

This will stabilize the subsequent inverse problem and allow us to extract a smoother, more robust risk-neutral density from market data.

\subsubsection{Hybrid Criterion Framework for Risk-Neutral Density Extraction}
\label{subsub:methodo_rnd}
Following the application of ARIES (Algorithm~\ref{alg:arb_filter}), we assume that the market $\Mcal$ is free of arbitrage opportunities. To extract a risk-neutral density, we formulate our problem on a bounded support $\mathcal{X} := (0,b)$ and enforce bid--ask intervals as inequality constraints. The upper boundary $b$ has to be sufficiently large relative to the observed strike range; its empirical choice is specified in \Cref{eq:upper-bound-support}.

\begin{remark}[On the bounded domain \(\mathcal{X}\)]
\label{rem:bounded-domain}
Beyond the empirical observation that very short-term return distributions are sharply concentrated, the assumption of a bounded domain is consistent with market mechanisms that limit extreme price moves over short horizons.\footnote{A first example is provided by the circuit breakers on the S\&P~500, which are applied on the downside and defined according to three levels. Levels 1 and 2 can be triggered once between 9:30 a.m. and 3:25 p.m.; level 1 corresponds to a 7\% drop relative to the previous close and leads to a trading halt of at least 15 minutes; level 2 corresponds to a 13\% drop, again followed by a minimum 15-minute halt; level 3 is triggered by a 20\% drop and stops trading for the rest of the session.

A second mechanism is the \emph{Limit Up--Limit Down} (LULD) scheme, applied to S\&P~500 constituents. Introduced after the 2010 flash crash, it aims at limiting extreme and abrupt price moves in individual stocks. It relies on dynamic price bands (upper and lower) around a reference price computed as the average of trades over the last five minutes. When the price exits this band (typically \(\pm 5\%\)), a limit state is triggered; if it persists for 15 seconds, a 5-minute trading pause is imposed.

These constraints mechanically cap instantaneous price variations, and therefore support the assumption of a bounded domain \(\mathcal{X}\) for the densities under consideration.}
\end{remark}

\begin{definition}[Admissible densities]
\label{def:admissible-densities}
We define the admissible set 
\begin{equation}
 \mathcal{A}
 :=
 \left\{
  f \in H^{1}(\mathcal{X})
  \;\left|\;
  \begin{aligned}
   & f(x) \ge 0 \text{ a.e. on } \mathcal{X},  \int_{\mathcal{X}} f(x)\,dx = 1, \int_{\mathcal{X}} x f(x)\,dx = F_{0}^{T}, \\
   & e^{-rT} \int_{\mathcal{X}} (x - K_{i})_{+} f(x)\,dx
     \in [C_{i}^{\mathrm{bid}}, C_{i}^{\mathrm{ask}}],
     \quad i = 1,\dots,N
  \end{aligned}\right.
 \right\}.
 \label{eq:A-admissible-set}
\end{equation}
In other terms, admissible functions are \(H^{1}(\mathcal{X})\)\footnote{We recall that  $H^{1}(\mathcal{X})$ is the Sobolev space
$H^{1}(\mathcal{X}) = \{f \in L^{2}(\mathcal{X}) : f' \in L^{2}(\mathcal{X})\}$, 
where $f'$ is understood in the weak sense.} probability densities
whose first moment matches the forward price and whose associated call prices
remain within the market bid--ask intervals.
\end{definition}
We fix $\lambda_1>0$ and $\lambda_2>0$ and introduce the continuous
hybrid $L^2$--entropy objective on the admissible set $\mathcal A$.

\begin{definition}[Continuous objective function]
Let $H:\mathcal A\to\R$ denote the functional
\begin{equation}
  H(f)
  := \lambda_1 \lVert f'\rVert_{L^2(\mathcal{X})}^2
     + \lambda_2 \int_\mathcal{X} f(x)\ln f(x)\,dx
  = \lambda_1 \lVert f'\rVert_{L^2(\mathcal{X})}^2
    - \lambda_2 S(f),
  \label{eq:def-H}
\end{equation}
where $S(f) := -\int_\mathcal{X} f(x)\ln f(x)\,dx$ denotes the (differential)  Shannon entropy
of $f$, with the usual convention $0\ln 0 := 0$.
\end{definition}

While $H(f)$ provides a theoretical characterization of the risk-neutral density
(see Appendix~\ref{subsec:cont_framework}), numerical calibration requires a discrete
formulation. We present the numerical implementation and establish existence and
uniqueness of the resulting discrete solution in
Section~\ref{subsec:exist-unique-disc}. 

\begin{definition}[Grid and decision variable]
\label{def:grid-decision-variable}
Let $M\in\mathbb N$ and set $\Delta s := b/M$. We consider a uniform grid of $M$ points on $\mathcal{X}$ compatible with the observed strikes, in the sense that each $K_j$, $j=1,\dots,N$, belongs to the grid. The grid points are defined by
\begin{equation}
  s_i := i\Delta s = b\,\frac{i}{M}, \qquad i=1,\ldots,M.
  \label{eq:grid_def}
\end{equation}
The decision variable is $p := (p_1,\ldots,p_M)^\top \in \R^M$.
\end{definition}

\begin{definition}[Discrete admissible set]
\label{def:discrete-admissible-set}
We define the discrete admissible set $\mathcal A_M \subset \R^M$ as the set of
vectors $p$ satisfying:
\begin{enumerate}
  \item \emph{Simplex constraint:}
  \begin{equation}
    \Sigma_M := \Bigl\{p\in\R^M:\ p_i\ge 0,\ \sum_{i=1}^M p_i = 1\Bigr\},
    \label{eq:simplex}
  \end{equation}
  \item \emph{Forward constraint:}
  \begin{equation}
    \sum_{i=1}^M s_i p_i = F_0^T,
    \label{eq:discrete-forward}
  \end{equation}
  \item \emph{Call price constraints:} for each $j=1,\ldots,N$,
  \begin{equation}
    e^{-rT}\sum_{i=1}^M (s_i-K_j)_+\,p_i \in
    \bigl[C_j^{\mathrm{bid}},\,C_j^{\mathrm{ask}}\bigr].
    \label{eq:discrete-call}
  \end{equation}
\end{enumerate}
\end{definition}

\noindent We define the first-order difference matrix $D^{(1)} \in \R^{M-1 \times M}$ as the linear operator mapping any vector $v = (v_1, \dots, v_{M})^\top \in \R^{M}$ to its forward differences:
\begin{equation}
    (D^{(1)}v)_i := v_{i+1} - v_i, \qquad i = 1, \dots, M-1.
    \label{eq:D1}
\end{equation}
\begin{definition}[Discrete objective function]
\label{def:discrete-objective}
 The discrete entropy is:
\begin{equation}
  S_M(p) := -\sum_{i=1}^M p_i \ln p_i,
  \label{eq:discrete-entropy}
\end{equation}
with the continuous extension at $p_i=0$.
We next define the discrete functional
$H_M:(\R_+)^M\to\R$ by
\begin{equation}
  H_M(p)
  := \frac{\lambda_1}{(\Delta s)^3}\|D^{(1)}p\|_2^2 - \lambda_2 S_M(p)
  = \frac{\lambda_1}{(\Delta s)^3}\sum_{i=1}^{M-1}(p_{i+1}-p_i)^2
     + \lambda_2\sum_{i=1}^M p_i\ln p_i.
  \label{eq:def-HM}
\end{equation}
\end{definition}

In summary, the SEDEx procedure amounts to minimizing $H_M(\cdot)$ over $\mathcal A_M$ introduced in \Cref{def:discrete-admissible-set}. The next section provides the theoretical guarantees for this procedure.

\section{Theoretical Results}
\label{sec:theoretical_results}
\subsection{Arbitrage Filtering Procedure}
\label{subsec:arbitrage_removal}

We now derive a tractable formulation of Problem \eqref{eq:P1} and study the properties of its solution.
To this end, we introduce a slack variable $\varepsilon'$ representing the minimal payoff across states, which can be interpreted as an arbitrage margin. This allows us to reformulate the max–min problem \eqref{eq:P1} as follows:
\begin{equation}
\label{eq:P2}
\begin{aligned}
\text{(P2)}\qquad
&
\max_{\substack{
    q^{\mathrm{ask}},\, q^{\mathrm{bid}},\, u,\, c,\, \varepsilon' \\[4pt]
    \pi_0(q^{\mathrm{ask}}, q^{\mathrm{bid}}, u) \,+\, c \;=\; 0 \\[2pt]
    \Pi(q^{\mathrm{ask}},q^{\mathrm{bid}},u,c,s) \;\ge\; \varepsilon' \,,\;\; \forall\, s \;\ge\; 0 \\[2pt]
    0 \;\le\; q^{\mathrm{bid}} \;\le\; Q^{\mathrm{bid}} \,,\;\; 0 \;\le\; q^{\mathrm{ask}} \;\le\; Q^{\mathrm{ask}}
}}
\quad \varepsilon'.
\end{aligned}
\end{equation}

Substituting \eqref{eq:C0zero} into \eqref{eq:Pi} and renaming $\varepsilon'= \varepsilon e^{rT} $, we obtain the equivalent payoff constraints
\begin{equation}
\sum_{i=1}^{N} (q_i^{\text{ask}} - q_i^{\text{bid}})(s - K_i)_{+}
    + u s - \big( \pi_0(q^{\text{ask}}, q^{\text{bid}}, u) + \varepsilon \big) e^{rT} \ge 0, 
    \quad \forall s \ge 0.
\label{eq:P2constraint}
\end{equation}

\noindent
By aggregating the initial cost and the arbitrage margin into a single term denoted $\alpha := \pi_0(q^{\text{ask}}, q^{\text{bid}}, u) + \varepsilon$ we can rewrite the problem as follows:
\begin{equation}
\label{eq:P3}
\begin{aligned}
& \text{(P3)} \quad \max_{q^{\mathrm{ask}},\, q^{\mathrm{bid}},\, u,\, \alpha} \quad 
\sum_{i=1}^{N} \big( q_i^{\mathrm{bid}} C_i^{\mathrm{bid}} - q_i^{\mathrm{ask}} C_i^{\mathrm{ask}} \big) \,-\, u S_0 e^{-qT} \,+\, \alpha \\
& \mathrm{s.t.} \quad \sum_{i=1}^{N} (q_i^{\mathrm{ask}} - q_i^{\mathrm{bid}})(s - K_i)_{+} \,+\, u s \,-\, \alpha e^{rT} \;\ge\; 0 \,,\;\; \forall\, s \;\ge\; 0 \\
& \phantom{\mathrm{s.t.} \quad} 0 \;\le\; q^{\mathrm{bid}} \;\le\; Q^{\mathrm{bid}} \,,\;\; 0 \;\le\; q^{\mathrm{ask}} \;\le\; Q^{\mathrm{ask}}.
\end{aligned}
\end{equation}

\vspace{12pt}
\noindent
For notational convenience, we collect the objective and payoff expressions into the following reduced-form quantities:
\begin{align}
\tilde{\pi}_0(q^{\mathrm{ask}}, q^{\mathrm{bid}}, u, \alpha)
&:= \sum_{i=1}^{N} \big(q_i^{\mathrm{ask}} C_i^{\mathrm{ask}} - q_i^{\mathrm{bid}} C_i^{\mathrm{bid}}\big)
+ u S_0 e^{-qT} - \alpha, \label{eq:C0prime}\\[4pt]
\tilde{\Pi}(q^{\mathrm{ask}}, q^{\mathrm{bid}}, u, \alpha, s)
&:= \sum_{i=1}^{N} (q_i^{\mathrm{ask}} - q_i^{\mathrm{bid}})(s - K_i)_{+}
+ u s - \alpha e^{rT}. \label{eq:Piprime}
\end{align}

Then \eqref{eq:P3} is equivalent to:
\begin{equation}
\label{eq:P3prime}
\begin{aligned}
\text{(P3$^{\prime}$)}\qquad
&
\max_{\substack{
    q^{\mathrm{ask}},\, q^{\mathrm{bid}},\, u,\, \alpha \\[4pt]
    \tilde{\Pi}(q^{\mathrm{ask}}, q^{\mathrm{bid}}, u, \alpha, s) \;\ge\; 0 \,,\;\; \forall\, s \;\ge\; 0 \\[2pt]
    0 \;\le\; q^{\mathrm{bid}} \;\le\; Q^{\mathrm{bid}} \,,\;\; 0 \;\le\; q^{\mathrm{ask}} \;\le\; Q^{\mathrm{ask}}
}}
\quad
-\,\tilde{\pi}_0(q^{\mathrm{ask}}, q^{\mathrm{bid}}, u, \alpha).
\end{aligned}
\end{equation}
\noindent
Next, we reduce the problem to a finite set of constraints, yielding a tractable linear program. Since $\tilde{\Pi}$ is piecewise linear in $s$, its extrema occur at breakpoints and boundaries, thus $\tilde{\Pi}(...,s) \geq 0$ for all $s \geq 0$ is equivalent to:
\begin{equation}
\tilde{\Pi}|_{s=0}\ge 0, \qquad 
\tilde{\Pi}|_{s=K_i}\ge 0 \ (i=1,\ldots,N), \qquad 
\lim_{s\to\infty}\frac{\ \tilde{\Pi}(...,s)}{s}\ge 0.
\end{equation}

\noindent
Consequently, the constraints in \eqref{eq:P3prime} can be rewritten as:
\begin{equation}
\begin{aligned}
&\sum_{i=1}^{j} (q_i^{\text{ask}} - q_i^{\text{bid}})(K_j - K_i)
    + u K_j - \alpha e^{rT} \ge 0, && j = 1, \ldots, N,\\
&\sum_{i=1}^{N} (q_i^{\text{ask}} - q_i^{\text{bid}}) + u \ge 0,\\
&\alpha \le 0, \quad 0 \le q^{\text{ask}} \le Q^{\text{ask}}, \quad 0 \le q^{\text{bid}} \le Q^{\text{bid}}.
\end{aligned}
\label{eq:constraints}
\end{equation}
Consequently, Problem \eqref{eq:P3prime} is a finite-dimensional linear program over a polyhedral feasible set, denoted by $\mathcal{C}$ and characterized by the constraints in \eqref{eq:constraints}. 
\begin{remark}[Homogeneity]
\label{rem:homogeneity}
The objective function and all constraints in \eqref{eq:P3prime} are positively homogeneous of degree~$1$ in 
$(q^{\text{ask}}, q^{\text{bid}}, u, \alpha)$.  
\end{remark}
The following proposition shows that this linear program is well-posed.
\begin{proposition}[Existence of an Optimal Solution]
\label{prop:existence-optimum}
Problem~\eqref{eq:P3prime} admits at least one optimal solution.
Moreover, an optimal solution can be chosen at a vertex of $\mathcal{C}$.
\end{proposition}

\begin{proof}
See Appendix~\ref{appendix:exist_opt_arb}.
\end{proof}

Having established existence, we now provide a complete characterization of optimal portfolios and their arbitrage content. Specifically, exactly one of three mutually exclusive cases holds, corresponding to strong, weak, or no arbitrage.
\begin{proposition}[Characterization of Optimal Solutions and Link to Arbitrage]
\label{prop:characterization_proof_version}
Let $x=(q^{\mathrm{ask}},q^{\mathrm{bid}},u,\alpha)$ and consider the linear program \eqref{eq:P3prime}
over the feasible polytope $\mathcal{C}$ defined in \eqref{eq:constraints}.
Let $x^*=(q^{\mathrm{ask}*},q^{\mathrm{bid}*},u^*,\alpha^*)$ be an optimal solution, and set
\[
\tilde{\Pi}^*(s):=\tilde{\Pi}(q^{\mathrm{ask}*},q^{\mathrm{bid}*},u^*,\alpha^*,s),\qquad
\tilde{\pi}_0^*:=\tilde{\pi}_0(q^{\mathrm{ask}*},q^{\mathrm{bid}*},u^*,\alpha^*).
\]
Exactly one of the following cases occurs:
\begin{enumerate}
\item \textnormal{\emph{Strong-arbitrage case.}}
If $\tilde{\pi}_0^*<0$, then the optimizer must saturate at least one
depth constraint, i.e., there exists $i\in\{1,\dots,N\}$ such that
\[
q_i^{\mathrm{ask}*}=Q_i^{\mathrm{ask}}
\quad\text{or}\quad
q_i^{\mathrm{bid}*}=Q_i^{\mathrm{bid}}.
\]
Moreover, the associated implementation $(w,c)$ constructed from $x^*$ via
\eqref{eq:C0zero} and \eqref{eq:w-from-q-P1} yields a strong arbitrage.

\item \textnormal{\emph{Weak-arbitrage case.}}
If $\tilde{\pi}_0^*=0$ and $\tilde{\Pi}^*$ is not identically zero, then the set of optimal solutions contains a
non-trivial segment along the ray generated by $x^*$:
\[
\{\lambda x^*:\ 0\le \lambda \le \bar\lambda(x^*)\}\subseteq \arg\max\ (P3'),
\]
where $\bar\lambda(x^*)$ is the maximal scaling factor compatible with the upper size bounds, i.e.
\begin{equation}
    \bar{\lambda}(x^*) =
\min\!\left\{
\min_{i:\,q_i^{\mathrm{ask}*}>0}\frac{Q_i^{\mathrm{ask}}}{q_i^{\mathrm{ask}*}},
\;
\min_{i:\,q_i^{\mathrm{bid}*}>0}\frac{Q_i^{\mathrm{bid}}}{q_i^{\mathrm{bid}*}}
\right\}.
\label{eq:lambda_bar}
\end{equation}
The extremal point $\bar\lambda(x^*)x^*$ lies on the boundary of $\mathcal{C}$ and therefore
saturates at least one depth constraint.\\
In addition, the associated implementation $(w,c)$ constructed from $x^*$ via
\eqref{eq:C0zero} and \eqref{eq:w-from-q-P1} yields a weak arbitrage.

\item \textnormal{\emph{No-arbitrage case.}}
The market is arbitrage-free if and only if
\[
\tilde{\Pi}^*(s)= 0 \ \text{for all } s\ge 0.
\]
In this case, the unique optimal solution is the null portfolio $x^*=0$.
\end{enumerate}
\end{proposition}
\begin{proof}
See Appendix~\ref{appendix:charac_opt_arb}
\end{proof}

\subsection{Well-Posedness of Risk-Neutral Density Extraction}
This subsection establishes that the proposed extraction procedure (SEDEx) is well-posed under no-arbitrage. The argument proceeds in three steps. First, we characterize admissible call price vectors consistent with market quotes and static arbitrage restrictions. Second, we show that any such vector admits a discrete risk-neutral representation on an augmented strike grid. Third, we show that the proposed density extraction problem admits a unique solution.

\label{subsec: hybrid_method}

\subsubsection{Admissible Call Price Vectors}

We assume throughout the remainder of \Cref{subsec: hybrid_method} that ARIES has been applied and that the market $\Mcal$ is free of arbitrage opportunities. We first introduce the notion of an admissible call price vector.
\begin{definition}[Admissible call price vector]
\label{def:admissible-call-vector}
A vector \(C := (C_{0}, \dots, C_{N}) \in \R_+^{N+1}\) is said to be an \emph{admissible call price vector} if it satisfies the following conditions:
\begin{enumerate}
 \item \textbf{Positivity.} For all \(i \in \{0,\dots,N\}\),
 \begin{equation}
   C_{i} > 0.
   \label{eq:positivity_adm_call}
 \end{equation}

 \item \textbf{Monotonicity.} For all \(i \in \{0,\dots,N-1\}\),
 \begin{equation}
   C_{i} > C_{i+1}.
   \label{eq:monotonicity}
 \end{equation}

 \item \textbf{Convexity.} For all \(i \in \{0,\dots,N-2\}\),
 \begin{equation}
   \frac{C_{i} - C_{i+1}}{K_{i+1} - K_{i}}
   \;>\;
   \frac{C_{i+1} - C_{i+2}}{K_{i+2} - K_{i+1}}.
   \label{eq:discrete-convexity-local}
 \end{equation}

 \item \textbf{Slope control on \([K_{0},K_{1}]\) (Lower Bound on Call Prices).} We have
 \begin{equation}
   C_{0} - C_{1} < K_{1} e^{-rT},
   \qquad \text{equivalently} \qquad
   \frac{C_{1} - C_{0}}{K_{1}} > - e^{-rT}.
   \label{eq:slope-K0-K1}
 \end{equation}

 \item \textbf{Bid--ask bounds.} For all \(i \in \{0,\dots,N\}\),
 \begin{equation}
   C_{i}^{\mathrm{bid}} \le C_{i} \le C_{i}^{\mathrm{ask}},
   \qquad
   C_{0}^{\mathrm{ask}} = C_{0}^{\mathrm{bid}} = S_{0} e^{-qT}.
   \label{eq:bid-ask-bounds}
 \end{equation}
\end{enumerate}
\end{definition}
The aforementioned local properties are, in fact, global; we provide the formal proof in Appendix~\ref{appendix:global_prop}. Conditions \eqref{eq:positivity_adm_call}--\eqref{eq:slope-K0-K1} correspond to the standard static arbitrage restrictions for call prices, while condition \eqref{eq:bid-ask-bounds} enforces consistency with executable market quotes.

\begin{remark}
\label{rem:vs-bf-rewrite}
With the notation of \Cref{def:vs-bf}, suppose that call option prices
$C_i$ are observed on the market at time $0$ for strikes $K_i$.
Then the initial value of the vertical spread supported on $(K_i,K_{i+1})$
satisfies
$
V_0(\mathrm{VS}_{i,i+1}) = C_i - C_{i+1}.
$
In this case, the strict monotonicity condition \eqref{eq:monotonicity}
is equivalent to the strict positivity of adjacent vertical-spread prices, namely
\begin{align}
\forall\, i \in \{0,\dots,N-1\}, \qquad
C_i > C_{i+1}
\quad \Longleftrightarrow \quad
V_0(\mathrm{VS}_{i,i+1}) > 0 .
\end{align}

Likewise, the strict convexity condition \eqref{eq:discrete-convexity-local}
is equivalent to the strict positivity of adjacent butterfly-spread prices:
\begin{align}
\forall\, i \in \{0,\dots,N-2\}:\ 
\frac{C_i - C_{i+1}}{K_{i+1} - K_i}
-
\frac{C_{i+1} - C_{i+2}}{K_{i+2} - K_{i+1}}
> 0
 \Longleftrightarrow 
V_0(\mathrm{BF}_{i,i+1,i+2}) > 0 .
\end{align}
\end{remark}

The remaining issue is feasibility. We will use a corollary of Motzkin’s theorem of the alternative (\Cref{thm:motzkin} in the \Cref{appendix:motzkin}; see also\ \cite{motzkin1936beitrage, abhyankar2002encyclopaedia}), which provides the basis for establishing the existence of an admissible call-price vector in 
\Cref{prop:existence-admissible-C}.

\begin{corollary}[Motzkin with strict positivity]
\label{cor:motzkin-positive}
Let $A\in\R^{m\times n}$, $B\in\R^{\ell\times n}$, $b\in\R^m$, and $c\in\R^\ell$.
Exactly one of the following two assertions holds:
\begin{enumerate}
\item[1.] The system $Ax\le b$, $Bx<c$ admits a solution $x>0$.
\item[2.] There exist $y\in\R^m_+$ and $z\in\R^\ell_+$ such that,
\[
A^\top y + B^\top z \ge 0 \ \text{and }\ b^\top y + c^\top z < 0,
\]
or
\[
A^\top y + B^\top z \ge 0,\ z\neq 0 \ \text{and } b^\top y + c^\top z \leq 0,
\]
or
\[
A^\top y + B^\top z \ge 0,\ A^\top y + B^\top z \neq 0 \ \text{and } b^\top y + c^\top z \leq 0.
\]
\end{enumerate}
\end{corollary}

\begin{proposition}[Existence of an admissible call-price vector]
\label{prop:existence-admissible-C}
Under the no-arbitrage assumption (\Cref{def:arbitrage}), there exists an admissible call price vector
\begin{equation}
 C = (C_{0}, C_{1}, \dots, C_{N}) \in \R^{N+1}
\end{equation}
satisfying all constraints of \Cref{def:admissible-call-vector}.
\end{proposition}

\begin{proof}
The detailed proof is deferred to Appendix~\ref{appendix:existence_call_price}; only its main outline is presented below.
\paragraph{Outline and intuition}
We encode the admissibility constraints in \Cref{def:admissible-call-vector} as a system of mixed strict and weak linear inequalities
\[
A_{\mathrm{mkt}}\,C \le b_{\mathrm{mkt}},
\qquad
A_{\mathrm{struct}}\,C < b_{\mathrm{struct}},
\qquad
C>0,
\]
where $A_{\mathrm{mkt}}$ collects the bid--ask bounds, and $A_{\mathrm{struct}}$ collects the shape constraints
(monotonicity, convexity, and slope constraint). Thus the goal is to find $C$ that satisfies these inequalities.
Next, we argue by contradiction. If no solution exists, that is, the admissible-price constraints are infeasible, Motzkin’s alternative (\Cref{cor:motzkin-positive}) provides a \emph{dual certificate} of infeasibility. Financially, this certificate can be read as an explicit static strategy constructed from two ingredients:
(i) executable trades at the quoted bid--ask prices, and
(ii) nonnegative combinations of elementary strategies.
By construction, this strategy has a terminal payoff that is nonnegative.
We then inspect its inception cost. Either the strategy produces a strictly negative initial cost (cash received upfront), which yields a \emph{strong arbitrage}; or it is costless and is strictly profitable in some future state, which yields a \emph{weak arbitrage} under a suitably chosen probability measure.
Both possibilities contradict the no-arbitrage assumption; hence the admissible-price system must be feasible.
\end{proof}

\subsubsection{Discrete Risk-Neutral Measure Construction}
\label{subsub:discrete-rn-construction}
Given an admissible call price vector $C$ (Definition~\ref{def:admissible-call-vector}), 
we construct a discrete risk-neutral measure that reproduces $C$ and is supported in 
the state space $\mathcal X=(0,b)$. The construction proceeds in three steps. We first 
augment the strike grid with a fictitious terminal strike $K_{N+1}$ and set 
$C_{N+1}=0$. We then derive a sequence of weights on the augmented grid. The natural 
allocation places a Dirac mass at $K_0=0\notin\mathcal X$; we therefore substitute 
$K_0$ by an interior strike $K'\in(K_0,K_1)$ and redistribute the leftmost mass on 
the support $\mathcal K':=\{K',K_1,\dots,K_{N+1}\}$. The admissible range for $K'$ is 
derived below.

\paragraph{Tail extrapolation.}
Let $\xi>1$ be chosen so that $K_{N+1}<b$ and the terminal strike belongs to the compatible grid of \Cref{def:grid-decision-variable}: 
\begin{equation}
K_{N+1}
:=
K_{N}
+
\xi\,
\frac{C_{N}}{(C_{N-1} - C_{N})/(K_{N} - K_{N-1})}
\, .
\label{eq:KNplus1-def}
\end{equation}
We set
\begin{equation}
C_{N+1} := 0.
\label{eq:CNplus1-def}
\end{equation}
Choosing $\xi>1$ guarantees that the call-price vector extended with 
$C_{N+1}$ remains strictly decreasing and convex.

\paragraph{Risk-neutral weights on $\mathcal K'$.}
Define the auxiliary quantities, for $i\in\{0,\dots,N\}$,
\begin{equation}
Q_i
:=
e^{rT}\,
\frac{C_i-C_{i+1}}{K_{i+1}-K_i}.
\label{eq:Qi-def}
\end{equation}
Under Definition~\ref{def:admissible-call-vector} together with \eqref{eq:CNplus1-def}, 
the sequence $(Q_i)_{i=0}^N$ is strictly decreasing and takes values in $(0,1)$. We define the weights 
at $K_2,\dots,K_{N+1}$ by
\begin{equation}
q'_j := Q_{j-1}-Q_j,
\qquad j=2,\dots,N,
\qquad
q'_{N+1}:=Q_N,
\label{eq:qprime-upper}
\end{equation}
which satisfy $\sum_{j=2}^{N+1}q'_j = Q_1$. The mass remaining for the two leftmost 
weights is therefore
\begin{equation}
m := 1-Q_1.
\label{eq:m-mass}
\end{equation}
Let
\begin{equation}
\bar K
:=
\frac{(1-Q_0)K_0 + (Q_0-Q_1)K_1}{1-Q_1}.
\label{eq:Kbary}
\end{equation}
For any given $K'\in(K_0,K_1)$, the two remaining weights are defined by
\begin{equation}
q'_0
:=
m\,\frac{K_1-\bar K}{K_1-K'},
\qquad
q'_1
:=
m\,\frac{\bar K-K'}{K_1-K'}.
\label{eq:qprime-lower}
\end{equation}

\begin{definition}[Discrete measure on $\mathcal K'$]
\label{def:discrete-rn-measure}
Let $K'\in(K_0,K_1)$ be a grid point and define \(q'=(q'_0,\dots,q'_{N+1})^\top\) by 
\eqref{eq:qprime-upper}--\eqref{eq:qprime-lower}. The associated discrete measure on 
$\mathcal K'$ is
\begin{equation}
\nu
:=
q'_0\,\delta_{K'}
+
\sum_{j=1}^{N+1} q'_j\,\delta_{K_j},
\label{eq:nu}
\end{equation}
where $\delta_K$ denotes the Dirac mass at $K$.
\end{definition}

\begin{remark}[Canonical interpretation]
\label{rem:canonical-interpretation}
Consider the discrete probability measure on the augmented grid 
$\{K_0,\dots,K_{N+1}\}$ defined by
\begin{equation}
\nu_0
:=
(1-Q_0)\,\delta_{K_0}
+
\sum_{j=1}^{N}(Q_{j-1}-Q_j)\,\delta_{K_j}
+
Q_N\,\delta_{K_{N+1}}.
\label{eq:nu0}
\end{equation}
The sequence $(Q_i)_{i=0}^N$ coincides with the survival function of $\nu_0$, namely 
$Q_i=\nu_0(S_T>K_i)$. Equivalently, $\nu_0$ arises as the distributional second 
derivative of the piecewise-linear interpolation of the extended call price vector 
on $[K_0,K_{N+1}]$, in the sense of \cite{BreedenLitzenberger1978}. The mass \eqref{eq:m-mass} and barycenter 
\eqref{eq:Kbary} then admit the representation
\begin{equation}
\nu_0(\{K_0,K_1\})=m,
\qquad
\bar K
=
\E^{\nu_0}\!\left[\,S_T\,\middle|\,S_T\in\{K_0,K_1\}\,\right].
\end{equation}
Moreover, $(q'_0,q'_1)$ in \eqref{eq:qprime-lower} is the unique solution of
\begin{equation}
q'_0+q'_1 = m,
\qquad
q'_0\,K' + q'_1\,K_1 = m\bar K,
\label{eq:left-tail-system}
\end{equation}
and therefore preserves the mass and the first moment on $\{K_0,K_1\}$. Accordingly, $\nu$ is referred to as the 
\emph{extended Breeden--Litzenberger measure} on $\mathcal K'$.
\end{remark}

\begin{proposition}[Replication and uniqueness on $\mathcal K'$]
\label{prop:boundary-shift}
Let $\nu$ be the extended Breeden--Litzenberger measure of 
Definition~\ref{def:discrete-rn-measure}. Then $\nu$ reproduces the admissible call 
price vector $C$, namely
\begin{equation}
e^{-rT}\E^\nu[(S_T-K_i)_+] \;=\; C_i,
\qquad i=0,\dots,N.
\label{eq:nu-matches-C}
\end{equation}
Moreover, $\nu$ is a probability measure on $\mathcal K'$ if and only if $K'$ 
satisfies the barycentric condition
\begin{equation}
K' < \bar K.
\label{eq:barycentric-condition}
\end{equation}
For any fixed support $\mathcal K'$, the probability measure satisfying 
\eqref{eq:nu-matches-C} is unique.
\end{proposition}
\begin{proof}
See Appendix~\ref{appendix:proof_boundary_shift}.
\end{proof}

This result establishes that the admissibility of a call price vector is equivalent 
to the existence of a discrete arbitrage-free pricing measure. With the above grid-compatible choice of auxiliary strikes, this measure induces a feasible element within the admissible optimization set $\mathcal A_M$.

\subsubsection{Existence and Uniqueness of the Minimizer}
\label{subsec:exist-unique-disc}
We conclude this section by establishing well-posedness of SEDEx, namely the minimization of $H_M$ over the admissible set $\mathcal{A}_M$ introduced in Section~\ref{subsub:methodo_rnd}. In particular, we prove existence of a minimizer and its uniqueness.
\begin{proposition}
\label{prop:AM-compact-convex}
For a compatible grid as in Definition~\ref{def:grid-decision-variable}, the set $\mathcal{A}_M$ is a non-empty compact convex subset of $\R^M$.
\end{proposition}

\begin{proof}
The constraint \eqref{eq:simplex} defines a simplex, hence a polytope in $\R^M$. The forward constraint \eqref{eq:discrete-forward}
defines an affine hyperplane, and each of the call price constraints
\eqref{eq:discrete-call} defines the intersection of two closed affine
half-spaces. Therefore $\mathcal{A}_M$ is the finite intersection of closed convex
sets, hence is itself closed and convex. Boundedness follows from the
simplex constraint. \\
The non-emptiness of $\mathcal{A}_M$ is guaranteed by the
construction of the discrete measure $\nu$ \eqref{eq:nu} under the no-arbitrage
assumption, with the auxiliary strikes chosen on the compatible grid. Hence $\mathcal{A}_M$ is a non-empty compact convex subset of $\R^M$.
\end{proof}

\begin{proposition}
\label{prop:HM-strict-convex}
The functional $H_M$ is continuous and strictly convex on $\mathcal{A}_M$.
\end{proposition}

\begin{proof}
The continuity is immediate from the expression \eqref{eq:def-HM}, so we focus on strict convexity.

\medskip
\noindent
\emph{Quadratic part.}
We have
\begin{equation}
  \lVert D^{(1)}p\rVert_2^2
  = (D^{(1)}p)^\top(D^{(1)}p)
  = p^\top (D^{(1)})^\top D^{(1)} p
  = p^\top Q p,
  \label{eq:Q-matrix}
\end{equation}
where $Q := (D^{(1)})^\top D^{(1)}\in\R^{M\times M}$ is a symmetric
tridiagonal matrix. Moreover, $Q$ is positive semidefinite and
its kernel is given by the one-dimensional subspace spanned by the vector
$\mathbbm{1} = (1,\ldots,1)^\top$:
\begin{equation}
  \ker(Q) = \mathrm{Vect}\{\mathbbm{1}\}.
  \label{eq:kernel-Q}
\end{equation}
Thus the quadratic form $p\mapsto p^\top Q p = \lVert D^{(1)}p\rVert_2^2$
is convex on $\R^M$, but not strictly convex on all of $\R^M$.

On the simplex $\Sigma_M$, however, strict convexity is recovered. Let
$p,q\in \Sigma_M$ with $p\neq q$ and $\theta\in(0,1)$. By adapting \eqref{eq:scalar-strict-convex} to the discrete
setting,
\begin{equation}
  \lVert D^{(1)}(\theta p + (1-\theta)q)\rVert_2^2
  = \theta \lVert D^{(1)}p\rVert_2^2
    + (1-\theta)\lVert D^{(1)}q\rVert_2^2
    - \theta(1-\theta)\lVert D^{(1)}(p-q)\rVert_2^2.
  \label{eq:discrete-quadratic-identity}
\end{equation}
The last term vanishes if and only if $D^{(1)}(p-q)=0$, i.e.\ $p-q$ is
constant. Since $p$ and $q$ both have total mass $1$, this is only
possible if $p=q$. Therefore, for $p\neq q$ in $\Sigma_M$,
\begin{equation}
  \lVert D^{(1)}(\theta p + (1-\theta)q)\rVert_2^2
  < \theta \lVert D^{(1)}p\rVert_2^2
    + (1-\theta)\lVert D^{(1)}q\rVert_2^2,
  \label{eq:discrete-quadratic-strict}
\end{equation}
which shows that $p\mapsto \lVert D^{(1)}p\rVert_2^2$ is strictly convex
on $\Sigma_M$, and hence on $\mathcal{A}_M\subset \Sigma_M$.

\medskip
\noindent
\emph{Entropy part.}
On the strictly positive orthant $(\mathbb{R}_{+})^M$, the discrete entropy
$S_M$ is twice continuously differentiable with Hessian
\begin{equation}
  \nabla^2 (-S_M)(p) = \mathrm{Diag}\left( \frac{1}{p_i} \right)_{i=1}^M,
  \label{eq:discrete-entropy-hessian}
\end{equation}
which is positive definite. Thus $-S_M$ is strictly convex on
$(\mathbb{R}_{+})^M$.

To handle the boundary points where some $p_i$ may vanish, consider
$p\neq q$ in $(\R_+)^M$ and $\theta\in (0,1)$. Define the index sets
\begin{equation}
  I := \{ i : p_i>0 \text{ and } q_i>0\}, \quad
  P_0 := \{ i : p_i=0,\ q_i>0\}, \quad
  Q_0 := \{ i : p_i>0,\ q_i=0\}.
  \label{eq:index-sets}
\end{equation}
Then
\begin{align}
  -S_M(\theta p + (1-\theta)q)
  &= \sum_{i=1}^M
       \bigl(\theta p_i + (1-\theta)q_i\bigr)
       \ln\bigl(\theta p_i + (1-\theta)q_i\bigr)
     \nonumber\\
  &= \sum_{i\in I}
       \bigl(\theta p_i + (1-\theta)q_i\bigr)
       \ln\bigl(\theta p_i + (1-\theta)q_i\bigr)
     \nonumber\\
  &\quad + \sum_{i\in Q_0} \theta p_i \ln(\theta p_i)
       + \sum_{i\in P_0} (1-\theta) q_i \ln\bigl((1-\theta) q_i\bigr).
  \label{eq:entropy-split}
\end{align}
Using the strict convexity in $I$, we
obtain
\begin{equation}
  \sum_{i\in I}
    \bigl(\theta p_i + (1-\theta)q_i\bigr)
    \ln\bigl(\theta p_i + (1-\theta)q_i\bigr)
  < \theta \sum_{i\in I} p_i\ln p_i
    + (1-\theta)\sum_{i\in I} q_i\ln q_i,
  \label{eq:entropy-I}
\end{equation}
whenever $I\neq\emptyset$. In addition, expanding
the last two terms in \eqref{eq:entropy-split} yields
\begin{align}
  \sum_{i\in Q_0} \theta p_i \ln(\theta p_i)
  &= \theta \sum_{i\in Q_0} p_i\ln p_i
     + \theta \ln\theta \sum_{i\in Q_0} p_i,
     \nonumber\\
  \sum_{i\in P_0} (1-\theta) q_i \ln\bigl((1-\theta) q_i\bigr)
  &= (1-\theta) \sum_{i\in P_0} q_i\ln q_i
     + (1-\theta)\ln(1-\theta) \sum_{i\in P_0} q_i.
  \label{eq:entropy-PQ}
\end{align}
If $I = \emptyset$ then either $P_0$ or $Q_0$ is non-empty. Since $\theta\in(0,1)$, any strictly positive mass in these sets yields a strictly negative term.\\
Hence $-S_M$ is strictly convex on $(\R_+)^M$, and \emph{a fortiori} on $\mathcal{A}_M$.

\medskip
Combining the strict convexity of the quadratic part on $\mathcal{A}_M$ with the
strict convexity of the entropy part on $\mathcal{A}_M$, we conclude that $H_M$ is
strictly convex on $\mathcal{A}_M$.
\end{proof}

\begin{theorem}
\label{thm:exist-unique-discrete-minimizer}
For any compatible grid $(s_i)_{i=1}^M$ in Definition~\ref{def:grid-decision-variable}, the discrete functional $H_M$ admits a unique minimizer $p^\ast$ over $\mathcal A_M$.
\end{theorem}

\begin{proof}
By Proposition~\ref{prop:AM-compact-convex}, the set $\mathcal{A}_M$ is a non-empty
compact convex subset of $\R^M$. The functional $H_M$ is continuous on
$(\R_+)^M$ and strictly convex on $\mathcal{A}_M$
(Proposition~\ref{prop:HM-strict-convex}). By the
Weierstrass theorem, any (lower semi-) continuous function on a compact set attains its
minimum; thus $H_M$ attains its minimum on $\mathcal{A}_M$. The strict convexity of
$H_M$ on $\mathcal{A}_M$ ensures that this minimizer is unique.
\end{proof}
\section{Numerical Results}
\label{sec:numerical_results}
\subsection{Data Cleaning and Preprocessing}

\paragraph{Data Source and Coverage.}

Our empirical analysis is based on intraday index and option data on the S\&P~500 from January~2012 to July~2023, obtained from the Cboe DataShop.

For the underlying index, the dataset provides bid and ask quotes sampled at five-minute intervals. Option data include bid and ask quotes and associated quoted sizes, executed trading volumes, opening and closing prices, as well as high and low transaction prices, all recorded at the same five-minute frequency. The dataset also includes contract characteristics; namely option type (standard SPX or weekly SPXW), strike prices and expiration dates. Finally, the dataset reports the open interest measured at the beginning of each trading day.

Although both SPX and SPXW options are European-style and cash-settled, they differ in their trading and settlement conventions. Weekly SPXW options expire from Monday to Friday at the market close (PM settlement) and remain tradable throughout the expiration day. In contrast, standard monthly SPX options expire on the third Friday of each month at the market open (AM settlement) and cease trading at the market close on the preceding day. As a consequence, 0DTE SPX options are exposed to overnight gap risk between the last trading opportunity and settlement, whereas 0DTE SPXW options are not.

\paragraph{Data Cleaning and Sample Construction.}
The data-cleaning procedure is deliberately parsimonious. Rather than aggressively filtering observations ex ante, we rely on ARIES (introduced in \Cref{subsubsec:ARIES}) to identify and eliminate arbitrage-violating quotes. This approach allows us to preserve broad market information while ensuring arbitrage-free inputs at the estimation stage.

We implement the following preliminary filters. First, we discard all observations on standard SPX options and retain only SPXW contracts. Second, we remove option quotes exhibiting a zero bid price, zero open interest, or zero bid and ask sizes, as these are unlikely to reflect meaningful trading interest or reliable price discovery. Third, except for the specific purpose of estimating the risk-free rate and dividend yield outlined hereafter, we exclude in-the-money options.

\paragraph{Data Normalization} 
Throughout the following implementation, we exploit the first-order homogeneity of call prices with respect to the underlying asset and the strike. To improve numerical stability, we normalize all market quantities by the spot price $S_0$. Equivalently, we work in units of spot and set $S_0 = 1$. Under this normalization, strikes are replaced by $K_i/S_0$, the forward price is $F_0^T/S_0 = e^{(r-q)T}$, and bid and ask call prices become $C_i^{\mathrm{bid}}/S_0$ and $C_i^{\mathrm{ask}}/S_0$.

\paragraph{Conventions.}
\label{para:conventions_T}
For options with at least one day to expiry, we retain quotes observed 15 minutes before the market close, that is, at 16{:}00~EST. Time to maturity is measured on a calendar-day basis using a 365-day year; for instance, an option with one remaining trading day to expiration has $T=1/365$. Using a 252-day convention yields quantitatively similar results and does not affect our conclusions.
Throughout the analysis, the underlying index level is proxied by the mid-quote, defined as the average of the bid and ask prices. Together, these conventions and the limited filtering described above yield a clean dataset suitable for the estimation and empirical analyses that follow.
\paragraph{Estimation of the Risk-Free Rate and Dividend Yield}
Accurate estimation of the risk-free rate, $r$, and the continuous dividend yield, $q$, is fundamental for extracting arbitrage-free risk-neutral densities. While standard practice often relies on Treasury yields, recent literature (\cite{van2022risk, diamond2022risk}) argues that option-implied rates provide a more consistent pricing benchmark. Unlike Treasury valuations, which may be distorted by convenience yields, implied rates reflect the actual funding costs facing market participants.
Accordingly, we estimate $r$ and $q$ directly from the option chain at each maturity using put--call parity\footnote{Note that if $T < 1/365$, we apply the convention $r=q=0$, as detailed in Appendix~\ref{appendix:procedure_rq}.}, following the comprehensive procedure detailed in Appendix~\ref{appendix:procedure_rq}.

\paragraph{Unified Call Option Dataset}

Our raw data universe contains European calls and puts quoted with bid--ask spreads.
More precisely, we observe $N_c$ call options with strikes
$K_i^c$ and bid--ask quotes $(C_i^{\mathrm{bid}},C_i^{\mathrm{ask}})$
for $i=1,\dots,N_c$, and $N_p$ put options with strikes
$K_j^p$ and bid--ask quotes $(P_j^{\mathrm{bid}},P_j^{\mathrm{ask}})$
for $j=1,\dots,N_p$.

For the theoretical analysis, it is convenient to work with a single family of payoffs. The next lemma shows that, given the forward level $F_0^T$ and the discount factor
$e^{-rT}$, put constraints can be rewritten as constraints on call payoffs at the same
strikes with appropriately shifted bid--ask bounds.

We begin by collecting the baseline restrictions that define the discrete
admissible (risk-neutral) distributions on the grid $(s_i)_{i=1}^M$, including the
forward condition:
\begin{equation}
\mathcal{D}_M
:=
\left\{
p \in \Sigma_M
\;\middle|\;
\sum_{i=1}^M s_i p_i = F_0^T
\right\},
\label{eq:DM-density-class}
\end{equation}
where $\Sigma_M$ is the simplex defined in \eqref{eq:simplex}.

For brevity, given an admissible distribution \(p\), we introduce the associated call and put price functions by
\begin{equation}
C^{p}(K)
:=
e^{-rT}\sum_{i=1}^{M} (s_i-K)_+\,p_i,
\qquad K\ge 0,
\label{eq:Cp-def}
\end{equation}
and
\begin{equation}
P^{p}(K)
:=
e^{-rT}\sum_{i=1}^{M} (K-s_i)_+\,p_i,
\qquad K\ge 0.
\label{eq:Pp-def}
\end{equation}

The admissible set induced by the raw call--put dataset is
\begin{equation}
\mathcal{A}_{M}^{cp}
:=
\left\{
p \in \mathcal{D}_M
\;\middle|\;
\begin{aligned}
& C^{p}(K_i^c)\in [C_i^{\mathrm{bid}}, C_i^{\mathrm{ask}}],
\qquad i=1,\dots,N_c,\\[2pt]
& P^{p}(K_j^p)\in [P_j^{\mathrm{bid}}, P_j^{\mathrm{ask}}],
\qquad j=1,\dots,N_p.
\end{aligned}
\right\}.
\label{eq:AM-admissible-set-cp-compact}
\end{equation}

For convenience we will transform the put prices into synthetic calls by:
\begin{equation}
[\widetilde{C}_{j}^{\mathrm{bid}}, \widetilde{C}_{j}^{\mathrm{ask}}]
=
\Bigl[
P_{j}^{\mathrm{bid}} + e^{-rT}\bigl(F_0^T-K_j^p\bigr),
\;
P_{j}^{\mathrm{ask}} + e^{-rT}\bigl(F_0^T-K_j^p\bigr)
\Bigr],
\qquad j=1,\dots,N_p.
\label{eq:transformed-bounds-discrete}
\end{equation}

The unified call admissible set is
\begin{equation}
\mathcal{A}_{M}^{c}
:=
\left\{
p \in \mathcal{D}_M
\;\middle|\;
\begin{aligned}
& C^{p}(K_i^c)\in [C_i^{\mathrm{bid}}, C_i^{\mathrm{ask}}],
\qquad i=1,\dots,N_c,\\[2pt]
& C^{p}(K_j^p)\in [\widetilde{C}_j^{\mathrm{bid}}, \widetilde{C}_j^{\mathrm{ask}}],
\qquad j=1,\dots,N_p.
\end{aligned}
\right\}.
\label{eq:AM-admissible-set-callonly-compact}
\end{equation}

\begin{lemma}[Equivalence of Admissible Sets]
\label{lemma:equivalence_adm_set}
We have $\mathcal{A}_{M}^{cp}=\mathcal{A}_{M}^{c}$. That is, the set of discrete
admissible distributions constrained by the original dataset of calls and puts
coincides with the set constrained by calls and synthetic calls.
\end{lemma}
\begin{proof}
See Appendix~\ref{appendix:equi_admissible_sets}.
\end{proof}

Building upon this result, we restrict our analysis to call options without loss of generality, ensuring consistency with the market framework established in Section~\ref{sec:Notation and Market Framework}. Specifically, the unified dataset is mapped onto a single family of $N := N_p + N_c$ call constraints. The first $N_p$ entries represent synthetic calls derived from out-of-the-money puts at strikes $K_j^p$, with associated bid--ask spreads $[\widetilde{C}_{j}^{\mathrm{bid}}, \widetilde{C}_{j}^{\mathrm{ask}}]$.

\subsection{Numerical Implementation}
Our numerical procedure involves two sequential stages: isolating an arbitrage-free set of options and subsequently extracting the risk-neutral density. For the first stage, Proposition~\ref{prop:characterization_proof_version} ensures that whenever an arbitrage is detected (strong or weak), at least one upper size bound is active at an optimal solution of \eqref{eq:P3prime}. Based on this property, ARIES (Algorithm~\ref{alg:arb_filter}, detailed in the Appendix) (i) solves the linear program \eqref{eq:P3prime}, (ii) removes an option associated with an active upper bound (selecting the one with the smallest available size if multiple bounds are simultaneously active), and (iii) repeats this process until the unique optimizer is the null portfolio. Since at least one option is removed at each iteration, the procedure terminates after at most $N$ iterations. The selection rule (ii) effectively eliminates arbitrage while preserving the maximum amount of market information, yielding a strictly arbitrage-free dataset denoted by the index set $\mathcal I^*$.\\

Having secured the arbitrage-free subset $\mathcal I^*$, we proceed to apply SEDEx to the arbitrage-free quotes. As established in Theorem~\ref{thm:exist-unique-discrete-minimizer}, the associated discrete optimization problem on this filtered dataset is well-posed and admits a unique minimizer. To implement this extraction, we first define the spatial grid and then specify the objective function's parameters.

Let the strike bounds of our filtered dataset be defined as

\begin{equation}
    K_{\inf} := \min_{i\in\mathcal I^*} K_i,
    \qquad
    K_{\sup} := \max_{i\in\mathcal I^*} K_i,
    \label{eq:strike-bounds}
\end{equation}
and let $\sigma_{\mathrm{ATM}}$ denote the implied volatility computed from at-the-money (ATM) mid price. 

A natural requirement on the spatial discretization, discussed in Section~\ref{subsec:mesh_size}, is that the mesh size $\Delta s$ be sufficiently fine relative to the typical variation in the underlying corresponding to a probability increment $\delta$; see \eqref{eq:relative-variation-def}. In practice, we set $\delta = 0.5\%$ and use the asymptotic guideline in \eqref{eq:R-min-asympt} as a target resolution,
\begin{equation}
\Delta s_{\mathrm{tar}}
=
\sigma_{\mathrm{ATM}} \sqrt{2\pi T}\,\delta.
\label{eq:target-mesh-size}
\end{equation}
Let $\eta$ denote the mesh of the normalized strike lattice.\footnote{SPXW strikes are quoted on a fixed exchange lattice, for example in 5-index-point increments over the relevant range. A filtered option chain may contain only a subset of these strikes, so consecutive observed strikes can be separated by larger gaps. We therefore take the mesh of the exchange lattice as the reference mesh and choose the SEDEx grid as its refinement.}
We choose the smallest integer
\begin{equation}
m
:=
\min\left\{
\ell\in\mathbb N^{*}:\frac{\eta}{\ell}\le \Delta s_{\mathrm{tar}}
\right\},
\label{eq:grid-refinement-factor}
\end{equation}
and set
\begin{equation}
\Delta s:=\frac{\eta}{m}.
\label{eq:compatible-mesh-size}
\end{equation}

We next define a preliminary upper support bound as the maximum of two complementary quantities:
\begin{equation}
\label{eq:preliminary-upper-bound}
b_0=\max\left(
K_{\sup}+(\kappa_{1}-1)(K_{\sup}-K_{\inf}),
\;
S_{0}e^{\kappa_{2}\sigma_{\mathrm{ATM}}\sqrt{T}}
\right).
\end{equation}
The first term extends the observed strike range by a proportion controlled by $\kappa_{1}>1$, thereby enlarging the support beyond the largest quoted strike $K_{\sup}$ while preserving the scale of the quoted interval $[K_{\inf},K_{\sup}]$.  
The second term provides a volatility-based upper bound derived from a lognormal scaling around the spot level $S_{0}$, where $\kappa_{2}>0$ controls the number of volatility standard deviations retained in the support.
In practice\footnote{For computational efficiency, the lower bound of the support can alternatively be truncated to $\min\left(
K_{\inf}-(\kappa_{1}-1)(K_{\sup}-K_{\inf}),
\;
S_{0}e^{-\kappa_{2}\sigma_{\mathrm{ATM}}\sqrt{T}}
\right)$. Empirical evidence suggests this yields identical results while substantially reducing the number of discretization points.}, we set $\kappa_1 = 1.50$ and $\kappa_2 = 10$.

The final upper endpoint of the numerical support is the nearest point of the refined lattice above \(b_0\):
\begin{equation}
b
:=
\Delta s\left\lceil \frac{b_0}{\Delta s}\right\rceil,
\qquad
M:=\frac{b}{\Delta s}.
\label{eq:upper-bound-support}
\end{equation}
Thus all filtered observed strikes belong to the numerical grid \((s_i)_{i=1}^{M}\) defined in \eqref{eq:grid_def}.

With the grid fully specified, the final requirement is to calibrate the hybrid criterion $H_M$ defined in \eqref{eq:def-HM}, which involves the regularization weights $\lambda_1$ and $\lambda_2$. Following Section~\ref{subsec:reg_weight}, we select $\lambda_1$ and $\lambda_2$ so that their ratio satisfies the scaling relation in \eqref{eq:lambda-ratio-scaling}, yielding a reasonable balance between the $L^2$ and entropy terms in the density recovery:
$$
\frac{\lambda_{1}}{\lambda_{2}}
\sim
-4\sqrt{\pi}\,\sigma_{\mathrm{ATM}}^{3}T^{3/2}
\ln(\sigma_{\mathrm{ATM}}\sqrt{T}).
$$
The complete sequence of SEDEx operations for this second stage is summarized in Algorithm~\ref{alg:rnd_extraction}, which is detailed in the Appendix.

\subsection{Consistency Tests}

This subsection provides a validation step between the theoretical results of \Cref{sec:theoretical_results} and the empirical analysis. It is organized around two complementary exercises. First, the Cousot conditions are used as a consistency test to determine whether observed bid--ask quotes contain static arbitrage and, if so, which type of violation is present. Second, a synthetic Heston panel serves as a standard toy model to evaluate the full procedure in a controlled environment where the underlying risk-neutral density is known.

\subsubsection{Cousot Conditions}
\label{subsec:cousot_cond}
Within the market setup introduced above, the absence of static arbitrage implies a set of model-free constraints on quoted call prices. These constraints correspond to the bid--ask counterparts of the standard positivity, monotonicity, convexity, and lower-bound conditions for call option prices. They are closely related to the inequalities established by \citet[Proposition~3]{Cousot2006}, except that in our framework they must hold in strict rather than weak form.

This distinction is important because weak inequalities only exclude strong (model-independent) arbitrage. As discussed in \citet[Remark~4]{Cousot2006}, weak arbitrages inherently depend on the choice of the probability measure and cannot be ruled out prior to its specification, as they may materialize on events of zero probability. Accordingly, weak inequalities are necessary model-free restrictions, but they do not exclude all forms of arbitrage.

When these weak inequalities hold, \citet[Proposition~5]{Cousot2006} shows that there exists at least one (atomic) martingale measure consistent with observed quotes. However, such a measure might assign zero probability to states that would otherwise generate weak arbitrages. Under our no-arbitrage definition (\Cref{def:arbitrage}), equality cases must be excluded because they generate zero-cost portfolios with nonnegative payoff everywhere and strictly positive payoff on a nontrivial set of states. Consequently, we work throughout with the strict counterparts of these inequalities.

\begin{enumerate}

\item \textbf{Positivity:} For all \(i \in \{0,\dots,N\}\), the ask price with strike \(K_i\) must be strictly positive:
\begin{equation}
\label{eq:pos-ask}
C_i^{\mathrm{ask}} > 0 .
\end{equation}

\item \textbf{Vertical Spread - Monotonicity:} For all \(0 \leq i < j \leq N\), the cost of a vertical spread with strikes \((K_i,K_j)\) executed at market quotes must be strictly positive:
\begin{equation}
\label{eq:vertical}
C_i^{\mathrm{ask}} - C_j^{\mathrm{bid}} > 0 .
\end{equation}

\item \textbf{Butterfly Spread - Convexity:} For all \(0 \leq i < j < k \leq N\), the cost of a butterfly strategy with strikes \((K_i,K_j,K_k)\) executed at market quotes must be strictly positive:
\begin{equation}
\label{eq:butterfly}
\frac{C_i^{\mathrm{ask}} - C_j^{\mathrm{bid}}}{K_j - K_i} - \frac{C_j^{\mathrm{bid}} - C_k^{\mathrm{ask}}}{K_k - K_j} > 0 .
\end{equation}

\item \textbf{Lower Bound:} For all \(i \in \{1,\dots,N\}\), call ask prices must satisfy the lower bound:
\begin{equation}
\label{eq:synth-put-2}
C_i^{\mathrm{ask}} - S_0 e^{-qT} + K_i e^{-rT} > 0 .
\end{equation}
\end{enumerate}
In particular, no inequality involves more than three quoted strikes.
These conditions will subsequently be employed as a consistency test within our arbitrage detection and removal procedure to ensure the absence of arbitrage in the processed data.

\subsubsection{Toy Model: Heston}
We consider a 1DTE slice of European call prices 
generated under the Heston stochastic volatility model 
\citep{heston1993closed}. Heston serves here solely as a controlled 
synthetic environment providing a known ground-truth risk-neutral 
density. The extraction procedure itself is fully model-free: it makes 
no parametric assumption about the underlying price dynamics. Parameter 
values are reported in the footnote below.\footnote{The Heston parameters used in the toy experiment are {\(v_0 = 0.117 \cdot10^{-1}\), \(\theta = 0.394\cdot10^{-1}\), \(\kappa = 1.0\), \(\sigma = 0.30\), and \(\rho = -0.70\).}}

Our synthetic option panel consists of 84 strikes ranging from \(0.87\) to \(1.03\). This asymmetric strike range is designed to mimic the structure of short-dated option data, where quoted strikes are often deeper on the in-the-money side than on the out-of-the-money side. Throughout the experiment, the option prices and associated density are computed using the COS method \citep{fang2009novel} and serve as references.

The purpose of this toy example is threefold. First, we verify that the extraction procedure in \Cref{alg:rnd_extraction} recovers a density consistent with the true Heston density when quotes are frictionless and arbitrage-free. Second, we study the effect of introducing bid--ask spreads while preserving arbitrage consistency. Third, we analyze the case where arbitrage violations are deliberately injected into the bid--ask quotes. We finally show that ARIES (\Cref{alg:arb_filter}) is necessary to restore feasibility and recover a stable density estimate.

\paragraph{Frictionless Baseline}We begin with a frictionless baseline in which the call prices are observed without bid--ask spreads and satisfy the no-arbitrage conditions. In this setting, SEDEx performs as expected. As shown in Figure~\ref{fig:heston_frictionless}.a, the recovered risk-neutral density closely overlaps the Heston density over the relevant support. The same conclusion emerges from the pricing side: the call prices implied by the extracted density are indistinguishable from the Heston model prices across the whole strike range considered (Figure~\ref{fig:heston_frictionless}.b). This first experiment therefore validates SEDEx in the ideal case, confirming the procedure's ability to recover the underlying risk-neutral distribution with negligible error from exact, arbitrage-free quotes.

\begin{figure}[t!]
    \centering
    \begin{subfigure}[b]{0.48\textwidth}
        \centering
        \includegraphics[width=\textwidth]{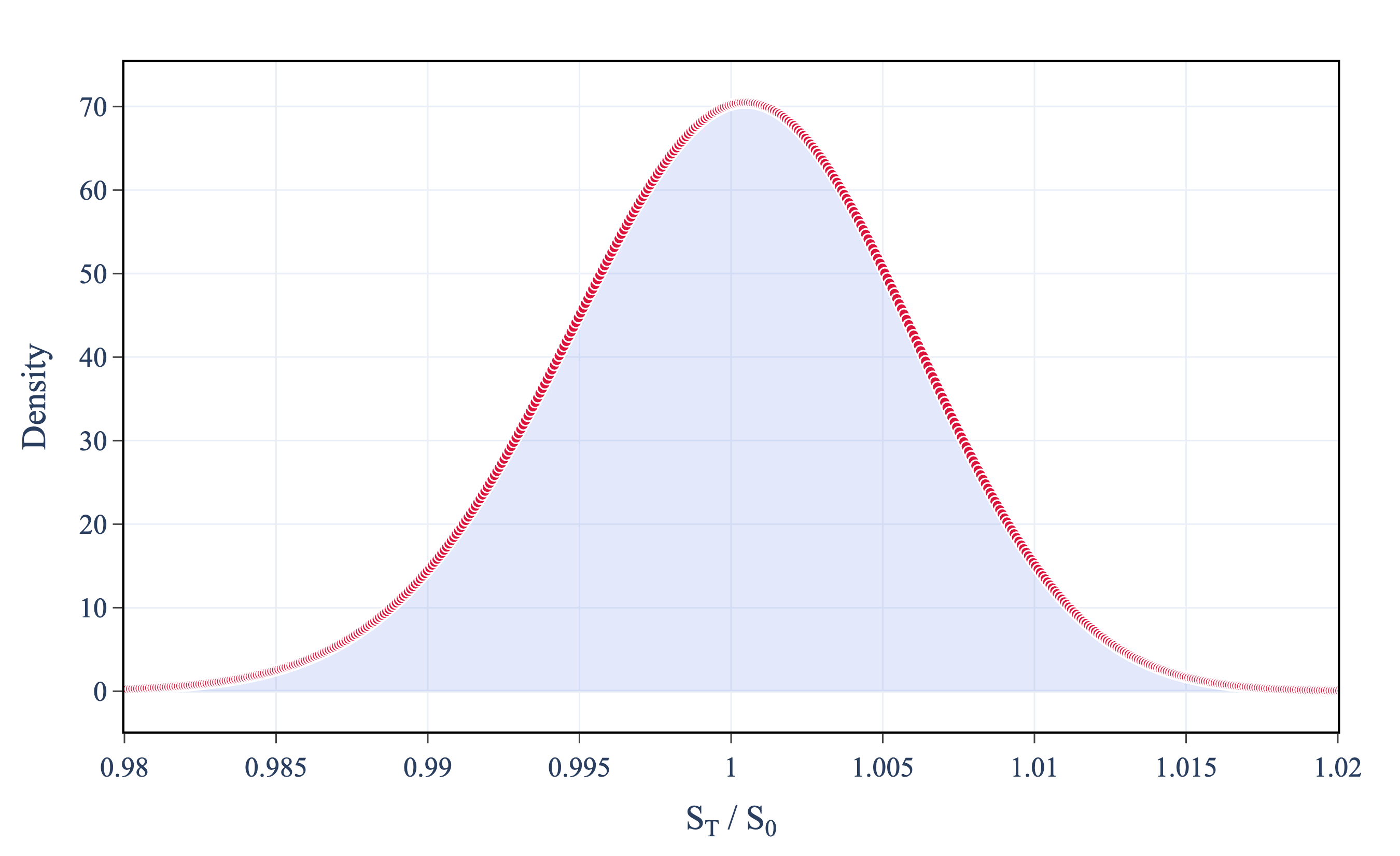}
        \caption*{(a)}
    \end{subfigure}
    \hfill
    \begin{subfigure}[b]{0.48\textwidth}
        \centering
        \includegraphics[width=\textwidth]{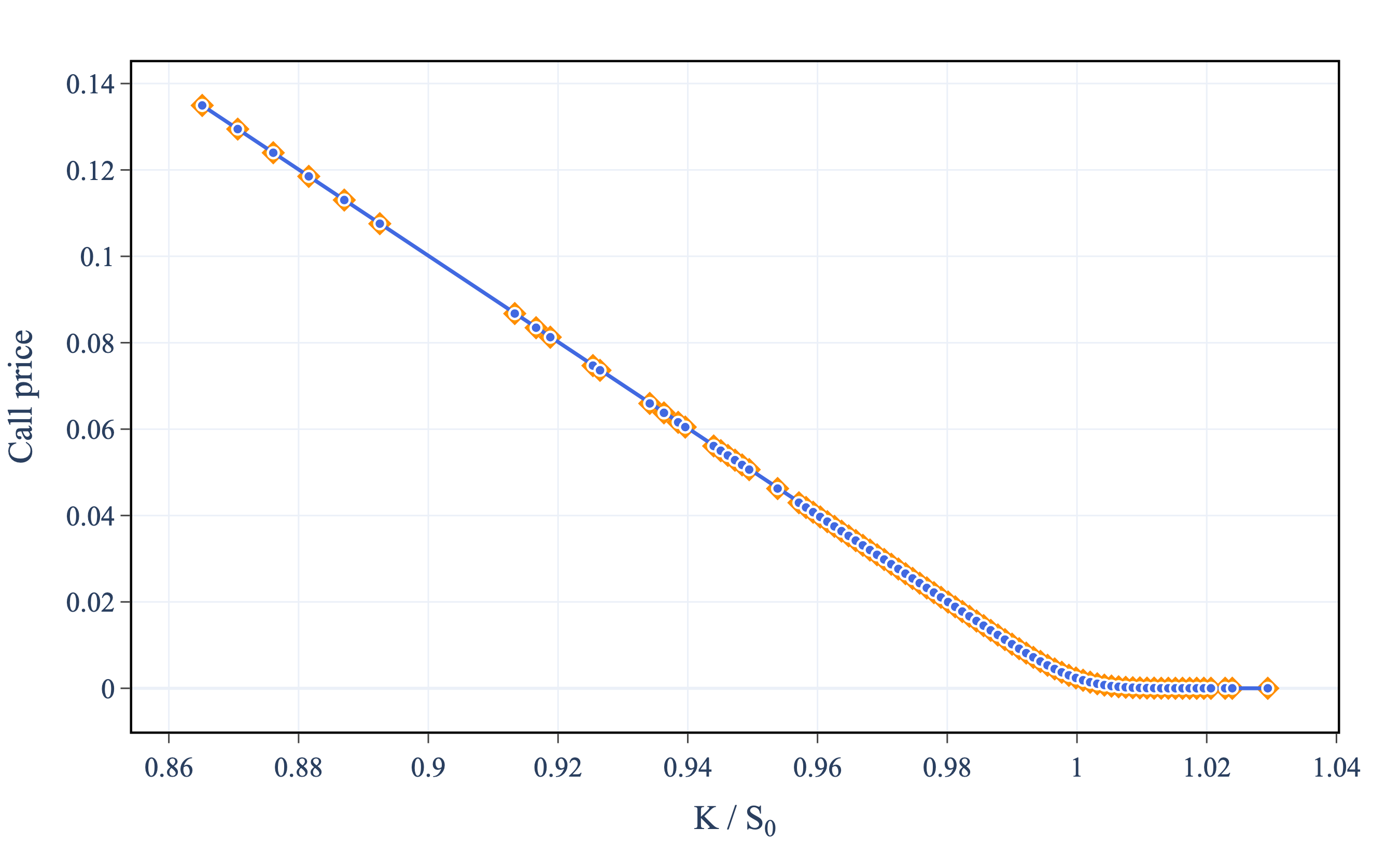}
        \caption*{(b)}
    \end{subfigure}
    \caption{Frictionless baseline for 1DTE options. \textit{(a)} Heston reference 
    risk-neutral density (blue solid line) and SEDEx density (red dots): the two 
    densities are visually indistinguishable at the scale of the figure. \textit{(b)} 
    Heston call prices (blue line with circles) and call prices implied by SEDEx 
    (orange diamonds): the two series overlap perfectly over the full strike range, up to a
    numerical tolerance of $10^{-4}$ in absolute price terms.}
    \label{fig:heston_frictionless}
\end{figure}

\paragraph{Arbitrage-Free Bid--Ask Setting}We next introduce artificial bid-ask spreads. In the absence of a standard benchmark in the  literature, we adopt the following parsimonious parametric form: taking the Heston prices as mid quotes, we construct bid
and ask by subtracting and adding, respectively, half of a
strike-dependent spread
\begin{equation}
    \label{eq:toy-spread}
    S(K) \;=\; S_{\mathrm{peak}}
              \exp\!\left(-\frac{(K - F_0^T)^2}{2h^2}\right)
              + S_{\mathrm{base}}.
\end{equation}
In this specification, $S_{\mathrm{base}} > 0$ defines the baseline spread, and the Gaussian component is centered at the forward price with amplitude $S_{\mathrm{peak}}$ and scale $h$.

The parameters are set to reflect typical short-dated quotes. The floor
$S_{\mathrm{base}}$ is set to one tick (0.05 index points), the minimum quoted spread
for SPXW contracts quoted below three points. This regime is typical of
short-dated deep out-of-the-money contracts. The amplitude
$S_{\mathrm{peak}}$ is set to four ticks. 
The width is taken to be $h = \sigma_{\mathrm{ATM}}\sqrt{T}$ to anchor the spread to the natural risk-neutral dispersion scale of the underlying over the life of the option.

The \textit{absolute} spread $S$ has a Gaussian shape, which may seem
inconsistent with the view that liquidity concentrates at the money, but the two are readily reconciled. Transaction costs are better measured by the
\emph{relative} spread
\begin{equation*}
    R(K) \;:=\; \frac{S(K)}{M(K)},
\end{equation*}
where $M(K)$ denotes the out-of-the-money mid price: the put for
$K < F_0^T$ and the call for $K \geq F_0^T$. At the forward, $R$ is
small: its numerator is bounded by a few ticks, its denominator
is the at-the-money premium, which is orders of magnitude larger.
Far from the money, option premia collapse toward the tick
floor. Both numerator and denominator are then of order
$S_{\mathrm{base}}$, so $R$ is of order one. The relative spread
therefore traces a U-shape with a global minimum at
$K = F_0^T$. This argument is made rigorous in the footnote below.\footnote{
That
$F_0^T$ is a strict local minimum of $R$ follows from two structural properties, independent of the parameter values. First, \(S'(F_0^T)=0\) by construction. Second, the mid price has a kink at $F_0^T$, with positive left derivative from the put and negative right derivative from the call. A quotient rule computation then gives
\begin{equation}
    R'\bigl((F_0^T)^-\bigr) < 0 < R'\bigl((F_0^T)^+\bigr).
    \label{eq:relative_spread_ratio}
\end{equation}

Global minimality over the traded strike range follows from the
log-slope dominance condition
\[
    \Lambda^{C}(K) \,>\, \Gamma(K) \text{ for } K > F_0^T,
    \qquad
    \Lambda^{P}(K) \,>\, \Gamma(K) \text{ for } K < F_0^T,
\]
with $\Lambda^{C}(K) = -\mathrm{d}\ln C(K)/\mathrm{d}K$,
$\Lambda^{P}(K) = \mathrm{d}\ln P(K)/\mathrm{d}K$, and
$\Gamma(K) = \left\lvert \mathrm{d}\ln S(K)/\mathrm{d}K \right\rvert$.
This condition is verified numerically for the parameter values used throughout.}

In this setting, the SEDEx density remains close to the reference Heston density, although the two no longer coincide exactly (Figure~\ref{fig:heston_bidask}.a). This discrepancy is not surprising. Once option prices are observed through bid--ask intervals rather than as exact point values, multiple risk-neutral densities may be consistent with the same admissible quotes. The role of our hybrid criterion is precisely to select a unique density within this feasible set. As shown in Figure~\ref{fig:heston_bidask}.b, the call prices implied by SEDEx are very close to the Heston reference prices. Furthermore, they remain within the prescribed bid--ask bounds over the full strike range, which is the relevant requirement in this setting.

\begin{figure}[t!]
    \centering
    \begin{subfigure}[b]{0.48\textwidth}
        \centering
        \includegraphics[width=\textwidth]{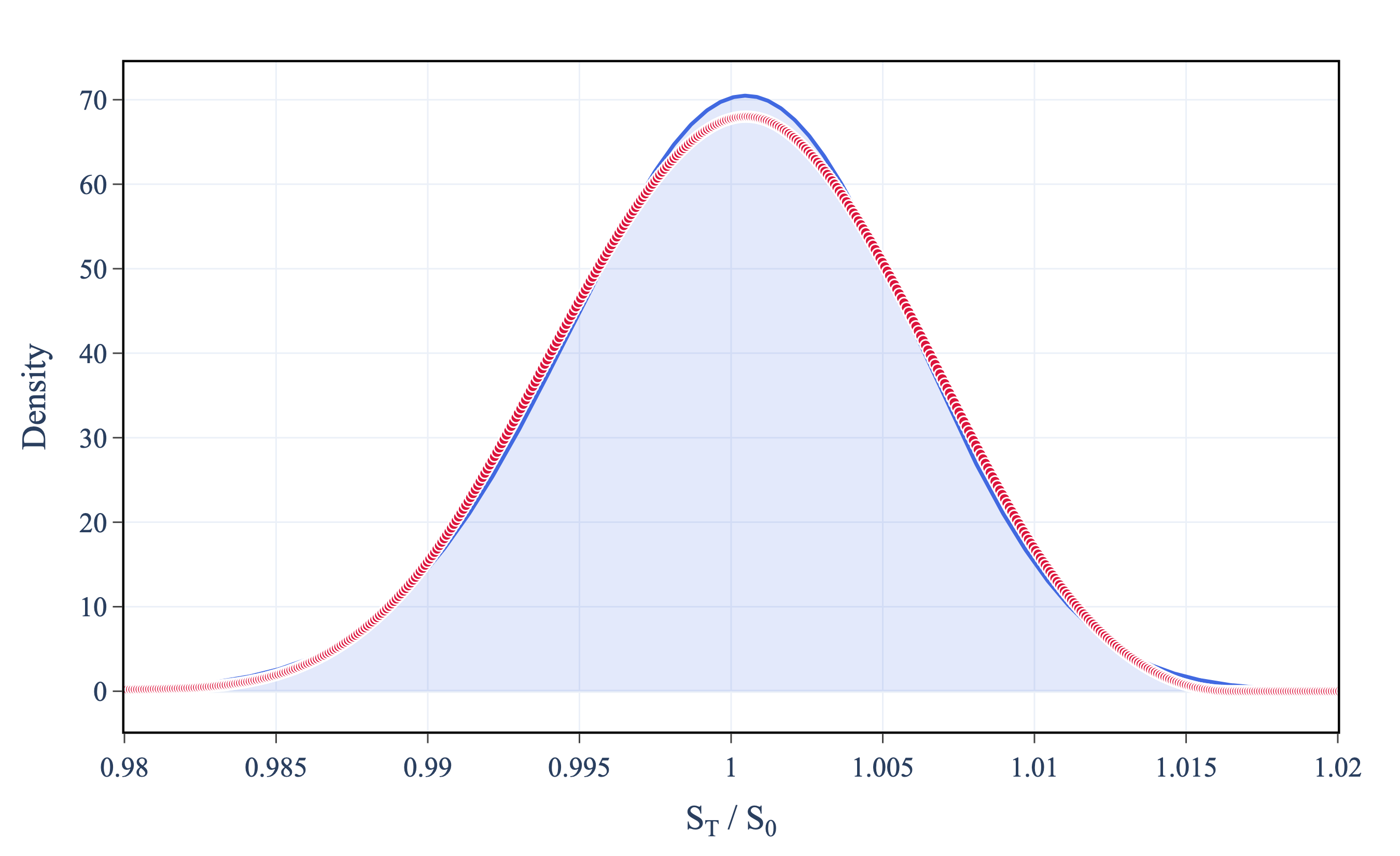}
        \caption*{(a)}
    \end{subfigure}
    \hfill
    \begin{subfigure}[b]{0.48\textwidth}
        \centering
        \includegraphics[width=\textwidth]{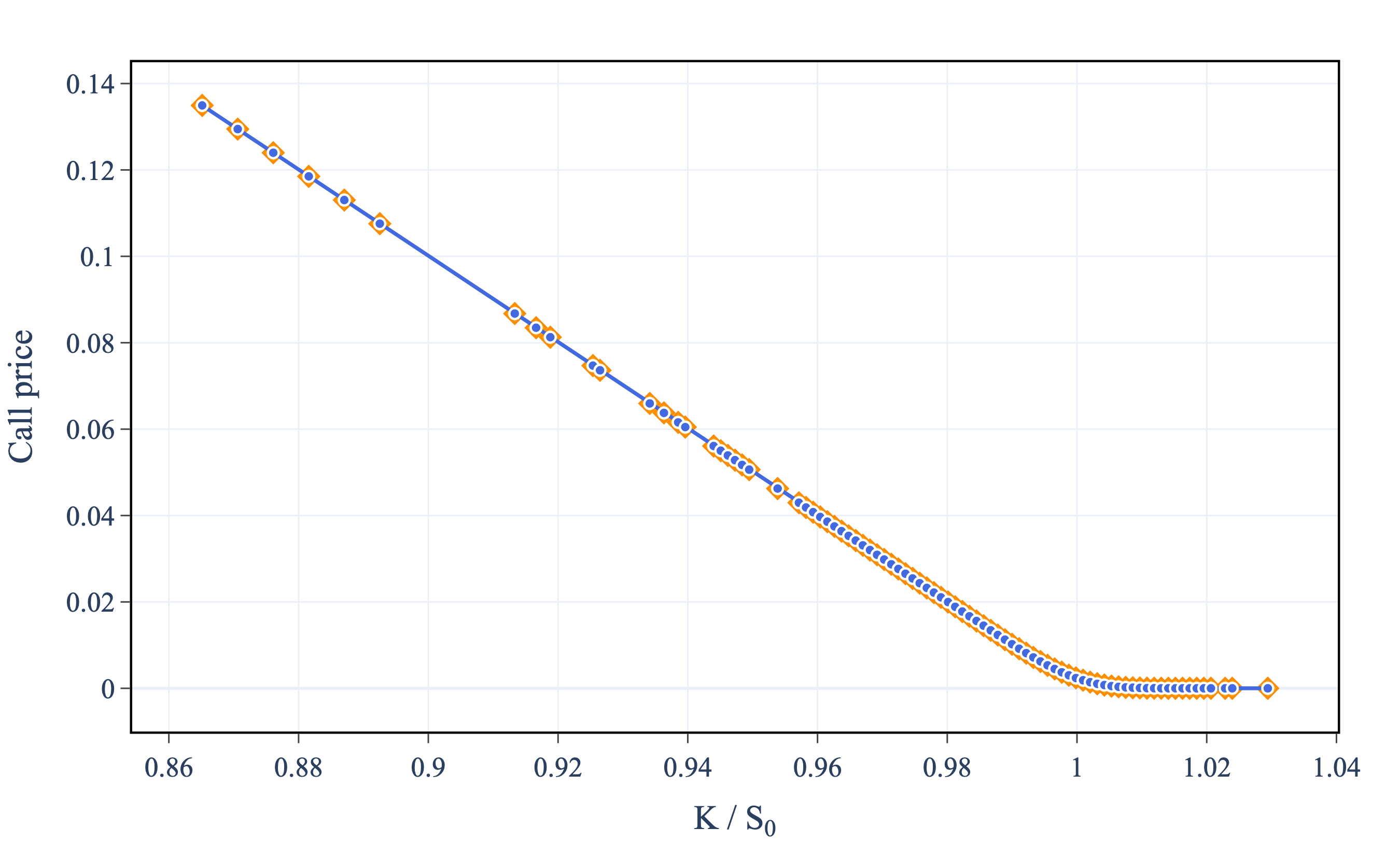}
        \caption*{(b)}
    \end{subfigure}
    \caption{Arbitrage-free bid--ask setting for 1DTE options. \textit{(a)} Heston 
    reference risk-neutral density (blue solid line) and SEDEx density (red dots): 
    both remain close but do not coincide exactly under interval constraints. 
    \textit{(b)} Heston reference call prices (blue line with circles) and call prices 
    implied by SEDEx (orange diamonds): the recovered prices satisfy the 
    bid--ask constraints up to a numerical tolerance of $10^{-7}$ in 
    absolute price terms.}
    \label{fig:heston_bidask}
\end{figure}
\paragraph{Arbitrage-Contaminated Bid--Ask Setting}Finally, we consider a deliberately contaminated bid--ask panel in which arbitrage violations are introduced by perturbing quoted prices across a broad portion of the strike range\footnote{The contamination is constructed by increasing selected bid quotes and decreasing selected ask quotes so as to create static arbitrage violations. Therefore, the contaminated bid--ask intervals are included in the original ones.}. Affecting over one-third of the strikes, this setup amplifies typical market contamination to an unusually severe level, incidentally serving as a rigorous stress test for the full pipeline.

\Cref{tab:cousot_heston_arbitrage}  reports the number of violated inequalities when the perturbed quotes are tested against the strict Cousot conditions of \Cref{subsec:cousot_cond}. The very large number of butterfly violations should be interpreted with care. Because our implementation checks convexity across all admissible butterfly combinations rather than strictly on adjacent triplets, a distortion affecting a limited subset of quotes can simultaneously invalidate a vast number of inequalities.
\begin{table}[H]
\centering
\caption{Number of violated Cousot inequalities after arbitrage injection in the Heston bid--ask panel.}
\label{tab:cousot_heston_arbitrage}
\vspace{0.2cm}
\begin{tabular}{lc}
\toprule
Type of inequality & Number of violations \\
\midrule
Lower bound & 18 (out of 84) \\
Vertical spread & 110 (out of 3,486) \\
Butterfly spread & 12,902 (out of 95,284) \\
\bottomrule
\end{tabular}
\end{table}
\noindent

Without prior application of ARIES, the inverse problem is infeasible, so no admissible density can match the contaminated quotes. This failure highlights the necessity of our filtering procedure. Because many violated constraints overlap, multiple quote-removal configurations could theoretically restore feasibility. Our procedure resolves this ambiguity by prioritizing the removal of the least informative quotes, characterized by the smallest available size. After filtering, density extraction becomes feasible again on the arbitrage-free subset extracted from the contaminated panel. 
As shown in \Cref{fig:heston_bidaskarb}, the SEDEx density remains very close to the Heston density and is visually indistinguishable from the density obtained in the arbitrage-free bid--ask experiment. The call prices implied by this density are, by construction, compatible with the arbitrage-free subset retained in the decontaminated panel. More interestingly, when the density is evaluated on the broader set of strikes from the original non-contaminated panel, the resulting call prices still fall within the original bid--ask bounds, including at strikes removed during filtering.

\begin{figure}[H]
    \centering
    \begin{subfigure}[b]{0.48\textwidth}
        \centering
        \includegraphics[width=\textwidth]{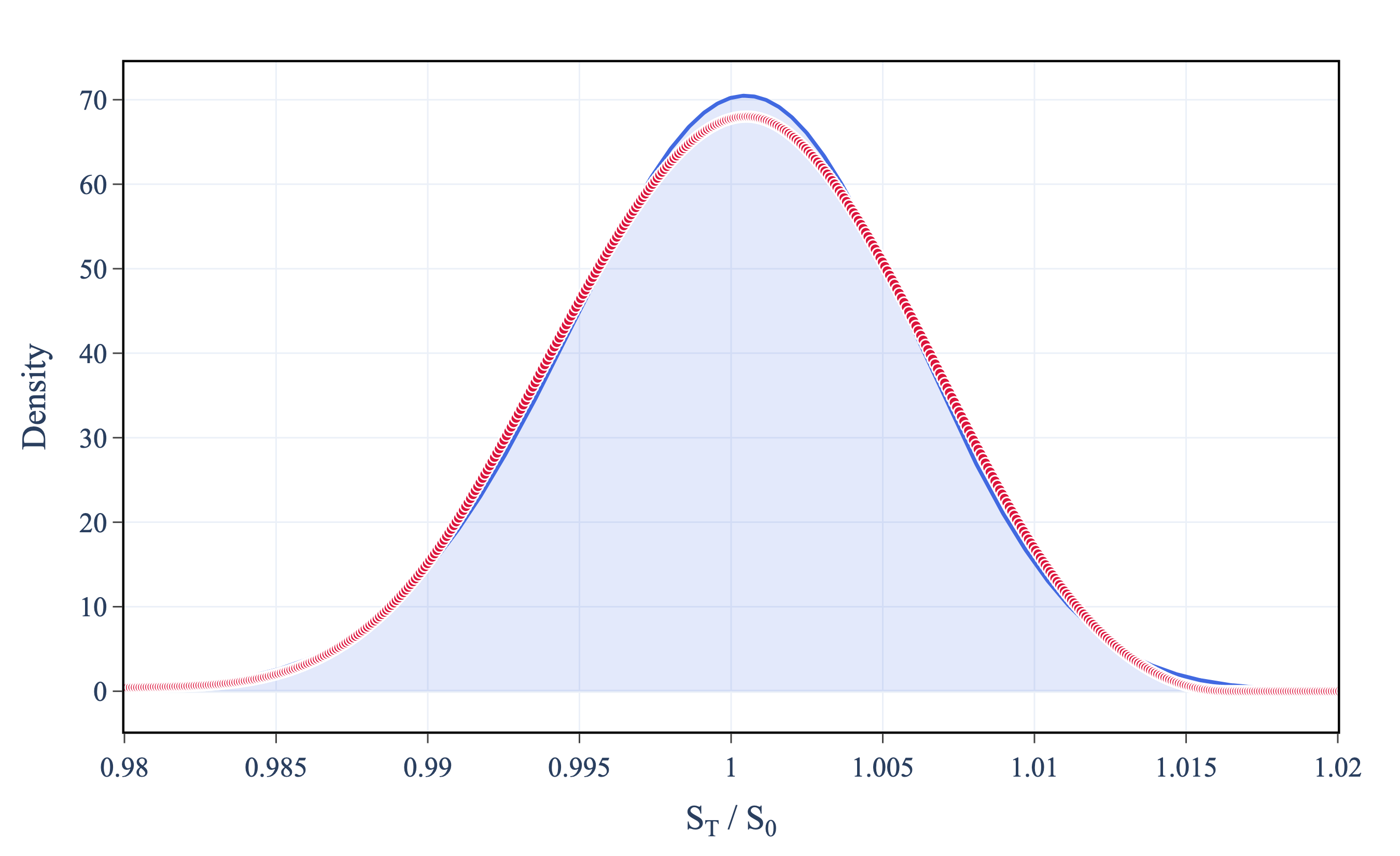}
        \caption*{(a)}
    \end{subfigure}
    \hfill
    \begin{subfigure}[b]{0.48\textwidth}
        \centering
        \includegraphics[width=\textwidth]{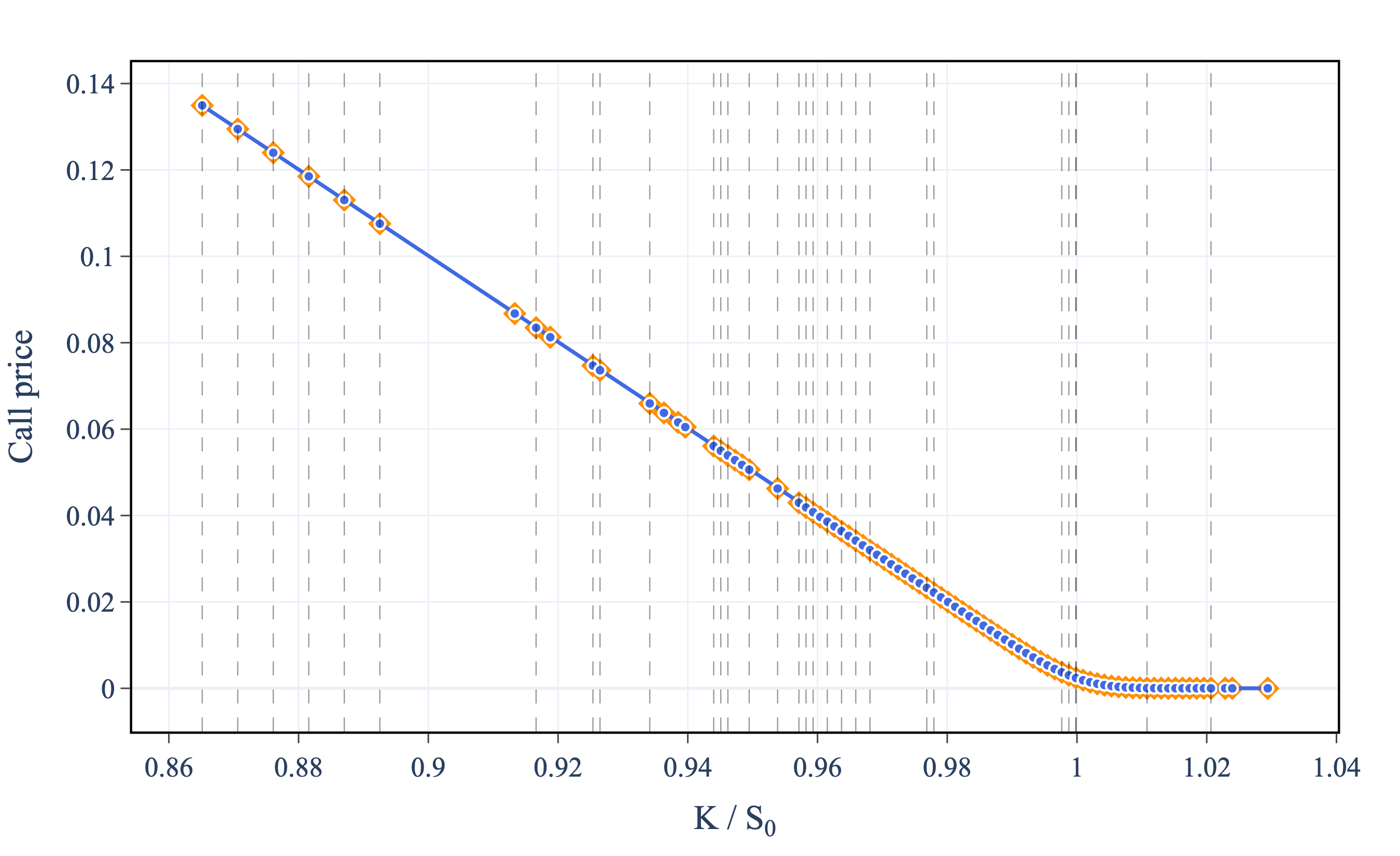}
        \caption*{(b)}
    \end{subfigure}
    \caption{Arbitrage-contaminated bid--ask setting for 1DTE options. \textit{(a)} 
    Heston reference risk-neutral density (blue solid line) and SEDEx density after 
    ARIES filtering (red dots). \textit{(b)} Heston reference call prices (blue line 
    with circles) and call prices implied by SEDEx (orange diamonds). 
    Dashed vertical lines indicate strikes removed by ARIES. Maximum 
    deviation from the original bid--ask bounds: $10^{-7}$ in absolute price terms.}
    \label{fig:heston_bidaskarb}
\end{figure}

In summary, this synthetic experiment illustrates the ability of the SEDEx procedure to recover the risk-neutral density in a controlled environment. Furthermore, it highlights the necessity of prior arbitrage filtering when quote inconsistencies render the inverse problem infeasible.

\subsection{Market Data: Results and Application}

We next evaluate the proposed methodology on representative market data, where the risk-neutral density is not directly observable, unlike in the toy experiment above. In all cases considered below, the full pipeline (from arbitrage filtering to density extraction) runs in a few seconds on a standard laptop, which makes the procedure suitable for intraday monitoring as well as for empirical work.

\paragraph{First Test: ``Business As Usual''}
We begin with the 1DTE slice observed on 2023-07-19, for which 120 strikes remain after preliminary cleaning and 118 after ARIES filtering\footnote{For this first test, the spot is 4565. ARIES detects a strong arbitrage at strike 4220 and a weak arbitrage at strike 4300. All other slices discussed in this subsection are already arbitrage-free.}. \Cref{fig:market_1dte_calm} reports the extracted density and the associated implied-volatility diagnostics, showing that SEDEx performs very well in this case. For comparison, we also calibrate the implied volatilities from mid prices to a standard SVI parameterization (\cite{Gatheral2014}),
\begin{equation}
w(k)=a+b\left[\rho (k-m)+\sqrt{(k-m)^2+\sigma^2}\right],
\qquad k=\log\!\left(\frac{K}{F_0^T}\right).
\label{eq:svi_slice}
\end{equation}

SVI provides a good overall fit, especially in the wings, but is less accurate around the money and can lie by as much as 1.90 volatility points outside the quoted interval\footnote{For this slice, the standard SVI calibration outputs parameters $(a,b,\rho,m,\sigma)=(-0.0357,\,0.1107,\,0.4095,\,0.1640,\,0.3544)$.}. By contrast, the smile implied by the SEDEx procedure remains admissible over the full chain with maximum error of $10^{-6}$ volatility points. More generally, this slice illustrates one of the main applications of risk-neutral density extraction: once an admissible density has been recovered from quoted prices, it can be used to construct a full implied-volatility smile by repricing options on a finer grid of strikes, thereby delivering a nonparametric interpolation and extrapolation of the smile.

\begin{figure}[t!]
    \centering
    \begin{subfigure}[b]{0.48\textwidth}
        \centering
        \includegraphics[width=\textwidth]{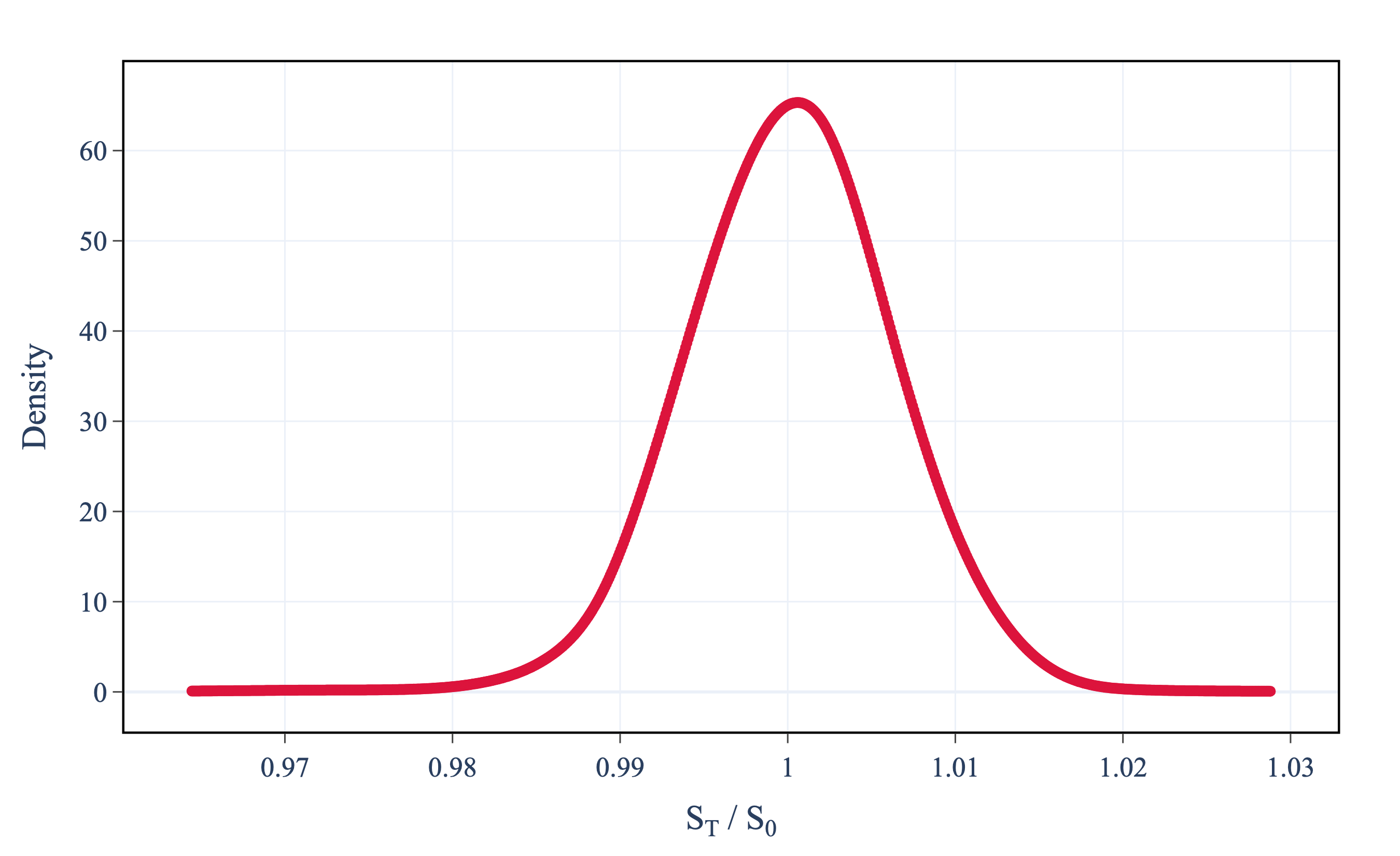}
        \caption*{(a)}
    \end{subfigure}
    \hfill
    \begin{subfigure}[b]{0.48\textwidth}
        \centering
        \includegraphics[width=\textwidth]{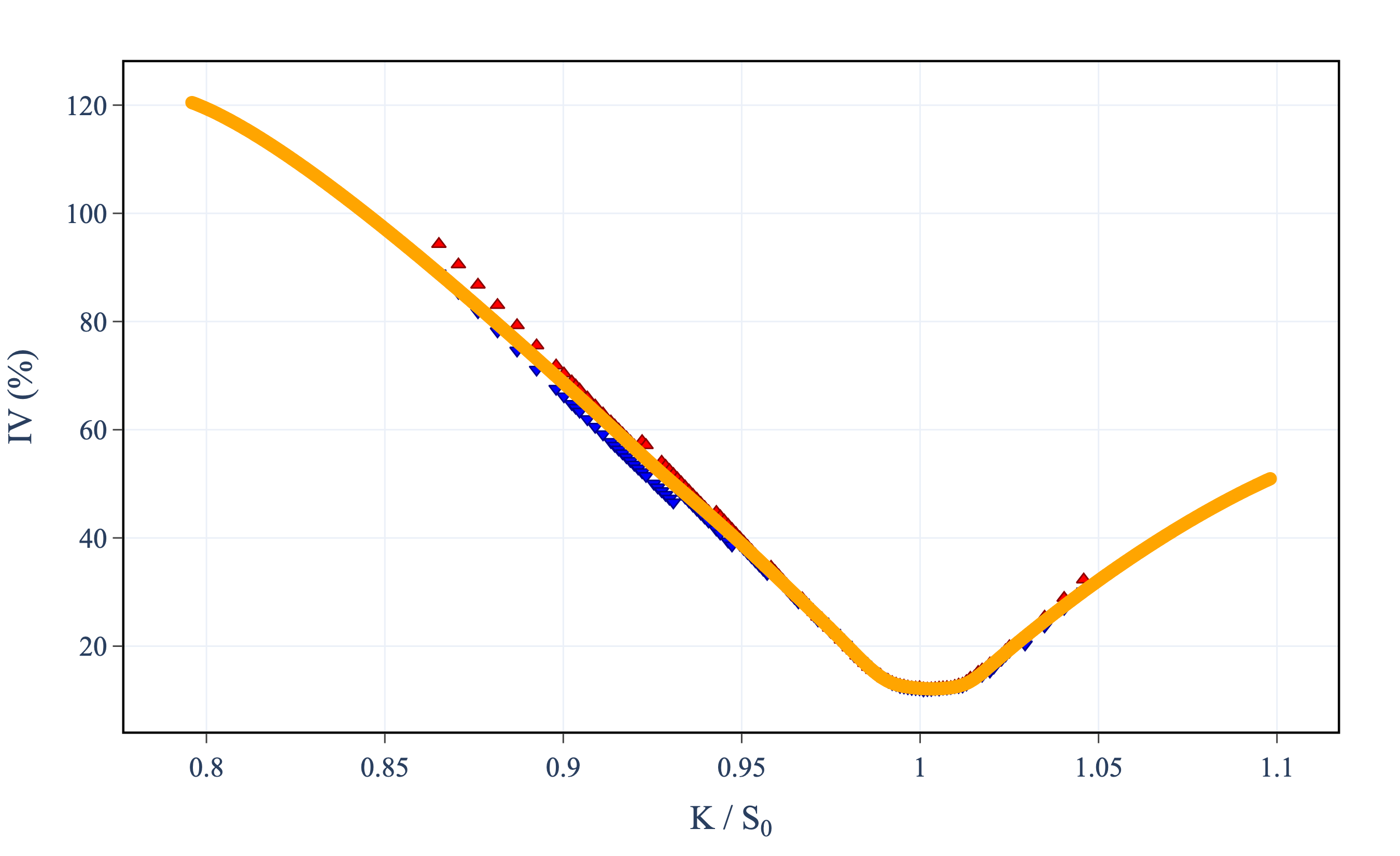}
        \caption*{(b)}
    \end{subfigure}
    \vspace{0.4em}
    \begin{subfigure}[b]{0.48\textwidth}
        \centering
        \includegraphics[width=\textwidth]{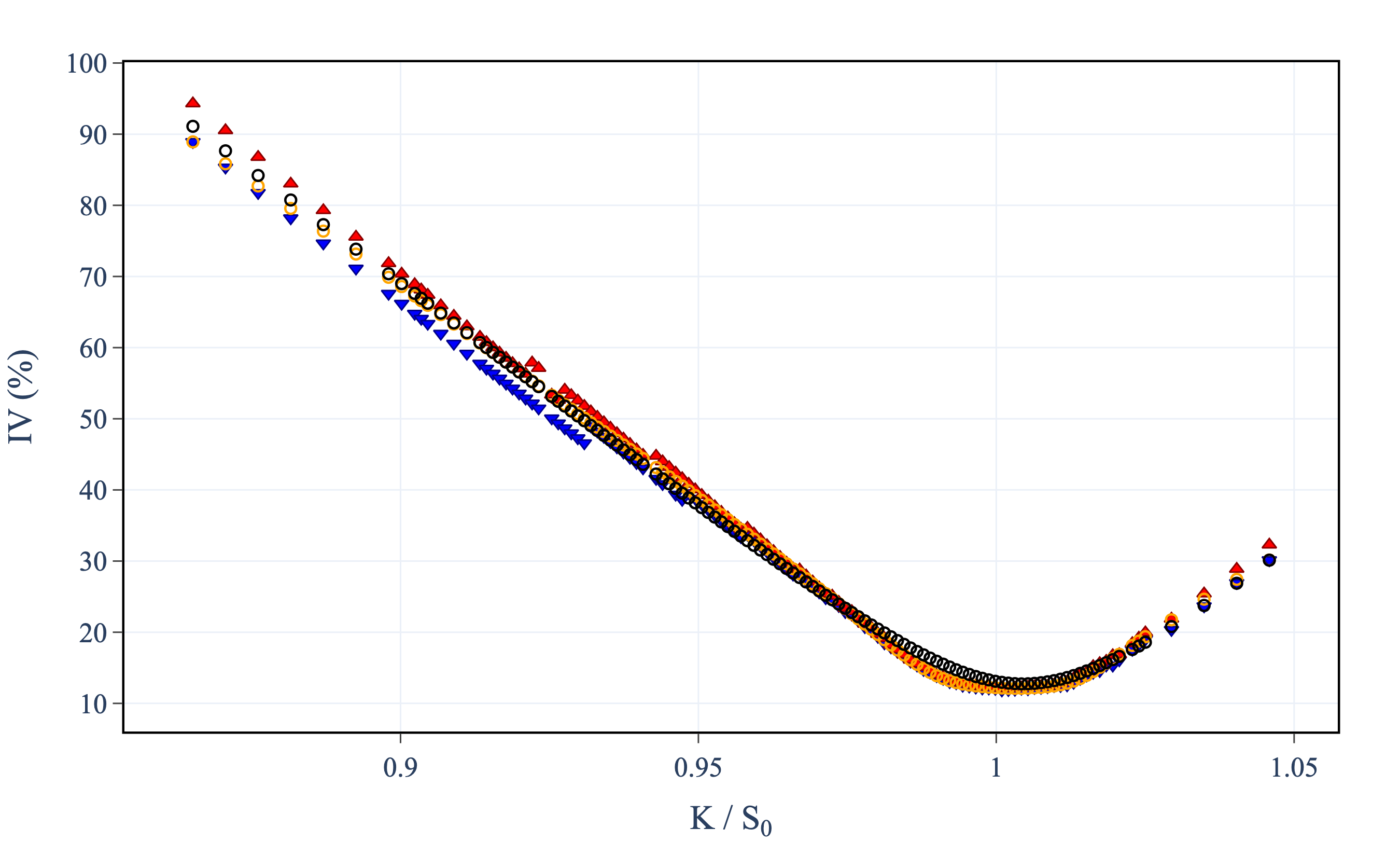}
        \caption*{(c)}
    \end{subfigure}
    \hfill
    \begin{subfigure}[b]{0.48\textwidth}
        \centering
        \includegraphics[width=\textwidth]{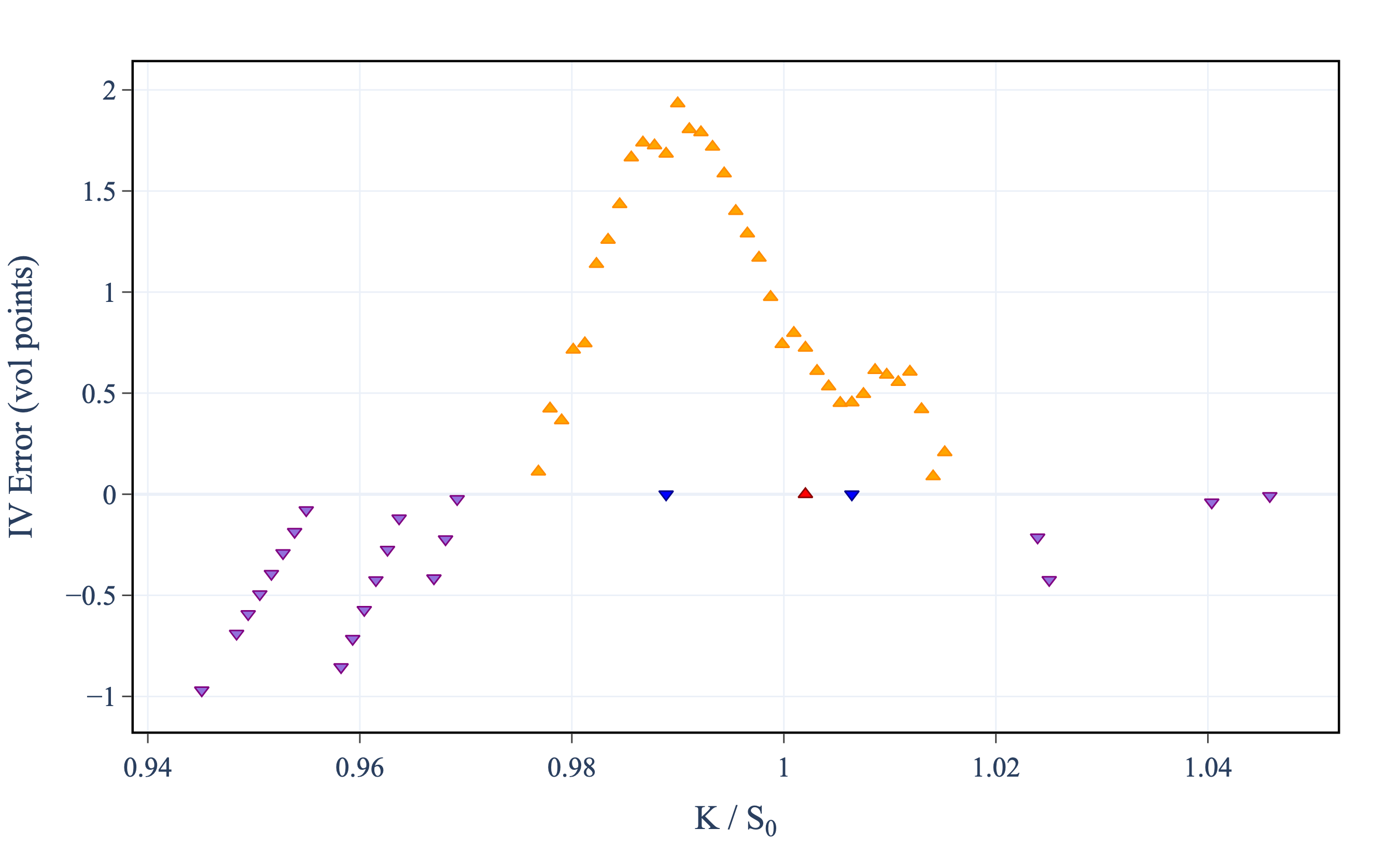}
        \caption*{(d)}
    \end{subfigure}
\caption{1DTE slice on 2023-07-19. \textit{(a)} Risk-neutral density extracted 
    by SEDEx (red dots). \textit{(b)} Implied volatility smile obtained by repricing 
    options on a finer strike grid: market ask (red triangles-up), market bid (blue 
    triangles-down), and SEDEx implied volatilities (orange circles). \textit{(c)} 
    Market bid--ask implied volatilities and SEDEx implied volatilities (orange circles) 
    compared with the SVI benchmark fit (black circles). \textit{(d)} Deviations from 
    the bid--ask envelope for SEDEx (red triangles-up for ask, blue triangles-down for 
    bid) and SVI (orange triangles-up for ask, purple triangles-down for bid), reported 
    only when the fitted volatility lies outside the quoted interval. SEDEx errors are 
    sparse, all within $10^{-6}$ vol points.}
    \label{fig:market_1dte_calm}
\end{figure}
\paragraph{Second Case: ``Binary-Outcome Event Driven Day''}A more revealing case is the 1DTE slice observed on 2022-12-13, with 106 strikes and expiration on 2022-12-14. This date sits between the November CPI release on December 13 and the FOMC decision on December 14,\footnote{On December 13, 2022, the U.S.\ Bureau of Labor Statistics reported that CPI inflation for November 2022 was 7.1\% year-over-year; on December 14, 2022, the FOMC raised the target range for the federal funds rate by 50 basis points to 4.25--4.50\%.} and the option chain displays a pronounced W-shaped smile. The extracted density in \Cref{fig:market_1dte_wshape} is slightly bimodal, suggesting that the market attached non-negligible probability to several distinct terminal scenarios over the remaining one-day horizon. This interpretation should, however, be stated carefully. As emphasized by \citet{GlassermanPirjol2023}, the relation between smile shape and density shape is not one-to-one: a bimodal risk-neutral density need not generate a W-shaped smile, and a unimodal density can generate oscillatory or W-shaped smiles. Accordingly, our point is not that the bimodality of the extracted density explains the W-shape. Rather, the coexistence of a pronounced W-shape and a slightly bimodal extracted density is economically consistent with short-dated event risk driven by multiple plausible macro outcomes.

This reading is close to the empirical evidence of \citet{AlexiouEtAl2025}, who document that concave short-expiry implied-volatility curves often appear around scheduled events, are frequently associated with bimodal risk-neutral densities, and are linked to a premium for gamma.

The comparison with SVI is especially instructive in this case.\footnote{For this slice, the standard SVI calibration outputs parameters  $(a,b,\rho,m,\sigma)=(0.0005,\,0.0088,\,-1.0000,\,-0.0600,\,0.0186)$. The calibration pushes $\rho$ essentially to its boundary and $\sigma$ close to zero, which is a symptom of the tension between the observed shape and the flexibility of the specification.}
The SVI pa\-ra\-me\-te\-ri\-za\-tion produces more frequent violations of the bid--ask envelope and errors exceeding 2 volatility points, while the SEDEx smile remains admissible over the full chain, with a maximum error of $10^{-6}$ volatility points. Despite this non-standard geometry, the density-based interpolation and extrapolation remain stable and yield a plausible smile outside the quoted strike range. This in turn highlights the flexibility of the proposed methodology, which accommodates irregular short-dated market shapes without imposing a parametric structure on the smile.

\begin{figure}[H]
    \centering
    \begin{subfigure}[b]{0.48\textwidth}
        \centering
        \includegraphics[width=\textwidth]{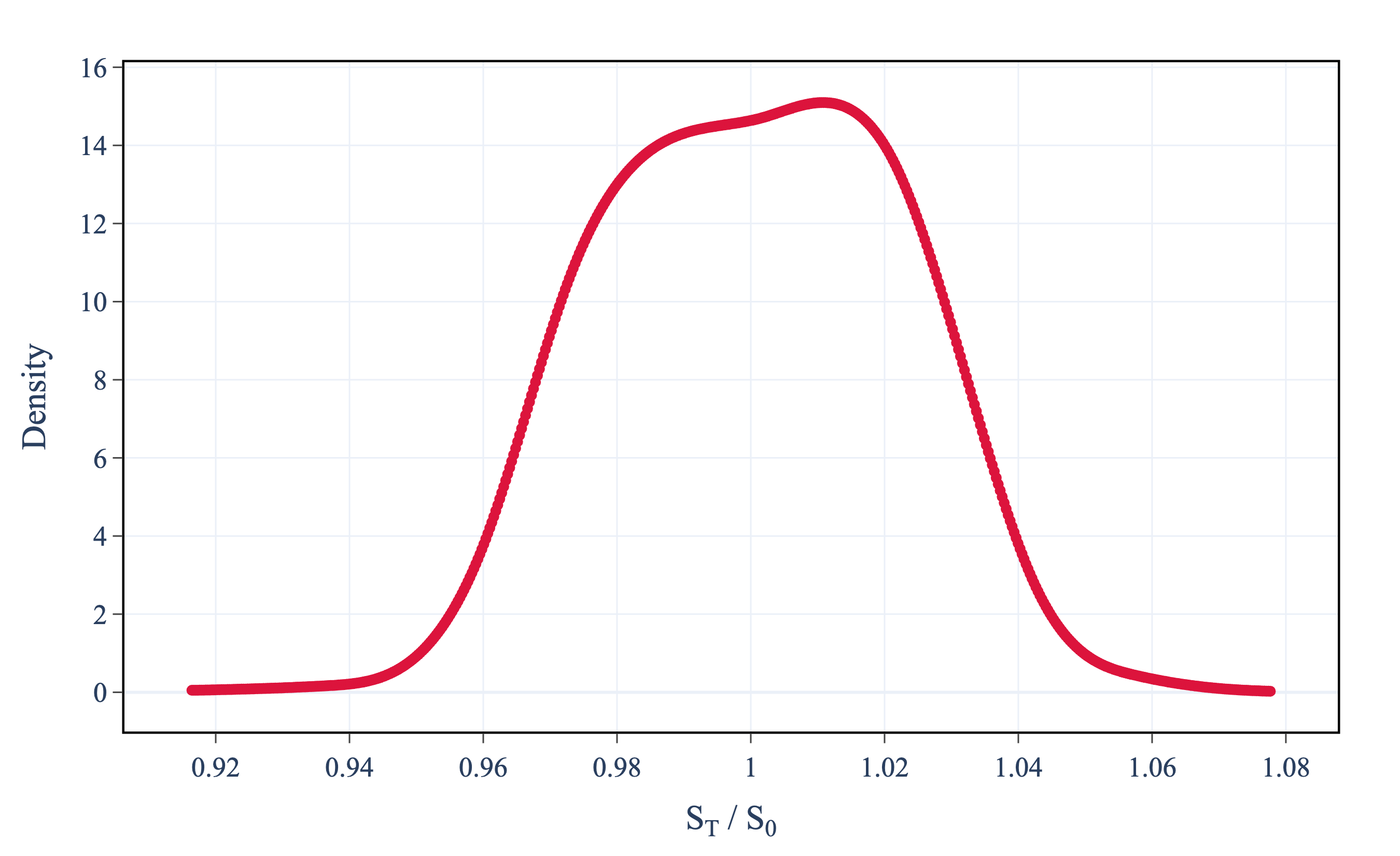}
        \caption*{(a)}
    \end{subfigure}
    \hfill
    \begin{subfigure}[b]{0.48\textwidth}
        \centering
        \includegraphics[width=\textwidth]{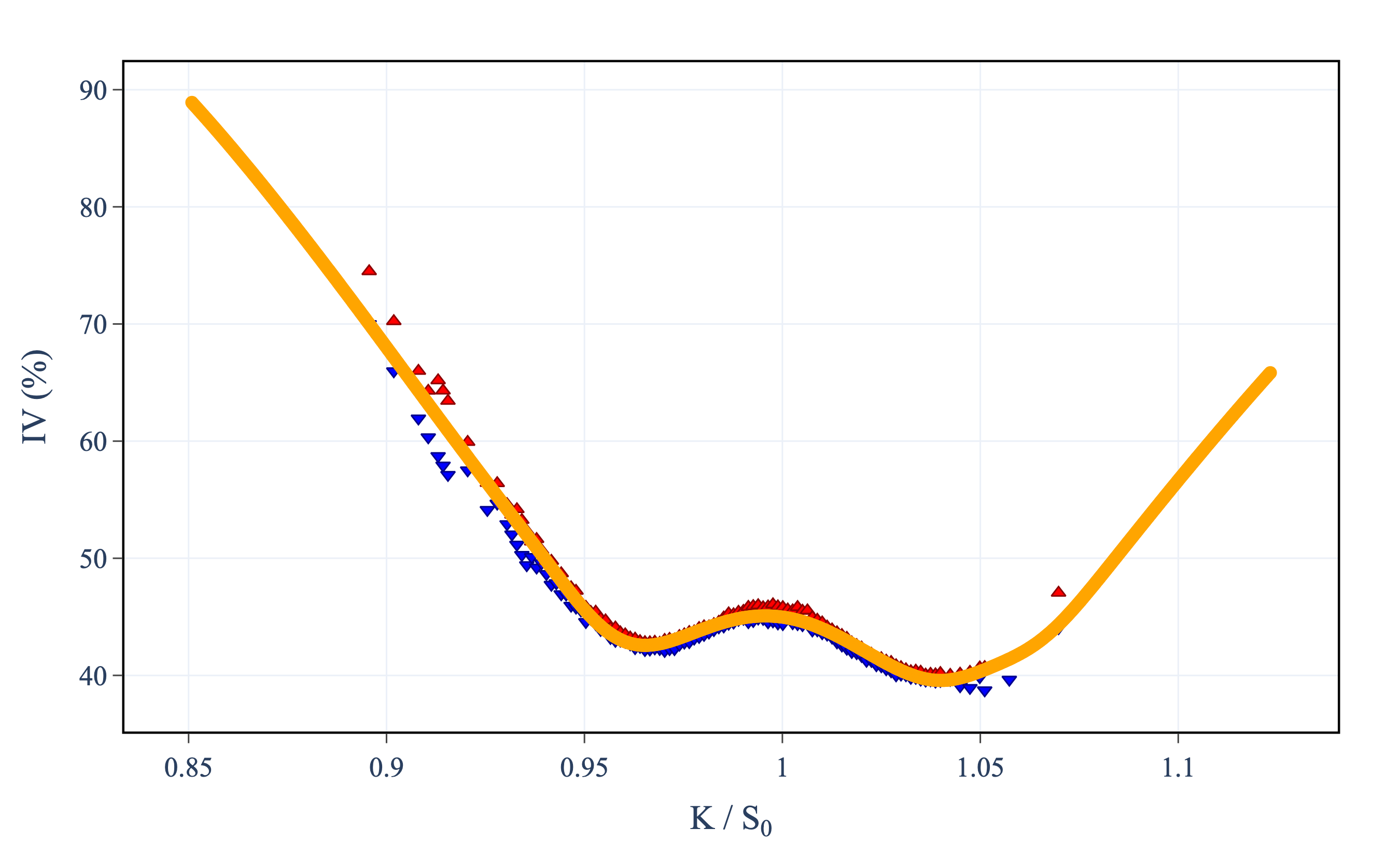}
        \caption*{(b)}
    \end{subfigure}
    \vspace{0.4em}
    \begin{subfigure}[b]{0.48\textwidth}
        \centering
        \includegraphics[width=\textwidth]{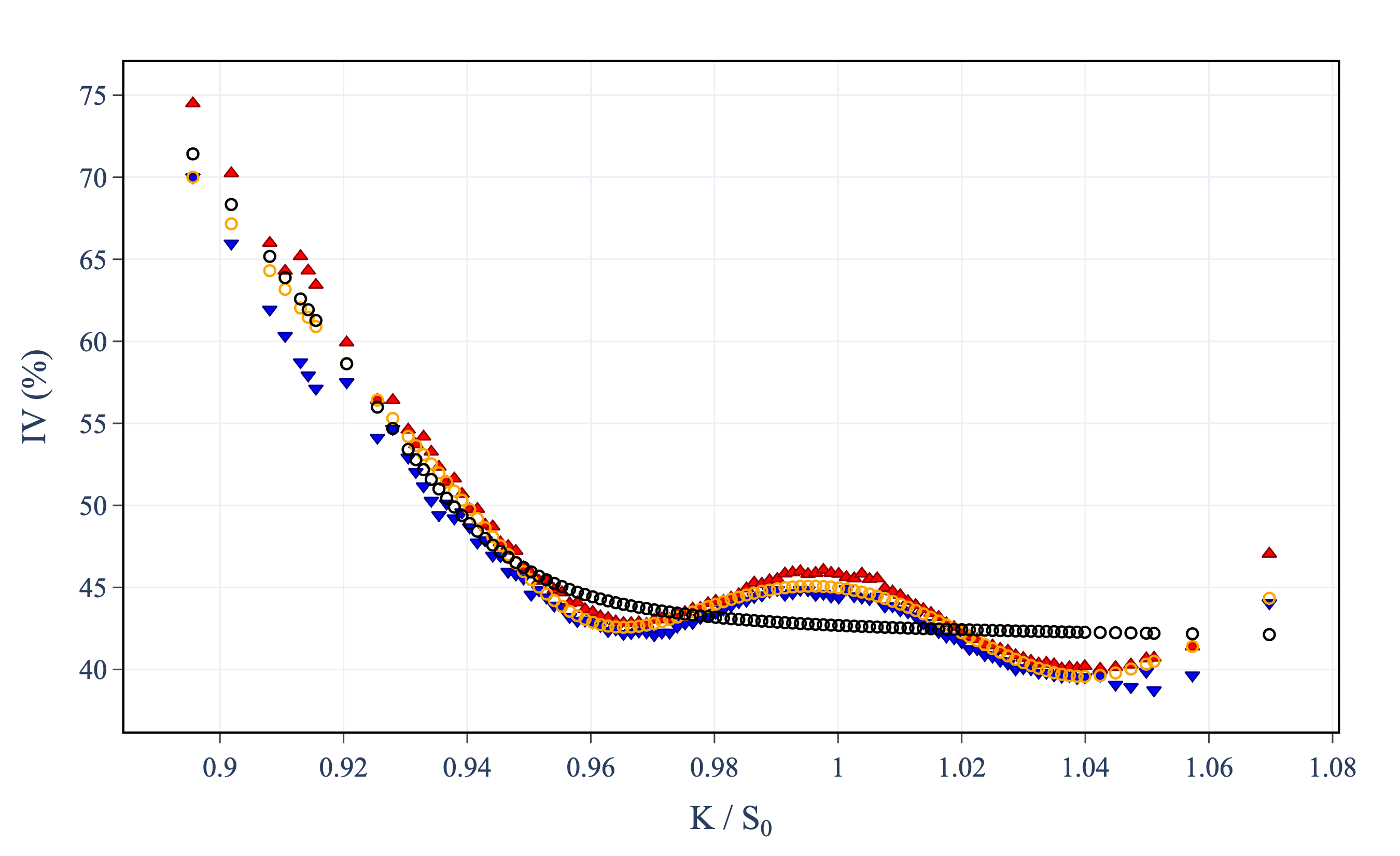}
        \caption*{(c)}
    \end{subfigure}
    \hfill
    \begin{subfigure}[b]{0.48\textwidth}
        \centering
        \includegraphics[width=\textwidth]{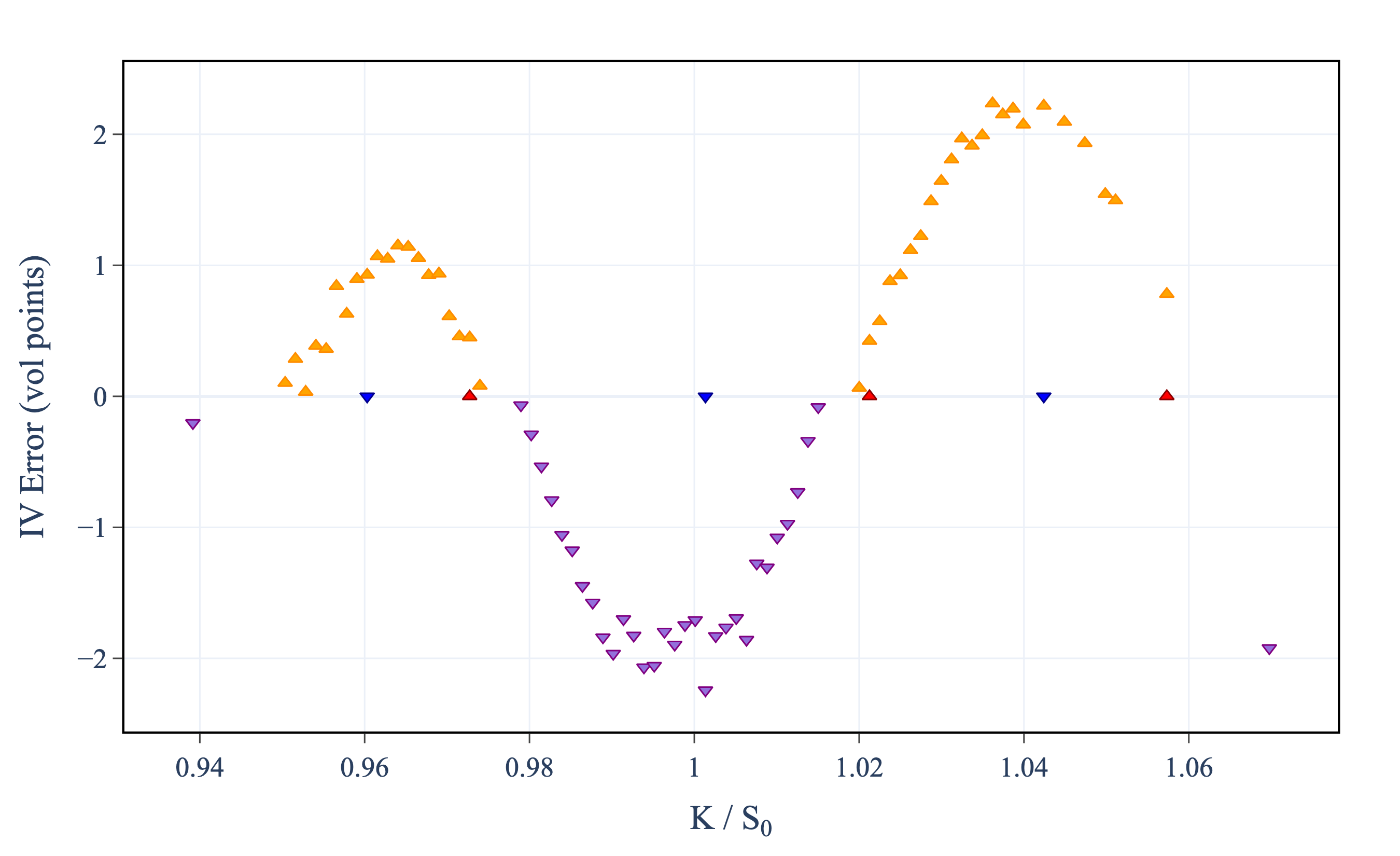}
        \caption*{(d)}
    \end{subfigure}
    \caption{Event-driven 1DTE slice on 2022-12-13. \textit{(a)} Risk-neutral density 
    extracted by SEDEx (red dots), which is slightly bimodal. \textit{(b)} Implied 
    volatility smile obtained by repricing options on a finer strike grid: market ask 
    (red triangles-up), market bid (blue triangles-down), and SEDEx implied volatilities 
    (orange circles). \textit{(c)} Market bid--ask implied volatilities and SEDEx 
    implied volatilities (orange circles) compared with the SVI benchmark fit (black 
    circles). \textit{(d)} Deviations from the bid--ask envelope for SEDEx (red 
    triangles-up for ask, blue triangles-down for bid) and SVI (orange triangles-up for 
    ask, purple triangles-down for bid), reported only when the fitted volatility lies 
    outside the quoted interval. SVI violations are frequent and of significant magnitude, 
    whereas SEDEx errors are sparse, all within $10^{-6}$ vol points.}
    \label{fig:market_1dte_wshape}
\end{figure}

\paragraph{Third Case: Very Short-Dated}We finally turn to 0DTE options observed on 2023-05-08. We consider two intraday timestamps, 10:30 and 15:00, corresponding to one hour after the open and one hour before expiration. The corresponding chains contain 44 and 14 strikes, respectively. \Cref{fig:market_0dte_intraday} shows that SEDEx remains stable even with few quoted strikes and only a few hours remaining to expiration. Two facts stand out. First, the support of the SEDEx density contracts dramatically over the day; by 15:00, most of the mass lies within roughly $\pm 0.5\%$ of spot, consistent with the concentration of the terminal distribution as expiration approaches. Second, the bid--ask envelope widens as maturity vanishes. This is precisely the regime in which procedures based solely on smooth parametric implied-volatility fits become fragile.

\begin{figure}[H]
    \centering
    \begin{subfigure}[b]{0.48\textwidth}
        \centering
        \includegraphics[width=\textwidth]{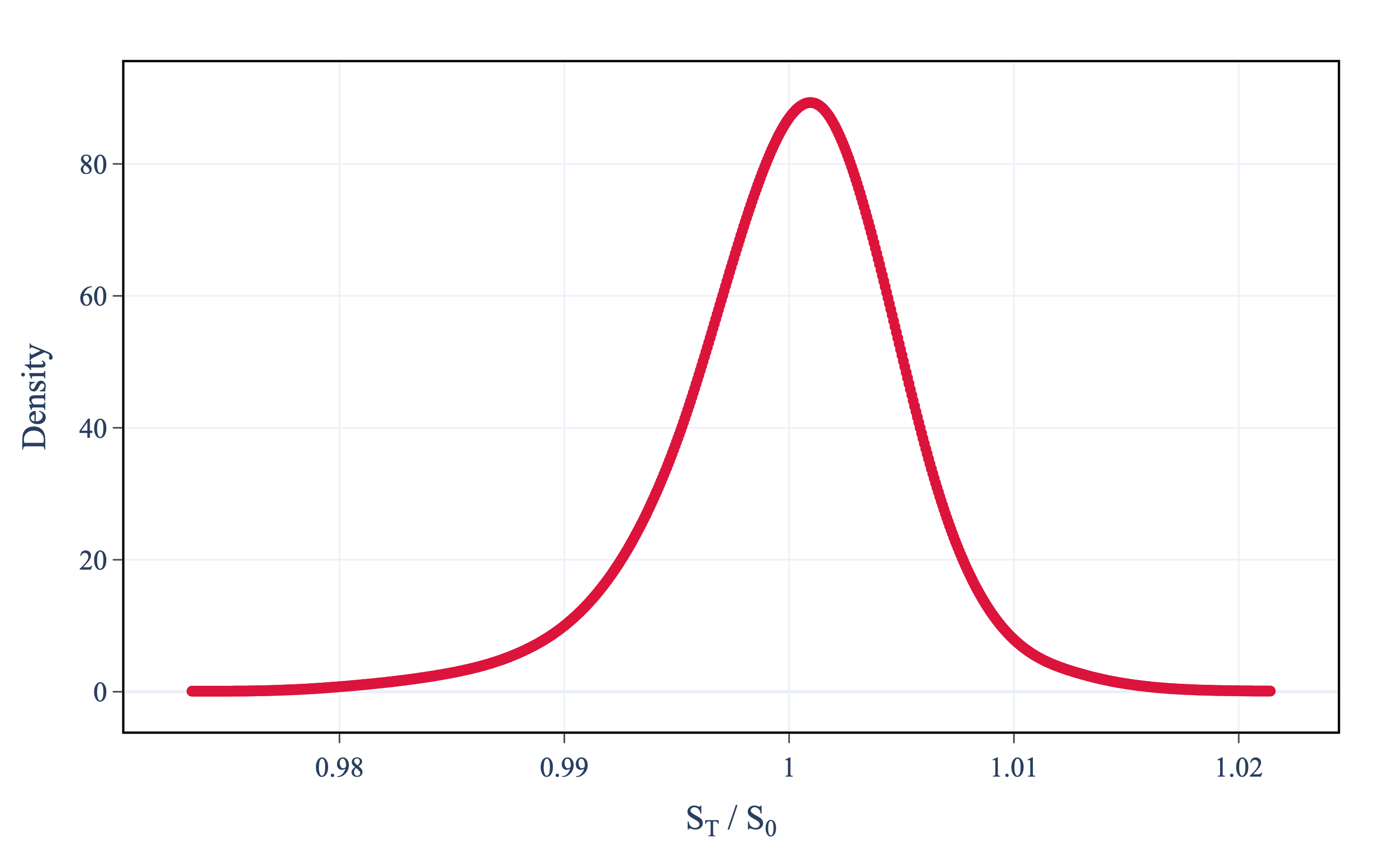}
        \caption*{(a)}
    \end{subfigure}
    \hfill
    \begin{subfigure}[b]{0.48\textwidth}
        \centering
        \includegraphics[width=\textwidth]{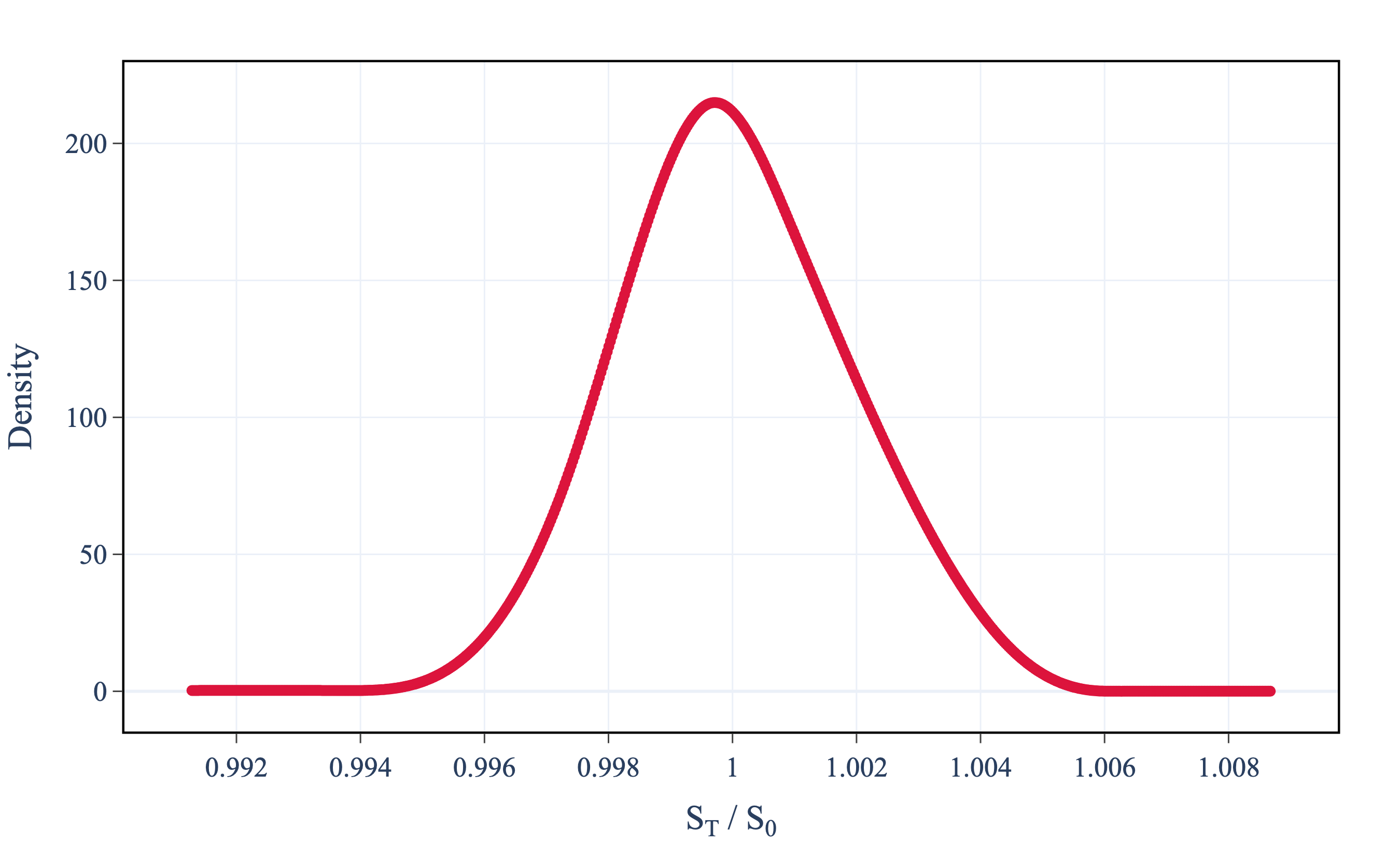}
        \caption*{(b)}
    \end{subfigure}
    \vspace{0.4em}
    \begin{subfigure}[b]{0.48\textwidth}
        \centering
        \includegraphics[width=\textwidth]{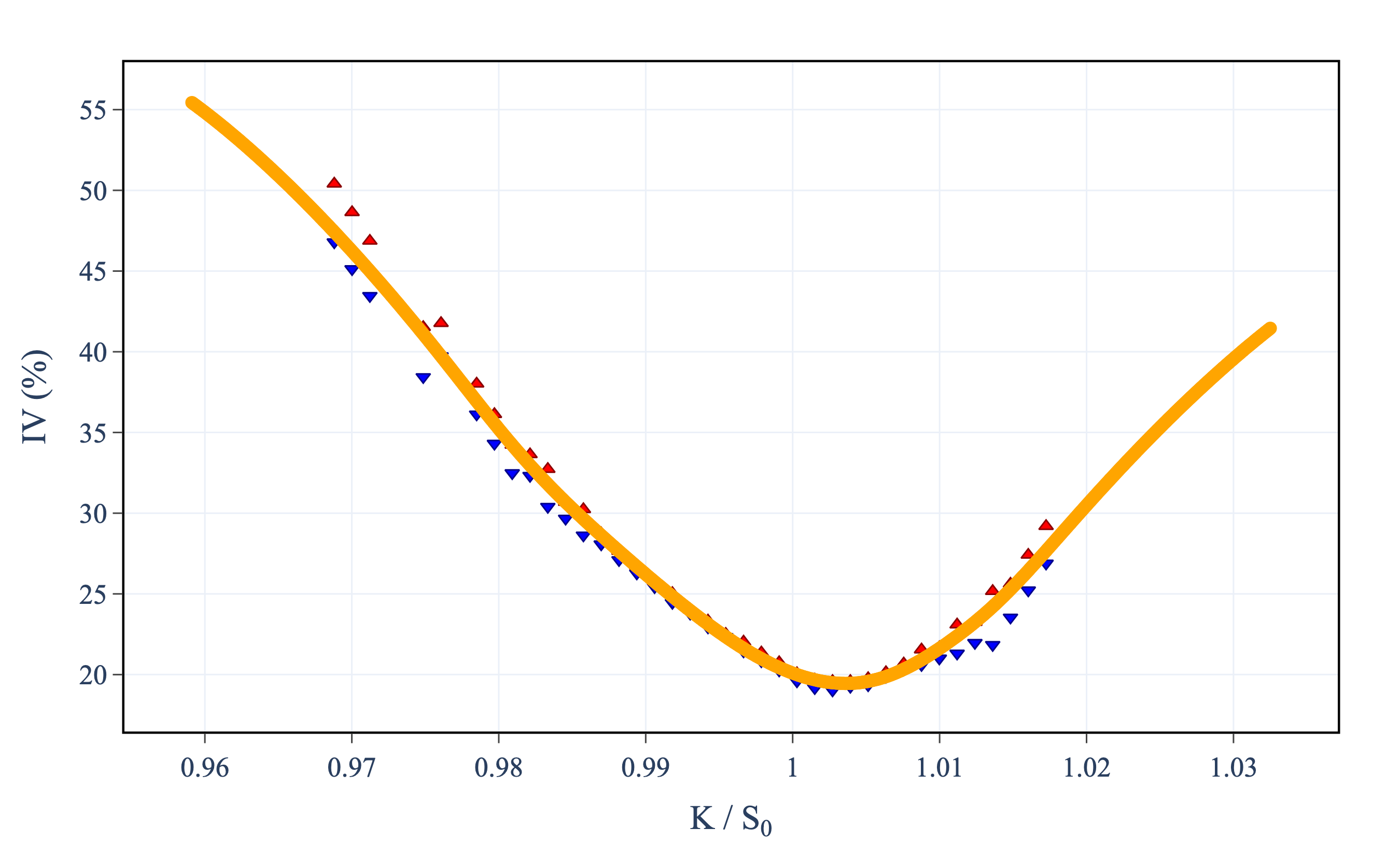}
        \caption*{(c)}
    \end{subfigure}
    \hfill
    \begin{subfigure}[b]{0.48\textwidth}
        \centering
        \includegraphics[width=\textwidth]{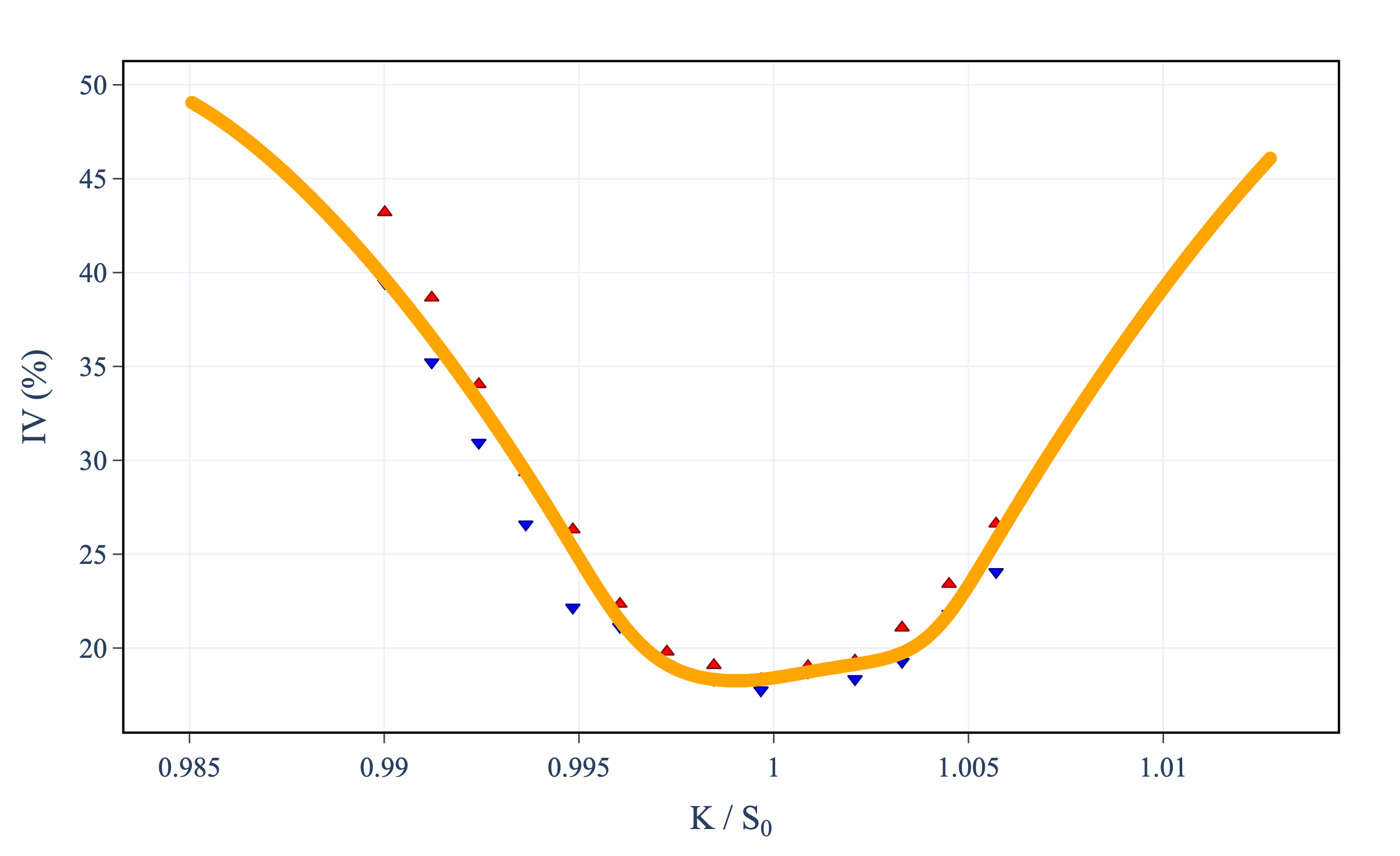}
        \caption*{(d)}
    \end{subfigure}
    \caption{0DTE intraday slices on 2023-05-08. \textit{(a)--(b)} Risk-neutral 
    densities extracted by SEDEx at 10:30 and 15:00 (red dots). \textit{(c)--(d)} 
    Corresponding implied volatility smiles: market ask (red triangles-up), market 
    bid (blue triangles-down), and SEDEx implied volatilities (orange circles). 
    The bid--ask constraints are satisfied up to a numerical tolerance of $10^{-6}$ 
    vol points.}
    \label{fig:market_0dte_intraday}
\end{figure}

\paragraph{Further Tests}To keep the main text focused, we defer a representative 7DTE slice to Appendix~\ref{appendix:7dte_market_slice}; see \Cref{fig:market_7dte_appendix}. The results are again satisfactory and indicate that the same pipeline extends smoothly from ultra-short-dated to short-dated maturities.

\medskip

Taken together, these results show that working directly with arbitrage-free bid--ask quotes provides a flexible approach to risk-neutral density recovery and implied-volatility smile construction.
\section{Summary and Conclusion}
\label{sec:conclusion}

This paper develops a model-free approach to risk-neutral density
extraction from short-dated option chains directly exploiting the
bid--ask quotes. It introduces two procedures, ARIES and SEDEx. ARIES removes
executable static arbitrages subject to
market-depth constraints. When several quotes can be removed, it
removes the quote with the smallest available size. This rule is
designed to preserve the most informative quotes. SEDEx recovers a density through a criterion leveraging smoothness and entropy under bid--ask constraints.

The two procedures are conceptually independent. Together, they form a
fast pipeline that keeps bid--ask intervals as the primitive market
constraint. Synthetic Heston experiments validate the two components in a
controlled setting. Applications to SPX option data show that the
pipeline remains fast and stable across different short-dated
settings, from a few hours to one week before expiry, including
scheduled-event days.

As an empirical application, the recovered densities are used to
construct implied-volatility smiles. These smiles remain
within the quoted bid--ask bounds. Extending the framework
to the joint treatment of the option surface across maturities is left
for future work.

\newpage
\bibliographystyle{plainnat}
\bibliography{biblio.bib}

\newpage
\appendix

\section{Detailed proofs}
\subsection{Lemmas}
\subsubsection{Equivalence of Discrete Admissible Sets}
\label{appendix:equi_admissible_sets}

\begin{proof}[Proof of \Cref{lemma:equivalence_adm_set}]
The proof relies on the algebraic identity relating the payoff of a put option to that
of a call option. Consider the fundamental identity:
\begin{equation}
(K - x)_{+} = (x - K)_{+} - (x - K), \quad \forall x, K \in \mathbb{R}.
\label{eq:payoff-identity-discrete}
\end{equation}
Let $p\in \mathcal{D}_M$. We evaluate the theoretical price of a put option against the discrete distribution $p$ and use the above identity:
\begin{align}
P^{p}(K_j^p)
&:= e^{-rT}\sum_{i=1}^M (K_j^p - s_i)_{+}\, p_i \notag\\
&= e^{-rT}\sum_{i=1}^M \left[ (s_i - K_j^p)_{+} - (s_i - K_j^p) \right] p_i
\quad \text{by \eqref{eq:payoff-identity-discrete}} \notag\\
&= e^{-rT}\sum_{i=1}^M (s_i - K_j^p)_{+}\, p_i
\;-\; e^{-rT}\left(\sum_{i=1}^M s_i p_i - K_j^p \sum_{i=1}^M p_i\right).
\label{eq:discrete-parity-derivation}
\end{align}
Using the simplex constraint $\sum_{i=1}^M p_i = 1$ and the forward constraint
$\sum_{i=1}^M s_i p_i = F_0^T$, this simplifies to the put--call parity relation:
\begin{equation}
P^{p}(K_j^p) = C^{p}(K_j^p) - e^{-rT}(F_0^T - K_j^p).
\label{eq:parity-relation-discrete}
\end{equation}

Now, consider the constraints defining $\mathcal{A}_{M}^{cp}$. A vector $p$ satisfies
the $j$-th put constraint if and only if
\begin{equation}
P_{j}^{\mathrm{bid}} \le P^{p}(K_j^p) \le P_{j}^{\mathrm{ask}}.
\label{eq:put-ineq-discrete}
\end{equation}
Substituting relation \eqref{eq:parity-relation-discrete} into \eqref{eq:put-ineq-discrete}
yields
\begin{equation}
P_{j}^{\mathrm{bid}} \le C^{p}(K_j^p) - e^{-rT}(F_0^T - K_j^p) \le P_{j}^{\mathrm{ask}}.
\end{equation}
Rearranging terms to isolate $C^{p}(K_j^p)$, we obtain
\begin{equation}
P_{j}^{\mathrm{bid}} + e^{-rT}(F_0^T - K_j^p)
\le C^{p}(K_j^p)
\le P_{j}^{\mathrm{ask}} + e^{-rT}(F_0^T - K_j^p).
\end{equation}
By definition of the transformed bounds in \eqref{eq:transformed-bounds-discrete},
this is equivalent to
\begin{equation*}
C^{p}(K_j^p)\in [\widetilde{C}_{j}^{\mathrm{bid}}, \widetilde{C}_{j}^{\mathrm{ask}}].
\end{equation*}
This demonstrates that the condition for the $j$-th put in $\mathcal{A}_{M}^{cp}$ is equivalent to the condition for the $j$-th synthetic call in $\mathcal{A}_{M}^{c}$.
Since the constraints on the original calls ($i=1,\dots,N_c$) and the moment conditions
defining $\mathcal{D}_M$ are identical in both sets, we conclude that
$\mathcal{A}_{M}^{cp} = \mathcal{A}_{M}^{c}$.
\end{proof}
\subsection{Propositions}
\subsubsection{Existence of a Solution to the Arbitrage Filtering Procedure}
\label{appendix:exist_opt_arb}
\begin{proof}[Proof of \Cref{prop:existence-optimum}]
The null portfolio $(q^{\mathrm{ask}}, q^{\mathrm{bid}}, u, \alpha)=0$ satisfies all inequalities, hence $\mathcal{C}\neq\emptyset$.
Moreover, note that $\mathcal{C}$ is a polyhedron as the intersection of finitely many half-spaces. We now show that $\mathcal{C}$ is bounded by verifying that each of its variables is bounded.

First, the quantities $q^{\mathrm{ask}}$ and $q^{\mathrm{bid}}$ are bounded due to the finite market sizes available, as specified in constraints~\eqref{eq:size-constraints-q-market}.
Second, by the definition of $\mathcal{C}$, $\alpha$ is bounded from above since $\alpha \le 0$.

Next, from the constraints in~\eqref{eq:constraints}, we have a lower bound for $u$:
\[
u \ge \sum_{i=1}^{N} (q_i^{\mathrm{bid}} - q_i^{\mathrm{ask}}).
\]
Since $q^{\mathrm{ask}}$ and $q^{\mathrm{bid}}$ are bounded, $u$ is bounded from below. 
Furthermore, since $\tilde{\pi}_0(q^{\mathrm{ask}}, q^{\mathrm{bid}}, u, \alpha) \le 0$, we can deduce an upper bound for $u$:
\[
u S_0 e^{-qT} \le \alpha - \sum_{i=1}^{N} \big(q_i^{\mathrm{ask}} C_i^{\mathrm{ask}} - q_i^{\mathrm{bid}} C_i^{\mathrm{bid}}\big).
\]
Given that $\alpha$ is bounded from above (by $0$) and the quantities $q$ are bounded, $u$ is also bounded from above.

Finally, using the same inequality $\tilde{\pi}_0 \le 0$, we obtain a lower bound for $\alpha$:
\[
\alpha \ge \sum_{i=1}^{N} \big(q_i^{\mathrm{ask}} C_i^{\mathrm{ask}} - q_i^{\mathrm{bid}} C_i^{\mathrm{bid}}\big) + u S_0 e^{-qT}.
\]
Since $u$ is now shown to be bounded from below and the quantities $q$ are bounded, $\alpha$ is bounded from below.

Since all variables are bounded, the polyhedron $\mathcal{C}$ is bounded and is therefore a polytope. 
Moreover, the objective function $\tilde{\pi}_0$ is linear. Because a polytope is a compact set, Weierstrass's theorem guarantees that the problem admits at least one optimal solution. Finally, by the Fundamental Theorem of Linear Programming \cite[Proposition B.20(c)]{bertsekas1997nonlinear}, an optimal solution is attained at least at a vertex of $\mathcal{C}$. This completes the proof.
\end{proof}
\subsubsection{Characterization of Optimal Solutions to the Arbitrage Filtering Procedure}
\label{appendix:charac_opt_arb}
\begin{proof}[Proof of \Cref{prop:characterization_proof_version}]
We proceed in five steps.

\paragraph{Step 1: Canonical map and interpretation.}
Given any feasible $x=(q^{\mathrm{ask}}, \allowbreak q^{\mathrm{bid}}, \allowbreak u, \allowbreak \alpha) \in \mathcal{C}$, define the induced net weights
$w$ by \eqref{eq:w-from-q-P1} and set the {cash position}
\begin{equation}
\label{eq:cash_from_alpha}
c:=-\alpha.
\end{equation}
Then the terminal payoff including cash equals
\begin{equation}
\label{eq:tildePi_is_payoff}
V_T(w)(s)+c e^{rT}
=
\sum_{i=1}^{N} (q_i^{\mathrm{ask}}-q_i^{\mathrm{bid}})(s-K_i)_+ + u s - \alpha e^{rT}
=\tilde{\Pi}(q^{\mathrm{ask}},q^{\mathrm{bid}},u,\alpha,s).
\end{equation}
Moreover, the associated time--$0$ {execution cost including cash} is
\begin{equation}
\label{eq:tildepi0_is_cost}
\pi_0(q^{\mathrm{ask}},q^{\mathrm{bid}},u)+c
=\pi_0(q^{\mathrm{ask}},q^{\mathrm{bid}},u)-\alpha
=\tilde{\pi}_0(q^{\mathrm{ask}},q^{\mathrm{bid}},u,\alpha).
\end{equation}
Hence, feasibility of $x$ is exactly the statement that the discounted payoff in \eqref{eq:tildePi_is_payoff}
is nonnegative for all $s\ge 0$, and the objective in \eqref{eq:P3prime} is precisely
\[
-\tilde{\pi}_0(x)=-(\pi_0(q^{\mathrm{ask}},q^{\mathrm{bid}},u)+c),
\]
i.e.\ the maximization of the {upfront cash inflow} subject to nonnegative payoff for all $s\ge 0$.

Finally, recall from \Cref{rem:V0-execution} that for a fixed net position $w$,
any round-trip trade with $q_i^{\mathrm{ask}}>0$ and $q_i^{\mathrm{bid}}>0$ at the same strike $K_i$
strictly increases the cost by $\min(q_i^{\mathrm{ask}},q_i^{\mathrm{bid}})(C_i^{\mathrm{ask}}-C_i^{\mathrm{bid}})\ge 0$
without changing the payoff. Therefore, any optimal solution can be chosen to be canonical in the sense that
\begin{equation}
\label{eq:canonical_no_roundtrip}
\forall i\in\{1,\dots,N\}:\quad q_i^{\mathrm{ask}}q_i^{\mathrm{bid}}=0.
\end{equation}
We henceforth work with such an optimal solution.

\paragraph{Step 2 : Homogeneity and maximal scaling.}
Fix any feasible $x\in\mathcal{C}$ and $\lambda\ge 0$.
By the positive homogeneity of $\tilde{\Pi}$ and $\tilde{\pi}_0$ (see \Cref{rem:homogeneity}),
\begin{equation}
\label{eq:homog_scaling}
\tilde{\Pi}(\lambda x,s)=\lambda \tilde{\Pi}(x,s),\qquad
\tilde{\pi}_0(\lambda x)=\lambda \tilde{\pi}_0(x).
\end{equation}
Hence all payoff constraints remain satisfied for $\lambda\ge0$.
The only restriction on scaling comes from the upper size bounds, so the maximal admissible scaling factor
is exactly the quantity $\bar\lambda(x)$ defined in \eqref{eq:lambda_bar}.
For every $\lambda\in[0,\bar\lambda(x)]$, the point $\lambda x$ belongs to $\mathcal{C}$ and the objective scales linearly:
\begin{equation}
\label{eq:objective_scales}
-\tilde{\pi}_0(\lambda x)=-\lambda\tilde{\pi}_0(x).
\end{equation}

\paragraph{Step 3 : Strong-arbitrage case: $\tilde{\pi}_0^*<0$.}
Assume $\tilde{\pi}_0^*<0$. Then the optimal objective value satisfies
\[
-\tilde{\pi}_0^*>0.
\]
Suppose, for contradiction, that no upper size bound is saturated at $x^*$, i.e.
$q_i^{\mathrm{ask}*}<Q_i^{\mathrm{ask}}$ and $q_i^{\mathrm{bid}*}<Q_i^{\mathrm{bid}}$ for all $i$.
Then $\bar\lambda(x^*)>1$, and by \eqref{eq:objective_scales} we would have
\[
-\tilde{\pi}_0(\bar\lambda(x^*)x^*)=-\bar\lambda(x^*) \tilde{\pi}_0^*>-\tilde{\pi}_0^*,
\]
contradicting optimality of $x^*$. Hence at least one depth constraint is saturated, as claimed.

We now show that $x^*$ generates a {strong arbitrage} in the sense of \Cref{def:arbitrage}.
Let $w^*$ be the net weights induced by $x^*$ via \eqref{eq:w-from-q-P1}.
Define the cash position {using the market budget constraint} as
\begin{equation}
\label{eq:cash_zero_cost_strong}
c^*:=-\pi_0(q^{\mathrm{ask}*},q^{\mathrm{bid}*},u^*),
\end{equation}
so that $V_0(w^*)+c^*=0$ (equivalently, \eqref{eq:C0zero} holds).
Using \eqref{eq:tildePi_is_payoff}--\eqref{eq:tildepi0_is_cost}, we can relate the corresponding terminal payoff to
$\tilde{\Pi}^*$:
\begin{align}
\label{eq:Pi_vs_tildePi_strong}
V_T(w^*)(s)+c^*e^{rT}
&=\tilde{\Pi}^*(s)+\bigl(\alpha^*-\pi_0(q^{\mathrm{ask}*},q^{\mathrm{bid}*},u^*)\bigr)e^{rT} \\
&=\tilde{\Pi}^*(s)-\tilde{\pi}_0^* e^{rT}. \notag
\end{align}
Since $x^*\in\mathcal{C}$, we have $\tilde{\Pi}^*(s)\ge 0$ for all $s\ge 0$.
Because $\tilde{\pi}_0^*<0$, the constant shift $(-\tilde{\pi}_0^*e^{rT})$ is strictly positive.
Therefore, \eqref{eq:Pi_vs_tildePi_strong} implies
\[
\forall s\ge 0:\quad V_T(w^*)(s)+c^*e^{rT} > 0,
\]
which is exactly the definition of a strong arbitrage \eqref{eq:strong-arbitrage}.

\paragraph{Step 4 : Weak-arbitrage case with $\tilde{\pi}_0^*=0$ and $\tilde{\Pi}^*\not\equiv 0$.}
Assume $\tilde{\pi}_0^*=0$ and $\tilde{\Pi}^*$ is not identically zero.
Then, by homogeneity \eqref{eq:homog_scaling}, for every $\lambda\in[0,\bar\lambda(x^*)]$, where
$\bar\lambda(x^*)$ is defined in \eqref{eq:lambda_bar}, we have $\lambda x^*\in\mathcal{C}$ and
\[
-\tilde{\pi}_0(\lambda x^*)=-\lambda\tilde{\pi}_0^*=0,
\]
so every $\lambda x^*$ attains the optimal objective value $0$.
Moreover, since $\tilde{\Pi}^*\not\equiv 0$, we have $x^*\neq 0$ and hence $\bar\lambda(x^*)>0$; the segment
$\{\lambda x^*:0\le\lambda\le\bar\lambda(x^*)\}$ is therefore non-trivial.
Finally, the extremal point $\bar\lambda(x^*)x^*$ lies on the boundary of $\mathcal{C}$ and thus
saturates at least one depth constraint.

To link $x^*$ to a weak arbitrage, implement the induced net weights $w^*$ and choose the cash position
\begin{equation}
\label{eq:cash_zero_cost_weak}
c^*:=-\pi_0(q^{\mathrm{ask}*},q^{\mathrm{bid}*},u^*).
\end{equation}
Because $\tilde{\pi}_0^*=0$, we have $\alpha^*=\pi_0(q^{\mathrm{ask}*},q^{\mathrm{bid}*},u^*)$, hence
$\alpha^*-\pi_0(\cdot)=0$ and \eqref{eq:Pi_vs_tildePi_strong} simplifies to
\begin{equation}
\label{eq:Pi_equals_tildePi_weak}
V_T(w^*)(s)+c^*e^{rT}=\tilde{\Pi}^*(s)\ge 0,\qquad \forall s\ge 0.
\end{equation}
Since $\tilde{\Pi}^*$ is a continuous piecewise–linear function on $\mathbb{R}_+$ with breakpoints in
$\{0,$ $K_1$,$\ldots,$ $K_N\}$ and $\tilde{\Pi}^*\not\equiv 0$, there exists a nonempty open interval
$I\subset\mathbb{R}_+$ such that $\tilde{\Pi}^*(s)>0$ for all $s\in I$.

Let $\mathbb{P}$ be any probability measure on $\R_+$ with a density (with respect to Lebesgue measure) that is strictly positive on $I$. Then we obtain
\[
\mathbb{P}\!\left(V_T(w^*)+c^*e^{rT}>0\right)
= \mathbb{P}\!\left(\tilde{\Pi}^*(S_T)>0\right)
\ge \mathbb{P}(S_T\in I)>0,
\]
which proves that $(w^*,c^*)$ yields a weak arbitrage in the sense of \Cref{eq:weak arbitrage}.

\paragraph{Step 5 : No-arbitrage case and uniqueness of the optimizer.}
First, assume that $\tilde{\Pi}^*(s)=0$ for all $s\ge 0$.
By \eqref{eq:tildePi_is_payoff}, this means
\begin{equation}
\label{eq:linear_indep_identity}
\sum_{i=1}^{N} (q_i^{\mathrm{ask}*}-q_i^{\mathrm{bid}*})(s-K_i)_+ + u^* s - \alpha^* e^{rT} \equiv 0
\quad\text{on }\mathbb{R}_+.
\end{equation}
Since the functions $\{(s-K_i)_{+},\, s,\, 1\}$ are linearly independent on $\mathbb{R}_{+}$, \eqref{eq:linear_indep_identity} yields
\[
q_i^{\mathrm{ask}*}-q_i^{\mathrm{bid}*}=0\ \ (i=1,\dots,N),\qquad u^*=0,\qquad \alpha^*=0.
\]
In particular, the induced net weights satisfy $w^*_{K_i}=0$ for all $i$.
Invoking the canonical optimality property \eqref{eq:canonical_no_roundtrip}, we further conclude that
$q_i^{\mathrm{ask}*}=q_i^{\mathrm{bid}*}=0$ for all $i$, hence $x^*=0$.
This proves uniqueness of the optimal solution under the canonical implementation convention.

Conversely, assume that the market is arbitrage-free in the sense of \Cref{def:arbitrage}
(item~\ref{item:non_arb}). If there existed a feasible $x\in\mathcal{C}$ such that
$\tilde{\pi}_0(x)=0$ and $\tilde{\Pi}(x,\cdot)\not\equiv 0$, then the construction in Step~4
would yield a weak arbitrage under any probability measure $\mathbb{P}$ whose law of $S_T$
charges every nonempty open interval, contradicting the arbitrage-free assumption.
Likewise, if there existed a feasible $x\in\mathcal{C}$ with $\tilde{\pi}_0(x)<0$, then Step~3
would yield a strong arbitrage, again a contradiction.
Therefore, arbitrage-free assumption forces $\tilde{\Pi}^*\equiv 0$ at the optimizer.

Finally, the three cases are mutually exclusive by definition:
either $\tilde{\pi}_0^*<0$, or $\tilde{\pi}_0^*=0$; in the latter case either $\tilde{\Pi}^*\equiv 0$ or not.
This completes the proof.
\end{proof}
\subsubsection{Global Monotonicity, Convexity, Lower Bounds on Call Prices.}
\label{appendix:global_prop}
\paragraph{From Local to Global Monotonicity}
Note that by \Cref{eq:monotonicity} it follows immediately that for all \(0 \le i < j \le N-1\) we have \(C_i > C_j\), that is, the call price sequence is globally strictly decreasing.

\paragraph{From Local to Global Convexity}

For \(0 \le i < j \le N\), we denote the discrete slope between \((K_{i},C_{i})\) and \((K_{j},C_{j})\) by
\begin{equation}
 \Delta_{i j}
 \;:=\;
 \frac{C_{i} - C_{j}}{K_{j} - K_{i}}.
 \label{eq:discrete-slope-ij}
\end{equation}
In particular, the local slopes between consecutive strikes are \(\Delta_{i,i+1}\), \(i=0,\dots,N-1\).

We will see now that conditions \eqref{eq:monotonicity} and \eqref{eq:discrete-convexity-local} imply global discrete convexity of the call prices
that is, the slopes \(\Delta_{i j}\) are decreasing in the strike. More precisely:
\begin{proposition}[Global convexity]
\label{prop:global-convexity}
Assume \eqref{eq:monotonicity} and \eqref{eq:discrete-convexity-local} hold; then: 
\begin{equation}
\forall i < j < k: \ 
 \frac{C_{i} - C_{j}}{K_{j} - K_{i}}
 -
 \frac{C_{j} - C_{k}}{K_{k} - K_{j}}
 \;>\; 0,
 \label{eq:discrete-convexity-global}
\end{equation}
\end{proposition}

\begin{proof}
Assume that, for all \(i\),
\begin{equation}
 \Delta_{i,i+1} > \Delta_{i+1,i+2} > 0,
 \qquad i = 0,\dots,N-2,
 \label{eq:local-slopes-decreasing}
\end{equation}
which encodes local discrete convexity and monotonicity. Fix \(i<j<k\). Using a telescoping argument, we can write
\begin{align}
 \Delta_{i j}
 &= \sum_{l=i}^{j-1} \omega^{(i,j)}_{l,l+1} \,\Delta_{l,l+1},
 \label{eq:convex-combination-Delta-ij}
 \\
 \Delta_{j k}
 &= \sum_{l=j}^{k-1} \omega^{(j,k)}_{l,l+1} \,\Delta_{l,l+1},
 \label{eq:convex-combination-Delta-jk}
\end{align}
where the weight \(\omega^{(i,j)}_{l,l+1}\) is given by
\begin{equation}
 \omega^{(i,j)}_{l,l+1}
 =
 \frac{K_{l+1} - K_{l}}{K_{j} - K_{i}}
 > 0,
 \qquad
 \sum_{l=i}^{j-1} \omega^{(i,j)}_{l,l+1} = 1,
\end{equation}
and similarly for \(\omega^{(j,k)}_{l,l+1}\). In other words, both \(\Delta_{i j}\) and \(\Delta_{j k}\) are convex combinations of the local slopes \(\Delta_{l,l+1}\) over the respective intervals.

By \eqref{eq:local-slopes-decreasing}, we have, on the one hand,
\begin{equation}
 \Delta_{l,l+1} > \Delta_{j-1,j} > \Delta_{j,j+1},
 \qquad l \in \{i,\dots,j-1\},
\end{equation}
and, on the other hand,
\begin{equation}
 \Delta_{l,l+1} < \Delta_{j,j+1},
 \qquad l \in \{j,\dots,k-1\}.
\end{equation}
Since convex combinations preserve these bounds, we obtain
\begin{equation}
 \Delta_{i j}
 \;>\;
 \Delta_{j-1,j}
 \;>\;
 \Delta_{j,j+1}
 \;>\;
 \Delta_{j k},
\end{equation}
which is exactly \eqref{eq:discrete-convexity-global}.
\end{proof}

\paragraph{Global Lower Bound on Call Prices}

We now derive global lower bounds for admissible call prices.

\begin{proposition}[Global lower bounds]
\label{prop:intrinsic-bounds}
Let \(C\) be an admissible call price vector in the sense of \Cref{def:admissible-call-vector}. Then, for all \(i \in \{1,\dots,N\}\),
\begin{equation}
 S_{0} e^{-qT} - K_{i} e^{-rT}
 \;<\;
 C_{i}
 \;<\;
 S_{0} e^{-qT}.
 \label{eq:global-intrinsic-bounds}
\end{equation}
\end{proposition}

\begin{proof}
The upper bound in \eqref{eq:global-intrinsic-bounds} is immediate: by monotonicity \eqref{eq:monotonicity} and \eqref{eq:call-strike-zero}, we have
\begin{equation}
 C_{i} < C_{0} = S_{0} e^{-qT}.
\end{equation}
\medskip
For the lower bound, fix \(i \in \{2,\dots,N\}\). Let
\begin{equation}
 \lambda := \frac{K_{1}}{K_{i}} \in (0,1),
 \qquad
 K_{1} = \lambda K_{i} + (1-\lambda) K_{0}.
\end{equation}
By convexity in the strike,
\begin{equation}
 C_{1} = C(K_{1})
 \;<\;
 \lambda C(K_{i}) + (1-\lambda) C(K_{0}),
\end{equation}
that is,
\begin{equation}
 C_{1}
 \;<\;
 \frac{K_{1}}{K_{i}} C_{i}
 +
 \frac{K_{i}-K_{1}}{K_{i}} C_{0}.
\end{equation}
Subtracting \(S_{0} e^{-qT}\) on both sides and using \(C_{0}=S_{0}e^{-qT}\), we obtain
\begin{equation}
 C_{1} - S_{0} e^{-qT}
 \;<\;
 \frac{K_{1}}{K_{i}} \big( C_{i} - S_{0} e^{-qT} \big).
\end{equation}
Hence
\begin{equation}
 \frac{C_{1} - S_{0} e^{-qT}}{K_{1}}
 \;<\;
 \frac{C_{i} - S_{0} e^{-qT}}{K_{i}}.
 \label{eq:slope-1-i}
\end{equation}
By the slope control on \([K_{0},K_{1}]\) stated in \eqref{eq:slope-K0-K1},
\begin{equation}
 \frac{C_{1} - S_{0} e^{-qT}}{K_{1}} > -e^{-rT}.
\end{equation}
Combining this with \eqref{eq:slope-1-i} yields
\begin{equation}
 \frac{C_{i} - S_{0} e^{-qT}}{K_{i}} > -e^{-rT},
\end{equation}
which is equivalent to
\begin{equation}
 C_{i} > S_{0} e^{-qT} - K_{i} e^{-rT}.
\end{equation}
This also holds for \(i=1\) by a direct application of \eqref{eq:slope-K0-K1}. The claim follows.
\end{proof}

\subsubsection{Existence of an Admissible Call-Price Vector}
\label{appendix:existence_call_price}
\begin{proof}[Proof of \Cref{prop:existence-admissible-C}]
The proof proceeds in six steps.
\paragraph{Step 1: Encoding the constraints as linear inequalities.}
We rewrite the conditions of Definition~\ref{def:admissible-call-vector} in a linear inequality form suitable for Motzkin’s theorem. This yields a system of mixed strict and non-strict inequalities:
\begin{equation}
    \label{eq:motzkin-primal-system}
    A_{\mathrm{mkt}}\,C \le b_{\mathrm{mkt}}, \qquad
A_{\mathrm{struct}}\,C < b_{\mathrm{struct}}, \qquad
C>0,
\end{equation}
where \(A_{\mathrm{mkt}}\) and \(b_{\mathrm{mkt}}\) capture the market bid--ask bounds, and \(A_{\mathrm{struct}}\) and \(b_{\mathrm{struct}}\) capture the structural shape constraints.

\begin{itemize}
\item \textbf{Structural constraints.}
\begin{itemize}
    \item \emph{Strict monotonicity.} For each \(i=1,\dots,N\), \(C_i<C_{i-1}\), equivalently \(C_{i}-C_{i-1}<0\). Recall that the difference matrix \(D^{(1)}\in\R^{N\times(N+1)}\) (\ref{eq:D1}) is given by
    \begin{equation}
        (D^{(1)}C)_i := C_{i}-C_{i-1},\qquad i=1,\dots,N,
    \end{equation}
    so that \(D^{(1)}C<0\) exactly encodes \eqref{eq:monotonicity}.

    \item \emph{Strict convexity.} For each \(i=1,\dots,N-1\),
    \[
        \frac{C_{i-1}-C_{i}}{K_{i}-K_{i-1}}
        \;>\;
        \frac{C_{i}-C_{i+1}}{K_{i+1}-K_{i}},
    \]
    which rearranges to
    \[
        -\Big(\alpha_{i-1,i}(C_{i-1}-C_{i})-\alpha_{i,i+1}(C_{i}-C_{i+1})\Big)<0,
    \]
    using \(\alpha_{a,b}\) from Definition~\ref{def:vs-bf}. Define \(D^{(2)}\in\R^{(N-1)\times(N+1)}\) by
    \[
        (D^{(2)}C)_i
        :=
        -\,\alpha_{i-1,i}\,C_{i-1}
        +(\alpha_{i-1,i}+\alpha_{i,i+1})\,C_i
        -\alpha_{i,i+1}\,C_{i+1},
        \qquad i=1,\dots,N-1,
    \]
    so that \(D^{(2)}C<0\) exactly encodes \eqref{eq:discrete-convexity-local}.

    \item \emph{Slope constraint on \([K_0,K_1]\).} The lower bound \eqref{eq:slope-K0-K1} can be written \((C_0-C_1)-K_1e^{-rT}<0\). Let \(s^\top:=(1,-1,0,\dots,0)\) and \(b_{\mathrm{SP1}}:=K_1e^{-rT}\). Then \(s^\top C<b_{\mathrm{SP1}}\) represents the slope constraint.
\end{itemize}

Collecting these, set
\begin{equation}
A_{\mathrm{struct}} :=
\begin{pmatrix}
D^{(1)}\\
D^{(2)}\\
s^\top
\end{pmatrix},
\qquad
b_{\mathrm{struct}} :=
\begin{pmatrix}
0\\
0\\
b_{\mathrm{SP1}}
\end{pmatrix}.
\label{eq:struct-block}
\end{equation}
so that \(A_{\mathrm{struct}}C<b_{\mathrm{struct}}\) compactly represents all strict inequalities \eqref{eq:monotonicity}--\eqref{eq:slope-K0-K1}.

\item \textbf{Market (bid--ask) constraints.}
Introduce vectors for the bid and ask quotes:
\[
C^{\mathrm{ask}}:=(C_0^{\mathrm{ask}},\dots,C_N^{\mathrm{ask}})^\top,
\qquad
C^{\mathrm{bid}}:=(C_0^{\mathrm{bid}},\dots,C_N^{\mathrm{bid}})^\top,
\]
with \(C_0^{\mathrm{ask}}=C_0^{\mathrm{bid}}=S_0e^{-qT}\) under assumption (Hyp-Spot). The bid--ask bounds \eqref{eq:bid-ask-bounds} can be written as two systems of weak inequalities: \(C\le C^{\mathrm{ask}}\) and \(-C\le -C^{\mathrm{bid}}\). Thus set
\begin{equation}
A_{\mathrm{mkt}} :=
\begin{pmatrix}
I_{N+1}\\
-I_{N+1}
\end{pmatrix},
\qquad
b_{\mathrm{mkt}} :=
\begin{pmatrix}
C^{\mathrm{ask}}\\
-\,C^{\mathrm{bid}}
\end{pmatrix}.
\label{eq:mkt-block}
\end{equation}
so that \(A_{\mathrm{mkt}}C\le b_{\mathrm{mkt}}\) corresponds exactly to the bid--ask constraints.

\item \textbf{Positivity.} We aim to find \(C>0\) as in the setup of \Cref{cor:motzkin-positive}.
\end{itemize}

In summary, the feasibility of an admissible call price vector \(C\) is equivalent to the linear system \eqref{eq:motzkin-primal-system}.

\paragraph{Step 2: Decomposing Motzkin certificate.}
Assume for contradiction that the system \eqref{eq:motzkin-primal-system} admits \emph{no} solution and apply \Cref{cor:motzkin-positive} with
\[
A:=A_{\mathrm{mkt}},\quad b:=b_{\mathrm{mkt}},
\qquad
B:=A_{\mathrm{struct}},\quad c:=b_{\mathrm{struct}}.
\]
It guarantees the existence of nonnegative multipliers $y\in\R^{2(N+1)}_+$ and $z\in\R^{(2N)}_+$ satisfying one of the three conditions listed in \Cref{cor:motzkin-positive}.
Introduce the aggregated dual vector
\begin{equation}
w:=A_{\mathrm{mkt}}^\top y + A_{\mathrm{struct}}^\top z \in \R^{N+1}.
\label{eq:def-w-motzkin}
\end{equation}
\noindent
\smallskip
\noindent\textbf{Block decompositions.}
Decompose \(y\) according to the bid--ask block \eqref{eq:mkt-block} as
\[
y=\bigl(y^{\mathrm{ask}},y^{\mathrm{bid}}\bigr),
\qquad
y^{\mathrm{ask}},y^{\mathrm{bid}}\in\R^{N+1}_+,
\]
so that, using \eqref{eq:mkt-block},
\begin{equation}
A_{\mathrm{mkt}}^\top y
=
y^{\mathrm{ask}}-y^{\mathrm{bid}}
=:w_{\mathrm{mkt}}
\in\R^{N+1}.
\label{eq:w-mkt-motzkin}
\end{equation}
Likewise, decompose \(z\) according to the structural block \eqref{eq:struct-block} as
\[
z=\bigl(z^{\mathrm{VS}},z^{\mathrm{BF}},z^{\mathrm{SP1}}\bigr),
\quad
z^{\mathrm{VS}}\in\R^N_+,\ 
z^{\mathrm{BF}}\in\R^{N-1}_+,\ 
z^{\mathrm{SP1}}\in\R_+,
\]
and define
\begin{equation}
w_{\mathrm{struct}}
:=
A_{\mathrm{struct}}^\top z
=
(D^{(1)})^\top z^{\mathrm{VS}} + (D^{(2)})^\top z^{\mathrm{BF}} + s\,z^{\mathrm{SP1}}.
\label{eq:w-struct-motzkin}
\end{equation}
Combining \eqref{eq:def-w-motzkin}--\eqref{eq:w-struct-motzkin} gives the key identity
\begin{equation}
w = w_{\mathrm{mkt}} + w_{\mathrm{struct}},
\qquad\text{or equivalently}\qquad
w_{\mathrm{mkt}} = w - w_{\mathrm{struct}}.
\label{eq:key-identity-motzkin}
\end{equation}

In the remainder of the proof, \eqref{eq:key-identity-motzkin} will be used in two complementary ways.
First, the left-hand side identifies the {executed} net quantities in the underlying and the calls and thus pins down the time--0 trading cost before accounting for the cash component.
The associated cash investment is $z^{\mathrm{SP1}}\,K_1e^{-rT}$ which delivers the cash flow \(z^{\mathrm{SP1}}K_1\) at maturity.
Note that \(w_{\mathrm{mkt}}\) may have mixed signs.
Second, the right-hand side decomposes the same executed position into a componentwise nonnegative exposure \(w\) (long-only in the underlying and calls) and a structural term \(-w_{\mathrm{struct}}\) which will be shown in the next step to be a nonnegative combination of elementary spreads (vertical and butterfly spreads) and the synthetic-put at strike \(K_1\).

\paragraph{Step 3: Computing the coordinates of \(w_{\mathrm{mkt}}\).}
We now expand the key identity \eqref{eq:key-identity-motzkin} componentwise. For each \(i=0,\dots,N\), the \(K_i\)-coordinate of \eqref{eq:key-identity-motzkin} is
\begin{equation}
(w_{\mathrm{mkt}})_{K_i} = w_{K_i} - (w_{\mathrm{struct}})_{K_i}.
\label{eq:wmkt-coordinate-identity}
\end{equation}
On the other hand, by \eqref{eq:w-mkt-motzkin} we have 
\begin{equation}
(w_{\mathrm{mkt}})_{K_i}=y^{\mathrm{ask}}_{K_i}-y^{\mathrm{bid}}_{K_i}.
\label{eq:wmkt-y-coordinate}
\end{equation}

\smallskip
For notational convenience, write \(z^{\mathrm{VS}}_{K_{i-1},K_i}:=(z^{\mathrm{VS}})_i\) for the quantity of \(\mathrm{VS}_{i-1,i}\), and
\(z^{\mathrm{BF}}_{K_{i-1},K_i,K_{i+1}}:=(z^{\mathrm{BF}})_i\) for the quantity of \(\mathrm{BF}_{i-1,i,i+1}\).
Moreover, \(z^{\mathrm{SP1}}\) is the quantity of the synthetic put \footnote{A synthetic put at strike \(K\) is the standard static replication obtained by combining one call at strike \(K\) with a short position in the underlying. Adding a cash position that pays \(K\) at maturity, its terminal payoff is \((K-S_T)_+\).} position at strike \(K_1\).

\bigskip
\noindent\textbf{Explicit coordinates.}
Using \eqref{eq:w-struct-motzkin} together with the explicit forms of \(D^{(1)}\), \(D^{(2)}\) and \(s\), we obtain the following identities for \(y^{\mathrm{ask}}-y^{\mathrm{bid}}\).

\medskip
\noindent\textbf{(a) Strike $K_0$.} At the lower boundary, only the first vertical-spread, the first butterfly, and the slope constraints contribute:
\begin{align}
(w_{\mathrm{mkt}})_{K_0}
&= w_{K_0} - (w_{\mathrm{struct}})_{K_0}\nonumber\\
&= w_{K_0}
+ z^{\mathrm{VS}}_{K_0,K_1}
+ \alpha_{0,1} z^{\mathrm{BF}}_{K_0,K_1,K_2}
- z^{\mathrm{SP1}}
= y^{\mathrm{ask}}_{K_0}-y^{\mathrm{bid}}_{K_0}.
\label{eq:wmkt-K0-motzkin}
\end{align}

\medskip
\noindent\textbf{(b) Strike $K_1$.} At the first interior node, only adjacent vertical-spread, the first two butterflies, and the slope constraints contribute:
\begin{align}
(w_{\mathrm{mkt}})_{K_1}
&= w_{K_1} - (w_{\mathrm{struct}})_{K_1}\nonumber\\
&= w_{K_1}
-\big(z^{\mathrm{VS}}_{K_0,K_1}-z^{\mathrm{VS}}_{K_1,K_2}\big)
-\big((\alpha_{0,1}+\alpha_{1,2})z^{\mathrm{BF}}_{K_0,K_1,K_2}-\alpha_{1,2}z^{\mathrm{BF}}_{K_1,K_2,K_3}\big)
+ z^{\mathrm{SP1}}\nonumber\\
&= y^{\mathrm{ask}}_{K_1}-y^{\mathrm{bid}}_{K_1}.
\label{eq:wmkt-K1-motzkin}
\end{align}

\medskip
\noindent\textbf{(c) Strike $K_i$ for $2\le i\le N-2$.} In the interior, only adjacent vertical-spread constraints and the three nearest butterfly constraints contribute:
\begin{align}
(w_{\mathrm{mkt}})_{K_i}
&= w_{K_i} - (w_{\mathrm{struct}})_{K_i}\nonumber\\
&= w_{K_i}
-\big(z^{\mathrm{VS}}_{K_{i-1},K_i}-z^{\mathrm{VS}}_{K_i,K_{i+1}}\big)\nonumber\\
&\quad
-\Big(
-\alpha_{i,i+1} z^{\mathrm{BF}}_{K_i,K_{i+1},K_{i+2}}
+(\alpha_{i-1,i}+\alpha_{i,i+1})z^{\mathrm{BF}}_{K_{i-1},K_i,K_{i+1}}
-\alpha_{i-1,i}z^{\mathrm{BF}}_{K_{i-2},K_{i-1},K_i}
\Big)\nonumber\\
&= y^{\mathrm{ask}}_{K_i}-y^{\mathrm{bid}}_{K_i}.
\label{eq:wmkt-Ki-motzkin}
\end{align}

\medskip
\noindent\textbf{(d) Strike $K_{N-1}$.} Near the upper boundary, only adjacent vertical-spread constraints and the last two butterfly constraints contribute:
\begin{align}
(w_{\mathrm{mkt}})_{K_{N-1}}
&= w_{K_{N-1}} - (w_{\mathrm{struct}})_{K_{N-1}}\nonumber\\
&= w_{K_{N-1}}
-\big(z^{\mathrm{VS}}_{K_{N-2},K_{N-1}}-z^{\mathrm{VS}}_{K_{N-1},K_N}\big)\nonumber\\
&\quad
-\Big(
(\alpha_{N-2,N-1}+\alpha_{N-1,N})\,z^{\mathrm{BF}}_{K_{N-2},K_{N-1},K_N}
-\alpha_{N-2,N-1}\,z^{\mathrm{BF}}_{K_{N-3},K_{N-2},K_{N-1}}
\Big)\nonumber\\
&= y^{\mathrm{ask}}_{K_{N-1}}-y^{\mathrm{bid}}_{K_{N-1}}.
\label{eq:wmkt-KNm1-motzkin}
\end{align}

\medskip
\noindent\textbf{(e) Strike $K_N$.} At the upper boundary, only the last vertical-spread constraint and the last butterfly constraint contribute:
\begin{align}
(w_{\mathrm{mkt}})_{K_N}
&= w_{K_N} - (w_{\mathrm{struct}})_{K_N}\nonumber\\
&= w_{K_N}
- z^{\mathrm{VS}}_{K_{N-1},K_N}
+ \alpha_{N-1,N}z^{\mathrm{BF}}_{K_{N-2},K_{N-1},K_N}
= y^{\mathrm{ask}}_{K_N}-y^{\mathrm{bid}}_{K_N}.
\label{eq:wmkt-KN-motzkin}
\end{align}

\paragraph{Step 4: Payoff analysis of the structural portfolio.}
We now make explicit the instrument decomposition of \(-w_{\mathrm{struct}}\) and derive the corresponding payoff inequality. Let \((e_1,\dots,e_{N+1})\) denote the canonical basis of \(\R^{N+1}\).

\smallskip
\noindent\textbf{Elementary spreads.}
For each adjacent pair \((K_{i-1},K_i)\), recall the vertical spread exposure and corresponding payoff
\[
\mathrm{VS}_{i-1,i}=e_{i}-e_{i+1},
\qquad
V_T(\mathrm{VS}_{i-1,i})=(S_T-K_{i-1})_+-(S_T-K_i)_+\ge 0.
\]
For each triple \((K_{i-1},K_i,K_{i+1})\), recall the butterfly exposure
\[
\mathrm{BF}_{i-1,i,i+1}
=\alpha_{i-1,i}e_{i}-(\alpha_{i-1,i}+\alpha_{i,i+1})e_{i+1}+\alpha_{i,i+1}e_{i+2},
\]
with terminal payoff
\[
V_T(\mathrm{BF}_{i-1,i,i+1})
=\alpha_{i-1,i}(S_T-K_{i-1})_+
-(\alpha_{i-1,i}+\alpha_{i,i+1})(S_T-K_i)_+
+\alpha_{i,i+1}(S_T-K_{i+1})_+
\ge 0.
\]
Moreover, \(-s=e_2-e_1\) is the position in the call at \(K_1\) and the underlying used in the standard synthetic-put replication. Adding a cash amount \(K_1\) at maturity yields the put payoff :
\[
V_T(e_2-e_1)+K_1=(K_1-S_T)_+\ge 0.
\]

\smallskip
\noindent\textbf{Identification with the structural matrices.}
By construction of \(D^{(1)}\), \(D^{(2)}\), and \(s\), we have the rowwise identities
\[
-\bigl(D^{(1)}_{i,\cdot}\bigr)^\top=\mathrm{VS}_{i-1,i},
\qquad
-\bigl(D^{(2)}_{i,\cdot}\bigr)^\top=\mathrm{BF}_{i-1,i,i+1},
\qquad
-s=e_2-e_1.
\]
Therefore, recalling \eqref{eq:w-struct-motzkin},
\begin{align}
-w_{\mathrm{struct}}
&=
\sum_{i=1}^{N} z^{\mathrm{VS}}_{K_{i-1},K_i}\,\mathrm{VS}_{i-1,i}
\;+\;
\sum_{i=1}^{N-1} z^{\mathrm{BF}}_{K_{i-1},K_i,K_{i+1}}\,\mathrm{BF}_{i-1,i,i+1}
\;+\;
z^{\mathrm{SP1}}(e_2-e_1).
\label{eq:wstruct-decomp-motzkin}
\end{align}
Since all coefficients \(z^{\mathrm{VS}},z^{\mathrm{BF}},z^{\mathrm{SP1}}\) are nonnegative and each term on the right-hand side has a nonnegative payoff at time \(T\) (after adding the cash \(K_1\) to the last term), we obtain the pointwise inequality
\begin{align}
V_T(-w_{\mathrm{struct}}) + z^{\mathrm{SP1}}K_1
=&
\sum_{i=1}^{N} z^{\mathrm{VS}}_{K_{i-1},K_i}\,V_T(\mathrm{VS}_{i-1,i})
+
\sum_{i=1}^{N-1} z^{\mathrm{BF}}_{K_{i-1},K_i,K_{i+1}}\,V_T(\mathrm{BF}_{i-1,i,i+1})\nonumber \\
&+
z^{\mathrm{SP1}}(K_1-S_T)_+
\ge 0,
\qquad \text{for all } S_T\ge 0.
\label{eq:struct-payoff-nonneg-motzkin}
\end{align}

\paragraph{Step 5: Combined payoff of the executed strategy.}
By \Cref{cor:motzkin-positive}, the aggregated vector \(w\) satisfies
\begin{equation}
w = A_{\mathrm{mkt}}^\top y + A_{\mathrm{struct}}^\top z \ge 0
\quad\text{(componentwise)}.
\label{eq:w-nonneg}
\end{equation}
Recall that we interpret \(w=(w_{K_0},\dots,w_{K_N})\) as a static portfolio holding \(w_{K_i}\) units of the call with strike \(K_i\); its terminal payoff is
\begin{equation}
V_T(w)=\sum_{i=0}^N w_{K_i}(S_T-K_i)_+.
\label{eq:w_motzkin_payoff}
\end{equation}
Since \(w\ge 0\) and \((S_T-K_i)_+\ge 0\) for all \(i\), we have
\begin{equation}
V_T(w)\ge 0
\qquad\text{for all } S_T\ge 0.
\label{eq:VT-w-nonneg}
\end{equation}
Using \eqref{eq:key-identity-motzkin} and the linearity of \(V_T\) yields
\[
V_T(w_{\mathrm{mkt}})=V_T(w-w_{\mathrm{struct}})
=V_T(w)+V_T(-w_{\mathrm{struct}}).
\]
Combining \eqref{eq:struct-payoff-nonneg-motzkin} with \eqref{eq:VT-w-nonneg} finally gives
\begin{equation}
V_T(w_{\mathrm{mkt}})+z^{\mathrm{SP1}}K_1 \ge 0
\qquad\text{for all } S_T\ge 0.
\label{eq:final-payoff-motzkin}
\end{equation}

\paragraph{Step 6: Inception cost analysis and conclusion.}
We turn the Motzkin certificate into an \emph{implementable} static strategy and conclude by comparing its time--$0$
\emph{value} $V_0(\cdot)$ (cf.\ \Cref{def:V0}) to its terminal payoff.

\medskip
\noindent\textbf{Step 6.1 (Scaling to satisfy depth constraints).}
Recall from Step~2 that we have a certificate $(y,z)$ with $y\ge 0$ and $z\ge 0$.
To enforce market depths at strikes $K_1,\dots,K_N$, we need
\begin{equation}
0\le y_i^{\mathrm{ask}}\le Q_i^{\mathrm{ask}},
\qquad
0\le y_i^{\mathrm{bid}}\le Q_i^{\mathrm{bid}},
\qquad i=1,\dots,N.
\label{eq:step6-depth-constraints}
\end{equation}
By positive homogeneity of \Cref{cor:motzkin-positive}, for any $\lambda>0$, the pair $(\lambda y,\lambda z)$ is again a valid certificate.
Choose
\begin{equation}
\lambda
:=
\min\Bigl(
1,\
\min_{1\le i\le N:\ y^{\mathrm{ask}}_{K_i}>0}\frac{Q_i^{\mathrm{ask}}}{y^{\mathrm{ask}}_{K_i}},\
\min_{1\le i\le N:\ y^{\mathrm{bid}}_{K_i}>0}\frac{Q_i^{\mathrm{bid}}}{y^{\mathrm{bid}}_{K_i}}
\Bigr)\in(0,1],
\label{eq:step6-lambda}
\end{equation}
and define the scaled certificate
\[
\tilde y:=\lambda y,\qquad \tilde z:=\lambda z.
\]
Then $\tilde y$ satisfies \eqref{eq:step6-depth-constraints}.

\medskip
\noindent\textbf{Step 6.2 (Executed trade at quotes, cash leg).}
Define the execution strategy $(q^{\mathrm{ask}},q^{\mathrm{bid}},u)$ by
\begin{equation}
q_i^{\mathrm{ask}}:=\tilde y^{\mathrm{ask}}_{K_i},\qquad
q_i^{\mathrm{bid}}:=\tilde y^{\mathrm{bid}}_{K_i},\qquad i=1,\dots,N,
\qquad
u:=\tilde y^{\mathrm{ask}}_{K_0}-\tilde y^{\mathrm{bid}}_{K_0}.
\label{eq:step6-execution-triple}
\end{equation}

\noindent
Define the scaled market and structural portfolios by
\begin{equation}
\tilde w_{\mathrm{mkt}}
:=A_{\mathrm{mkt}}^\top \tilde y
=\lambda w_{\mathrm{mkt}},
\qquad
\tilde w_{\mathrm{struct}}
:=A_{\mathrm{struct}}^\top \tilde z
=\lambda w_{\mathrm{struct}}.
\label{eq:step6-def-wmkt-wstruct-tilde}
\end{equation}
Moreover, define the scaled aggregated exposure
\begin{equation}
\tilde w
:=A_{\mathrm{mkt}}^\top \tilde y + A_{\mathrm{struct}}^\top \tilde z
=\lambda w.
\label{eq:step6-def-w-tilde}
\end{equation}
With this notation, the key identity \eqref{eq:key-identity-motzkin} becomes
$\tilde w_{\mathrm{mkt}}=\tilde w-\tilde w_{\mathrm{struct}},$
and $\tilde w_{\mathrm{mkt}}$ is the new induced net weights (\Cref{rem:w-from-q}) :
$
\tilde w_{\mathrm{mkt}}
=
\tilde y^{\mathrm{ask}}-\tilde y^{\mathrm{bid}}.
$
By \Cref{rem:V0-execution}, for this net position we have
\begin{equation}
V_0(\tilde w_{\mathrm{mkt}})\le \pi_0(q^{\mathrm{ask}},q^{\mathrm{bid}},u)=b_{\mathrm{mkt}}^\top \tilde y.
\label{eq:step6-V0-le-bmkt}
\end{equation}
\smallskip
Furthermore, by construction of $b_{\mathrm{struct}}$ \eqref{eq:struct-block}, its only nonzero entry corresponds to the slope constraint and equals $K_1e^{-rT}$. Hence
\begin{equation}
b_{\mathrm{struct}}^\top \tilde z
=
\tilde z^{\mathrm{SP1}}K_1e^{-rT}.
\label{eq:step6-bstruct-tilde}    
\end{equation}
This term is the associated cash deposit at time~$0$, which accrues at the risk-free rate and delivers
$\tilde z^{\mathrm{SP1}}K_1$ at maturity. Adding \eqref{eq:step6-bstruct-tilde} to \eqref{eq:step6-V0-le-bmkt} finally yields
\begin{equation}
V_0(\tilde w_{\mathrm{mkt}})+\tilde z^{\mathrm{SP1}}K_1e^{-rT}
\le
b_{\mathrm{mkt}}^\top \tilde y+b_{\mathrm{struct}}^\top \tilde z.
\label{eq:step6-V0pluscash-le-dual}
\end{equation}

\medskip
\noindent\textbf{Step 6.3 (Terminal payoff).}
Step~5 yields $V_T(w_{\mathrm{mkt}})+z^{\mathrm{SP1}}K_1\ge 0$ \eqref{eq:final-payoff-motzkin}.
Scaling by $\lambda$ and using linearity of $V_T$ gives
\begin{equation}
V_T(\tilde w_{\mathrm{mkt}})+\tilde z^{\mathrm{SP1}}K_1
=
\lambda\bigl(V_T(w_{\mathrm{mkt}})+z^{\mathrm{SP1}}K_1\bigr)
\ge 0
\qquad\text{for all }S_T\ge 0.
\label{eq:step6-terminal-nonneg}
\end{equation}

\medskip
\noindent\textbf{Step 6.4 (Conclusion : strong or weak arbitrage).}
Since the primal system \eqref{eq:motzkin-primal-system} is infeasible, the Motzkin certificate satisfies
$b_{\mathrm{mkt}}^\top y+b_{\mathrm{struct}}^\top z\le 0$, with strict inequality in the first alternative of item~2 in
\Cref{cor:motzkin-positive}. Multiplying by $\lambda>0$ yields
\begin{equation}
b_{\mathrm{mkt}}^\top \tilde y+b_{\mathrm{struct}}^\top \tilde z\le 0.
\label{eq:step6-dual-nonpos}
\end{equation}
Combining \eqref{eq:step6-V0pluscash-le-dual} and \eqref{eq:step6-dual-nonpos} gives
\begin{equation}
V_0(\tilde w_{\mathrm{mkt}})+\tilde z^{\mathrm{SP1}}K_1e^{-rT}\le 0.
\label{eq:step6-V0pluscash-nonpos}
\end{equation}

\smallskip
\noindent\emph{Case 1: $b_{\mathrm{mkt}}^\top \tilde y+b_{\mathrm{struct}}^\top \tilde z<0$.}
Then \eqref{eq:step6-V0pluscash-le-dual} implies
\[
V_0(\tilde w_{\mathrm{mkt}})+\tilde z^{\mathrm{SP1}}K_1e^{-rT}<0.
\]
Define
\begin{equation}
c:=-\Bigl(V_0(\tilde w_{\mathrm{mkt}})+\tilde z^{\mathrm{SP1}}K_1e^{-rT}\Bigr)>0.
\label{eq:step6-c-def}
\end{equation}
Consider the static strategy consisting of the option/underlying position $\tilde w_{\mathrm{mkt}}$ together with
the cash position $\tilde c:=c+\tilde z^{\mathrm{SP1}}K_1e^{-rT}$. By construction, its time--$0$ value is exactly zero.
At maturity, combining \eqref{eq:step6-terminal-nonneg} with \eqref{eq:step6-c-def} yields
\[
V_T(\tilde w_{\mathrm{mkt}})+\tilde z^{\mathrm{SP1}}K_1+c e^{rT}
\ge c e^{rT}>0
\qquad\text{for all }S_T\ge 0.
\]
Hence $(\tilde w_{\mathrm{mkt}},\tilde c)$ is a strong arbitrage, precluding the no-arbitrage condition of \Cref{def:arbitrage}.

\smallskip
\noindent\emph{Case 2: $b_{\mathrm{mkt}}^\top \tilde y+b_{\mathrm{struct}}^\top \tilde z=0$.}
Then \eqref{eq:step6-V0pluscash-nonpos} gives $V_0(\tilde w_{\mathrm{mkt}})\le 0$.
If $V_0(\tilde w_{\mathrm{mkt}})<0$, then we fall back to Case~1.
Thus we may assume
\begin{equation}
V_0(\tilde w_{\mathrm{mkt}})=0.
\label{eq:step6-V0-zero}
\end{equation}
Plugging \eqref{eq:step6-V0-zero} into \eqref{eq:step6-V0pluscash-nonpos} yields $\tilde z^{\mathrm{SP1}}K_1e^{-rT}\le 0$.
Since $\tilde z^{\mathrm{SP1}} \ge 0$, this implies $\tilde z^{\mathrm{SP1}}=0.$
Therefore \eqref{eq:step6-terminal-nonneg} reduces to $V_T(\tilde w_{\mathrm{mkt}})\ge 0 \label{eq:step6-terminal-nonneg-no-cash}$.\\

We now prove that $V_T(\tilde w_{\mathrm{mkt}})$ is not identically zero.
By \Cref{cor:motzkin-positive}, the infeasibility certificate cannot be trivial, i.e.
$
(\tilde z,\tilde w)\neq (0,0),
$
and in any case Step~5 yields $\tilde w\ge 0$.\\

\smallskip
If $\tilde w\neq 0$, then there exists $j$ such that $\tilde w_{K_j}>0$. Hence
\[
V_T(\tilde w)(S_T)
=\sum_{i=0}^N \tilde w_{K_i}(S_T-K_i)_+
\ge \tilde w_{K_j}(S_T-K_j)_+,
\]
so $V_T(\tilde w)(S_T)>0$ for all $S_T>K_j$. Using \eqref{eq:step6-def-wmkt-wstruct-tilde} and the linearity of $V_T$,
\[
V_T(-\tilde w_{\mathrm{struct}})
=
\lambda V_T(-w_{\mathrm{struct}})
\ge 0,
\]
where the inequality follows from \eqref{eq:struct-payoff-nonneg-motzkin}. Thus, we conclude that 
$V_T(\tilde w_{\mathrm{mkt}})=V_T(\tilde w)+V_T(-\tilde w_{\mathrm{struct}})>0$ is strictly positive on the interval $I:=(K_j,\infty)$.\\

\smallskip
If instead $\tilde w=0$ but $\tilde z\neq 0$, then $\tilde w_{\mathrm{mkt}}=-\tilde w_{\mathrm{struct}}$.
In addition, Step~4 (with $\tilde z^{\mathrm{SP1}}=0$) yields the following decomposition
\begin{equation}
-\,\tilde w_{\mathrm{struct}}
=
\sum_{i=1}^{N} \tilde z^{\mathrm{VS}}_{K_{i-1},K_i}\,\mathrm{VS}_{i-1,i}
+\sum_{i=1}^{N-1} \tilde z^{\mathrm{BF}}_{K_{i-1},K_i,K_{i+1}}\,\mathrm{BF}_{i-1,i,i+1}.
\label{eq:step6-wstruct-decomp-no-cash}
\end{equation}
where the coefficients are nonnegative and not all zero. For each elementary payoff $g\in\{V_T(\mathrm{VS}_{i-1,i}),\,V_T(\mathrm{BF}_{i-1,i,i+1})\}$ we have $g(S_T)\ge 0$ for all $S_T\ge 0$
and $g\not\equiv 0$. More precisely, pick an index with a strictly positive coefficient in \eqref{eq:step6-wstruct-decomp-no-cash}. For a vertical spread $\mathrm{VS}_{i-1,i}$, we have
$V_T(\mathrm{VS}_{i-1,i})(S_T)=K_i-K_{i-1}>0$ for all $S_T\ge K_i$. For a butterfly $\mathrm{BF}_{i-1,i,i+1}$, we have $V_T(\mathrm{BF}_{i-1,i,i+1})(S_T)>0$ for all $S_T\in(K_{i-1},K_{i+1})$.\\
In either case, there exists a nonempty open interval $I\subset \R_+$ on which the corresponding payoff is strictly positive.
Since all coefficients in \eqref{eq:step6-wstruct-decomp-no-cash} are nonnegative, it follows that
\begin{equation}
V_T(\tilde w_{\mathrm{mkt}})(S_T)>0
\qquad\text{for all }S_T\in I.
\label{eq:step6-positive-on-interval}
\end{equation}

\smallskip
So in both cases,
\Cref{eq:step6-positive-on-interval}
holds with $I$  a nonempty open interval of $R_+$.
Let $\mathbb{P}$ be any probability measure on $\R_+$ with a density (with respect to Lebesgue measure) that is strictly positive on $I$. Then we obtain
\[
V_T(\tilde w_{\mathrm{mkt}})\ge 0 \quad \mathbb{P}\text{-a.s.},
\qquad\text{and}\qquad
\mathbb{P}\bigl(V_T(\tilde w_{\mathrm{mkt}})>0\bigr)\ge \mathbb{P}(I)>0.
\]
Together with \eqref{eq:step6-V0-zero}, this shows that $(\tilde w_{\mathrm{mkt}},0)$ is a weak arbitrage with respect to such an absolutely continuous $\mathbb{P}$, contradicting item~\ref{item:non_arb}.

\medskip
\noindent\textbf{Conclusion.}
Both cases contradict no-arbitrage. Hence \eqref{eq:motzkin-primal-system} is feasible, and an admissible call price vector exists.

\end{proof}
\subsubsection{Replication of Call-Price Vector and Uniqueness Given the Support}
\label{appendix:proof_boundary_shift}
\begin{proof}[Proof of Proposition \ref{prop:boundary-shift}]
The proof is carried out in four steps. 
\paragraph{Step 1 (reproduction of call prices at \(K_1,\dots,K_N\)).}
Fix \(i\in\{1,\dots,N\}\).
Since \(K'<K_1\le K_i\), we have \((K'-K_i)_+=(K_1-K_i)_+=\cdots=(K_i-K_i)_+=0\), hence
\begin{equation}
\E^\nu[(S_T-K_i)_+]
=
\sum_{j=i+1}^{N+1} (K_j-K_i)\,q'_j.
\label{eq:nu-call-sum}
\end{equation}
Using \eqref{eq:qprime-upper} and \(q'_{N+1}=Q_N\), we obtain
\begin{align}
\E^\nu[(S_T-K_i)_+]
&= \sum_{j=i+1}^{N} (K_j-K_i)(Q_{j-1}-Q_j) + (K_{N+1}-K_i)Q_N \nonumber\\
&= \sum_{j=i+1}^{N} (K_j-K_i)Q_{j-1}
   -\sum_{j=i+1}^{N} (K_j-K_i)Q_j
   + (K_{N+1}-K_i)Q_N.
\end{align}
Re-indexing the first sum yields
\begin{align}
\E^\nu[(S_T-K_i)_+]
&= \sum_{j=i}^{N-1} (K_{j+1}-K_i)Q_j
   -\sum_{j=i+1}^{N} (K_j-K_i)Q_j
   + (K_{N+1}-K_i)Q_N \nonumber\\
&= (K_{i+1}-K_i)Q_i
   + \sum_{j=i+1}^{N-1} (K_{j+1}-K_j)Q_j
   + (K_{N+1}-K_N)Q_N.
\label{eq:nu-call-telescope}
\end{align}
By definition \eqref{eq:Qi-def}, for \(j=0,\dots,N-1\),
\(
(K_{j+1}-K_j)Q_j=e^{rT}(C_j-C_{j+1})
\),
and by \eqref{eq:CNplus1-def},
\(
(K_{N+1}-K_N)Q_N=e^{rT}C_N
\).
Plugging these identities into \eqref{eq:nu-call-telescope} yields
\[
\E^\nu[(S_T-K_i)_+]
=
e^{rT}\Big[(C_i-C_{i+1})+\sum_{j=i+1}^{N-1}(C_j-C_{j+1})+C_N\Big]
=
e^{rT}C_i,
\]
hence \(e^{-rT}\E^\nu[(S_T-K_i)_+]=C_i\) for all \(i=1,\dots,N\).

\paragraph{Step 2 (reproduction at \(K_0\)).}
Recall that \(K_0=0\) and \(C_0=S_0e^{-qT}\) in \Cref{def:admissible-call-vector}. Then \((S_T-K_0)_+=S_T\), so
\begin{equation}
\E^\nu[(S_T-K_0)_+]
=
\E^\nu[S_T]
=
q'_0K' + \sum_{j=1}^{N+1} K_j q'_j.
\label{eq:nu-forward-start}
\end{equation}
By definition of \((q'_0,q'_1)\) in \eqref{eq:qprime-lower}, we have
\begin{equation}
q'_0+q'_1=m,
\qquad
q'_0K' + q'_1K_1 = m\bar K.
\label{eq:leftblock-mass-moment}
\end{equation}
Moreover, multiplying \eqref{eq:Kbary} by \(m=1-Q_1\) yields the identity
\begin{equation}
m\bar K=(1-Q_0)K_0+(Q_0-Q_1)K_1.
\label{eq:mKbary-expand}
\end{equation}
Combining \eqref{eq:nu-forward-start}--\eqref{eq:mKbary-expand} with \eqref{eq:qprime-upper} gives
\begin{align}
\E^\nu[S_T]
&=
(1-Q_0)K_0
+\sum_{j=1}^{N} K_j (Q_{j-1}-Q_j)
+K_{N+1}Q_N.
\label{eq:nu-forward-canonical-form}
\end{align}
A straightforward telescoping computation rewrites the right-hand side as
\begin{align}
(1-Q_0)K_0
+\sum_{j=1}^{N} K_j (Q_{j-1}-Q_j)
+K_{N+1}Q_N
&=
K_0+\sum_{i=0}^{N} (K_{i+1}-K_i)Q_i.
\label{eq:nu-forward-telescope}
\end{align}
Finally, by \eqref{eq:Qi-def} and \(C_{N+1}=0\),
\[
\sum_{i=0}^{N} (K_{i+1}-K_i)Q_i
=
e^{rT}\sum_{i=0}^{N} (C_i-C_{i+1})
=
e^{rT}C_0.
\]
Since \(K_0=0\), \eqref{eq:nu-forward-telescope} implies \(\E^\nu[S_T]=e^{rT}C_0\), i.e.
\(e^{-rT}\E^\nu[(S_T-K_0)_+]=C_0\), completing \eqref{eq:nu-matches-C}.

\paragraph{Step 3 (probability condition).}
We first recall that \Cref{def:admissible-call-vector} implies 
\begin{equation}
0<Q_N<\cdots<Q_1<Q_0<1.
\label{eq:Qi-in-01}
\end{equation}

\emph{Positivity.}
For \(j=2,\dots,N\), \eqref{eq:Qi-in-01} gives \(Q_{j-1}>Q_j\), hence \(q'_j=Q_{j-1}-Q_j>0\).
Moreover, \(q'_{N+1}=Q_N>0\).
Next, since \(Q_1<1\), we have \(m=1-Q_1>0\), and since \(\bar K\) is a convex combination of \(K_0\) and \(K_1\), it follows that \(K_1-\bar K>0\).
Thus, for any \(K'\in(K_0,K_1)\), we therefore have
\[
q'_0
=
m\,\frac{K_1-\bar K}{K_1-K'}
>0.
\]
Finally,
\[
q'_1
=
m\,\frac{\bar K-K'}{K_1-K'}
\ge 0
\quad\Longleftrightarrow\quad
K'\le \bar K,
\]
and the strict condition \(K'<\bar K\) yields \(q'_1>0\).

\emph{Unit mass.}
By telescoping,
\[
\sum_{j=2}^{N+1} q'_j
=
\sum_{j=2}^{N} (Q_{j-1}-Q_j) + Q_N
=
Q_1,
\]
and by construction \(q'_0+q'_1=m=1-Q_1\).
Hence \(\sum_{j=0}^{N+1}q'_j=1\).

\paragraph{Step 4 (Uniqueness given the support).}
Fix \(\mathcal K'=\{K',K_1,\dots,K_{N+1}\}\).
Consider the linear system in the unknown \(q'\in\R^{N+2}\) formed by:
(i) the probability constraint \(\sum_{j=0}^{N+1}q'_j=1\);
(ii) the \(N+1\) pricing constraints \eqref{eq:nu-matches-C} for \(i=0,\dots,N\).
This yields a square system \(Aq'=b\) of dimension \((N+2)\times(N+2)\).
A direct computation
shows that
\begin{equation}
\det(A)=(K_1-K')\prod_{i=1}^{N}(K_{i+1}-K_i)\neq 0,
\end{equation}
since the strikes are distinct.
Therefore the solution \(q'\) is unique, and so is the corresponding probability measure \(\nu\) supported on \(\mathcal K'\).
\end{proof}
\subsection{Motzkin's Theorem of the Alternative}
\label{appendix:motzkin}
\begin{theorem}[Motzkin's theorem of the alternative, \cite{motzkin1936beitrage, abhyankar2002encyclopaedia}]
\label{thm:motzkin}
Let $A\in\R^{m\times n}$, $B\in\R^{\ell\times n}$, $b\in\R^m$, and $c\in\R^\ell$.
The following two statements are equivalent:
\begin{enumerate}
\item[\textbf{(P)}] The system $Ax\le b$ and $Bx<c$ admits a solution $x\in\R^n$.
\item[\textbf{(D)}] For all $y\in\R^m_+$ and $z\in\R^\ell_+$,
\[
A^\top y + B^\top z = 0 \ \Longrightarrow\ b^\top y + c^\top z \ge 0,
\]
and
\[
A^\top y + B^\top z = 0,\ z\neq 0 \ \Longrightarrow\ b^\top y + c^\top z > 0.
\]
\end{enumerate}
\end{theorem}

\subsection{Continuous Framework}
\label{subsec:cont_framework}
\subsubsection{Existence and Uniqueness of the Minimizer}
\label{subsec:exist-unique-cont-disc}
The following result ensures that $\mathcal{A}$ is non-empty, even when the bid--ask spread is null.
\begin{lemma}[Smooth $H^{1}$ density matching the constraints]
\label{lemma:H1-density-matching-calls}
Let $\nu$ be the discrete risk-neutral measure in \eqref{eq:nu} supported on
$\mathcal K'=\{K',K_1,\dots,K_{N+1}\}$, reproducing the call prices $C_i$ for
$i=1,\dots,N$ and preserving the forward price $F_0^T$.
Then there exists a function $f\in C_c^\infty(\mathcal{X})\subset H^{1}(\mathcal{X})$
such that
\[
f\ge 0\ ,\qquad
\int_{\mathcal{X}} f(s)\,ds = 1,\qquad
\int_{\mathcal{X}} s\,f(s)\,ds = F_0^T,
\]
and
\[
e^{-rT}\int_{\mathcal{X}} (s-K_i)_+\,f(s)\,ds = C_i,
\qquad i=1,\dots,N.\]
\end{lemma}

\begin{proof}
We establish the result through the following four steps.
\paragraph{Step 1: localized smooth atoms.}
Fix $\varphi\in C_c^\infty(-1,1)$ with $\varphi\ge 0$ and $\int_{\mathbb{R}}\varphi=1$,
and define $\varphi_\varepsilon(u):=\varepsilon^{-1}\varphi(u/\varepsilon)$.
Let $(x_0,\dots,x_{N+1}):=(K',K_1,\dots,K_{N+1})$ and set
\[
d:=\min\Bigl\{x_0,\ b-x_{N+1},\ \min_{0\le j\le N}(x_{j+1}-x_j)\Bigr\}>0,
\qquad
\bar\varepsilon:=\frac{d}{4}.
\]
For $\varepsilon\in(0,\bar\varepsilon)$ and $j=0,\dots,N+1$, define
\[
g_{j,\varepsilon}(s):=\varphi_\varepsilon(s-x_j).
\]
Then $g_{j,\varepsilon}\in C_c^\infty(\mathcal X)$, $g_{j,\varepsilon}\ge 0$,
$\int_{\mathcal X} g_{j,\varepsilon}=1$, and the supports of
$(g_{j,\varepsilon})$
are pairwise disjoint and contained in $\mathcal X$.

\paragraph{Step 2: finite-dimensional parameterization.}
For $a=(a_0,\dots,a_{N+1})^\top\in\mathbb{R}^{N+2}$, define
\[
f_{\varepsilon,a}(s):=\sum_{j=0}^{N+1} a_j\, g_{j,\varepsilon}(s)\in C_c^\infty(\mathcal X).
\]
Clearly $f_{\varepsilon,a}\ge 0$ a.e.\ whenever $a_j\ge 0$ for all $j$.
Introduce the moments
\[
L^{(\varepsilon)}_j:=\int_{\mathcal{X}} s\,g_{j,\varepsilon}(s)\,ds,
\qquad
M^{(\varepsilon)}_{i,j}:=\int_{\mathcal{X}} (s-K_i)_+\,g_{j,\varepsilon}(s)\,ds,
\]
for $j=0,\dots,N+1$ and $i=1,\dots,N$. The constraints
\[
\int_{\mathcal{X}} f_{\varepsilon,a}=1,\qquad
\int_{\mathcal{X}} s f_{\varepsilon,a}(s)\,ds = F_0^T,\qquad
\int_{\mathcal{X}} (s-K_i)_+ f_{\varepsilon,a}(s)\,ds = e^{rT}C_i
\ (i=1,\dots,N)
\]
are equivalent to the linear system
\[
A_\varepsilon\, a=b,
\]
with
\[
A_\varepsilon=
\begin{pmatrix}
1 & \cdots & 1\\
L^{(\varepsilon)}_0 & \cdots & L^{(\varepsilon)}_{N+1}\\
M^{(\varepsilon)}_{1,0} & \cdots & M^{(\varepsilon)}_{1,N+1}\\
\vdots & & \vdots\\
M^{(\varepsilon)}_{N,0} & \cdots & M^{(\varepsilon)}_{N,N+1}
\end{pmatrix},\qquad
b=(1,F_0^T,e^{rT}C_1,\dots,e^{rT}C_N)^\top.
\]

\paragraph{Step 3: identification of the limit matrix and invertibility.}
Because $g_{j,\varepsilon}$ is a mollifier centered at $x_j$,
we have for each fixed $(i,j)$,
\[
L^{(\varepsilon)}_j \to x_j,
\qquad
M^{(\varepsilon)}_{i,j}\to (x_j-K_i)_+.
\]
Hence $A_\varepsilon\to A_0$ entrywise (thus also in any matrix norm), where
\[
A_{0}:=
\begin{pmatrix}
1 & 1 & 1 & \cdots & 1\\
x_{0} & x_{1} & x_{2} & \cdots & x_{N+1}\\
0 & 0 & x_{2}-K_{1} & \cdots & x_{N+1}-K_{1}\\
0 & 0 & 0 & \ddots & \vdots\\
\vdots & \vdots & \vdots & \ddots & x_{N+1}-K_{N}
\end{pmatrix}.
\]
The matrix $A_0$ is block upper-triangular, with a $2\times 2$ leading block
$\begin{psmallmatrix}1&1\\ x_0&x_1\end{psmallmatrix}$ and an upper-triangular
$(N\times N)$ trailing block whose diagonal entries are $(x_{i+1}-K_i)_{i=1}^N$.
A direct Laplace expansion yields
\[
\det(A_0)=\prod_{j=0}^{N}(x_{j+1}-x_j)>0,
\]
since $x_0<x_1<\cdots<x_{N+1}$ and $x_i=K_i$ for $i=1,\dots,N+1$. Therefore $A_0$ is invertible.

By construction of $\nu$, its weight vector
$q'=(q'_0,\dots,q'_{N+1})^\top$ satisfies the same constraints on $\mathcal K'$,
hence
\[
A_0 q' = b.
\]

\paragraph{Step 4: stability of the inverse and positivity.}
Since $A_\varepsilon\to A_0$ and $\det(A_0)\neq 0$, continuity of the determinant
implies $\det(A_\varepsilon)\neq 0$ for all $\varepsilon$ small enough.
Thus $A_\varepsilon$ is invertible and the unique solution is
\[
a(\varepsilon):=A_\varepsilon^{-1}b.
\]
Moreover, the inverse map is continuous on $\mathrm{GL}({N+2},\mathbb{R})$, so
\[
a(\varepsilon)\to A_0^{-1}b = q'.
\]
Since $q'_j>0$ for all $j$, for $\varepsilon$ small
enough one has $a_j(\varepsilon)>0$ for all $j$.

\paragraph{Conclusion.}
Take $\varepsilon>0$ sufficiently small so that $A_\varepsilon$ is invertible and the unique solution
$a_\varepsilon:=A_\varepsilon^{-1}b$ is strictly positive componentwise.
Defining $f:=f_{\varepsilon,a_\varepsilon}=\sum_{j=0}^{N+1}(a_\varepsilon)_j\,g_{j,\varepsilon}$, we obtain a function
$f\in C_c^\infty(\mathcal X)\subset H^1(\mathcal X)$ with $f\ge 0$.
Finally, since $A_\varepsilon a_\varepsilon=b$ encodes exactly the mass constraint, the forward constraint, and the $N$
call constraints, $f$ satisfies all required equalities, which concludes the proof.
\end{proof}

We now show that $H$ is well-posed on $\mathcal A$ and admits a unique minimizer.
\begin{theorem}
\label{thm:exist-unique-continuous-minimizer}
The functional $H$ admits a unique minimizer $f^\ast$ over $\mathcal{A}$.
\end{theorem}

\begin{proof}
We proceed in four steps.
\noindent
\paragraph{Step 1: existence of a minimizing sequence and bounds in $H^1(\mathcal{X})$.}
Let $(f_k)_{k\geq 1}$ be a minimizing sequence in $\mathcal{A}$, i.e.
\begin{equation}
  H(f_k) \longrightarrow \inf_{f\in \mathcal{A}} H(f)
  \quad \text{as } k\to\infty.
  \label{eq:minimizing-sequence}
\end{equation}
Since $(f_k)$ is minimizing, there exists $C>0$ such that
\begin{equation}
  H(f_k) \leq C, \qquad \forall k\geq 1.
  \label{eq:H-bounded}
\end{equation}

We first derive a uniform upper bound on the entropy term $S(f_k)$.
By construction, each $f_k\in \mathcal{A}$ is a probability density on $\mathcal{X}$ such
that
\begin{equation}
  \int_\mathcal{X} x f_k(x)\,dx = F_0^T.
  \label{eq:forward-constraint}
\end{equation}
It is a classical result in information theory that the (Shannon) entropy
of a probability density on $\R_+$ with prescribed mean $F_0^T$ is
maximized by the exponential density with parameter $1/F_0^T$, and the
maximal entropy is equal to $1 + \ln(F_0^T)$. Since our densities are supported in $(0,b)\subset
\R_+$ and satisfy the same mean constraint, we obtain the uniform bound
\begin{equation}
  S(f_k) \leq 1 + \ln(F_0^T), \qquad \forall k\geq 1. \footnote{Recall that the Kullback--Leibler divergence is always non-negative : for any densities $f$ and $g$,
$\int_{0}^{+\infty} f(x) \ln\!\left( \frac{f(x)}{g(x)} \right)\,dx \ge 0$. 
Choose $g$ to be the density of an exponential distribution $\mathcal{E}(1/F_0^{T})$, which implies
$S(f_k) \le -\int_{0}^{+\infty} f_k(x)\ln g(x)\,dx$. 
Now, on $\mathcal{X}$ we have
\begin{align*}
-\int_{0}^{b} f_k(x)\ln g(x)\,dx
&= - \int_{0}^{b} f_k(x)
      \left(
        -\ln(F_0^{T}) - \frac{x}{F_0^{T}}
      \right)
      dx \\
&= \left( \int_{0}^{b} f_k(x)\,dx \right)\ln(F_0^{T})
   + \frac{1}{F_0^{T}} \int_{0}^{b} x f_k(x)\,dx .
\end{align*}

Since each $f_k$ is a probability density on $\mathcal{X}$ with
expectation $F_0^{T}$, we obtain
\[
S(f_k) \le 1 + \ln(F_0^{T}).
\]}
  \label{eq:entropy-upper-bound}
\end{equation}

From the definition \eqref{eq:def-H} we have
\begin{equation}
  H(f_k)
  = \lambda_1 \lVert f_k'\rVert_{L^2(\mathcal{X})}^2 - \lambda_2 S(f_k),
  \qquad \forall k\geq 1,
  \label{eq:H-explicit}
\end{equation}
so that
\begin{equation}
  \lVert f_k'\rVert_{L^2(\mathcal{X})}^2
  = \frac{1}{\lambda_1}\bigl( H(f_k) + \lambda_2 S(f_k)\bigr)
  \leq \frac{1}{\lambda_1}
       \bigl( C + \lambda_2(1 + \ln(F_0^T))\bigr),
  \label{eq:bound-derivative}
\end{equation}
for all $k\geq 1$, where we used \eqref{eq:H-bounded} and
\eqref{eq:entropy-upper-bound}. Hence $(f_k')$ is bounded in
$L^2(\mathcal{X})$.

Next we use the Gagliardo--Nirenberg inequality in bounded domains 
(\cite{Nirenberg1959}) to obtain an $L^2$ bound for $(f_k)$.
In dimension~1, there exists a constant $\widehat{C}>0$ such that for all
$u\in H^1(\mathcal{X})$,
\begin{equation}
  \lVert u\rVert_{L^2(\mathcal{X})}
  \leq \widehat{C}
    \lVert u'\rVert_{L^2(\mathcal{X})}^{1/3}
    \lVert u\rVert_{L^1(\mathcal{X})}^{2/3}
  + \widehat{C}\lVert u\rVert_{L^1(\mathcal{X})}.
  \label{eq:GN}
\end{equation}
Applying \eqref{eq:GN} to $u = f_k$ and using that each $f_k$ is a
probability density, i.e. $\lVert f_k\rVert_{L^1(\mathcal{X})} = 1$, we obtain
\begin{equation}
  \lVert f_k\rVert_{L^2(\mathcal{X})}
  \leq \widehat{C}
    \left(
      \sqrt{\frac{1}{\lambda_1}
        \bigl( C + \lambda_2(1 + \ln(F_0^T))}
      + 1
    \right),
  \label{eq:L2-bound}
\end{equation}
for all $k\geq 1$, where we used the uniform bound
\eqref{eq:bound-derivative}. Thus $(f_k)$ is bounded in $L^2(\mathcal{X})$.
Combining \eqref{eq:bound-derivative} and \eqref{eq:L2-bound}, we deduce
that $(f_k)$ is bounded in $H^1(\mathcal{X})$.

We now show that $(f_k)$ is equicontinuous and uniformly bounded. For any
$x,y\in \mathcal{X}$ and any $k\geq 1$, the Cauchy--Schwarz inequality yields
\begin{equation}
  |f_k(x) - f_k(y)|
  = \left|\int_x^y f_k'(s)\,ds\right|
  \leq \lVert f_k'\rVert_{L^2(\mathcal{X})}\,|x-y|^{1/2},
  \label{eq:equicontinuity}
\end{equation}
so $(f_k)$ is equicontinuous on $\mathcal{X}$, since $(f_k')$ is bounded in
$L^2(\mathcal{X})$.

Moreover, for any $x\in \mathcal{X}$ and any $k\geq 1$,
\begin{align}
  |f_k(x)|
  &= \left|
        f_k(x) - \frac{1}{b}\int_0^b f_k(y)\,dy
        + \frac{1}{b}\int_0^b f_k(y)\,dy
     \right|
     \nonumber\\
  &\leq \frac{1}{b} \int_0^b |f_k(x) - f_k(y)|\,dy + \frac{1}{b}
     \nonumber\\
  &\leq \frac{1}{b} \lVert f_k'\rVert_{L^2(\mathcal{X})}
                 \int_0^b |x-y|^{1/2}\,dy + \frac{1}{b}
     \nonumber\\
  &\leq \frac{2}{3} b^{1/2}\lVert f_k'\rVert_{L^2(\mathcal{X})}
       + \frac{1}{b},
  \label{eq:uniform-bound}
\end{align}
where we used \eqref{eq:equicontinuity}. Since $(f_k')$ is bounded in $L^2(\mathcal{X})$, this shows that
$(f_k)$ is uniformly bounded in $L^\infty(\mathcal{X})$.

\medskip
\noindent
\paragraph{Step 2: compactness and admissibility of the limit.}
By the Banach--Alaoglu theorem (and reflexivity of $H^1(\mathcal{X})$), the
boundedness of $(f_k)$ in $H^1(\mathcal{X})$ implies the existence of a
subsequence, still denoted by $(f_k)$, and a function $f^\ast\in
H^1(\mathcal{X})$ such that
\begin{equation}
  f_k \rightharpoonup f^\ast
  \quad \text{weakly in } H^1(\mathcal{X}),
  \label{eq:H1-weak}
\end{equation}
and in particular
\begin{equation}
  f_k' \rightharpoonup (f^\ast)'
  \quad \text{weakly in } L^2(\mathcal{X}).
  \label{eq:derivative-weak}
\end{equation}

On the other hand, the uniform boundedness
\eqref{eq:uniform-bound} and the equicontinuity
\eqref{eq:equicontinuity} allow us to apply the
Arzelà--Ascoli theorem, which yields (up to extraction of a further
subsequence) the uniform convergence of $(f_k)$ towards $f^\ast$ on
$\overline{\mathcal{X}}$:
\begin{equation}
  f_k \longrightarrow f^\ast
  \quad \text{uniformly on } \overline{\mathcal{X}}.
  \label{eq:uniform-convergence}
\end{equation}
In particular,
\begin{equation}
  f_k(x) \longrightarrow f^\ast(x)
  \quad \text{for almost every } x\in \mathcal{X},
  \label{eq:ae-convergence}
\end{equation}
and since $\mathcal{X}$ has finite measure, \eqref{eq:uniform-convergence}
implies
\begin{equation}
  f_k \longrightarrow f^\ast
  \quad \text{in } L^p(\mathcal{X}), \quad \forall\,p\geq 1.
  \label{eq:Lp-convergence}
\end{equation}

We now verify that $f^\ast\in \mathcal{A}$. First,
\eqref{eq:ae-convergence} and the non-negativity of each $f_k$ yield
$f^\ast\geq 0$ almost everywhere on $\mathcal{X}$. Second, taking $p=1$ in
\eqref{eq:Lp-convergence} and using that each $f_k$ is a probability
density, we obtain
\begin{equation}
  \int_\mathcal{X} f^\ast(x)\,dx
  = \lim_{k\to\infty} \int_\mathcal{X} f_k(x)\,dx
  = 1,
  \label{eq:limit-density}
\end{equation}
so $f^\ast$ is a probability density on $\mathcal{X}$.

Moreover, for each strike $K_i$ we know that the function
$x\mapsto (x-K_i)_+$ belongs to $L^2(\mathcal{X})$. Then, we obtain convergence of the corresponding expectations (and therefore the call prices) under $f_k$ towards those under $f^\ast$.
Finally, $f^\ast$ satisfies all constraints defining $\mathcal{A}$.

\medskip
\noindent
\paragraph{Step 3: lower semicontinuity of the objective and existence of a minimizer.}
We now show that $H$ is weakly lower semicontinuous on $\mathcal{A}$.

First, the map $f\mapsto \lVert f'\rVert_{L^2(\mathcal{X})}^2$ is convex and
continuous on $H^1(\mathcal{X})$, hence weakly lower semicontinuous.

Next, we address the entropy term. Since the function
$t\mapsto t\ln t$ is continuous on $\R_+$ and reaches its minimum at
$t = e^{-1}$ with the value $-e^{-1}$, the shifted function
\[
  t \longmapsto t\ln t + e^{-1}
\]
is non-negative and continuous on $\R_+$. Combining this with the
pointwise convergence \eqref{eq:ae-convergence}, we obtain
\begin{equation}
  f_k(x)\ln f_k(x) + e^{-1}
  \longrightarrow f^\ast(x)\ln f^\ast(x) + e^{-1}
  \quad \text{for a.e.\ } x\in \mathcal{X}.
  \label{eq:pointwise-entropy}
\end{equation}
By Fatou's lemma,
\begin{equation}
  \int_\mathcal{X} \bigl( f^\ast\ln f^\ast + e^{-1}\bigr)
  \leq \liminf_{k\to\infty}
       \int_\mathcal{X} \bigl( f_k\ln f_k + e^{-1}\bigr),
  \label{eq:fatou}
\end{equation}
which implies
\begin{equation}
  \int_\mathcal{X} f^\ast\ln f^\ast
  \leq \liminf_{k\to\infty}
       \int_\mathcal{X} f_k\ln f_k.
  \label{eq:entropy-lsc}
\end{equation}
More general versions of this result can be found in
\cite{poljak1969,yoffe1977}; \citet[Theorem 5.14]{fonseca_leoni_2007}.

Combining the weak lower semicontinuity of the quadratic term and
\eqref{eq:entropy-lsc}, we obtain
\begin{align}
  H(f^\ast)
  &= \lambda_1 \lVert (f^\ast)'\rVert_{L^2(\mathcal{X})}^2
     + \lambda_2 \int_\mathcal{X} f^\ast\ln f^\ast
     \nonumber\\
  &\leq \lambda_1
       \liminf_{k\to\infty} \lVert f_k'\rVert_{L^2(\mathcal{X})}^2
     + \lambda_2
       \liminf_{k\to\infty} \int_\mathcal{X} f_k\ln f_k
     \nonumber\\
  &\leq \liminf_{k\to\infty}
       \left(
         \lambda_1 \lVert f_k'\rVert_{L^2(\mathcal{X})}^2
         + \lambda_2 \int_\mathcal{X} f_k\ln f_k
       \right)
     \nonumber\\
  &= \liminf_{k\to\infty} H(f_k)
   = \inf_{f\in \mathcal{A}} H(f),
  \label{eq:H-lsc}
\end{align}
where we used \eqref{eq:minimizing-sequence} in the last equality.
Thus $f^\ast$ is a minimizer of $H$ over $\mathcal{A}$.

\medskip
\noindent
\paragraph{Step 4: strict convexity and uniqueness.}
We now prove that $H$ is strictly convex on $\mathcal{A}$, which yields uniqueness
of the minimizer.

First, recall that the map $u\mapsto \lVert u\rVert_{L^2(\mathcal{X})}^2$ is
strictly convex on $L^2(\mathcal{X})$: for any $a,b\in\R$ and
$\theta\in (0,1)$,
\begin{equation}
  |\theta a + (1-\theta)b|^2
  = \theta |a|^2 + (1-\theta)|b|^2
    - \theta(1-\theta)|a-b|^2.
  \label{eq:scalar-strict-convex}
\end{equation}
Integrating \eqref{eq:scalar-strict-convex} over $\mathcal{X}$ gives, for
$f,g\in L^2(\mathcal{X})$,
\begin{equation}
  \lVert \theta f + (1-\theta)g\rVert_{L^2(\mathcal{X})}^2
  = \theta \lVert f\rVert_{L^2(\mathcal{X})}^2
    + (1-\theta)\lVert g\rVert_{L^2(\mathcal{X})}^2
    - \theta(1-\theta)\lVert f-g\rVert_{L^2(\mathcal{X})}^2.
  \label{eq:L2-strict-convex}
\end{equation}
In particular, if $f\neq g$ on a set of positive measure, then
\begin{equation}
  \lVert \theta f + (1-\theta)g\rVert_{L^2(\mathcal{X})}^2
  < \theta \lVert f\rVert_{L^2(\mathcal{X})}^2
    + (1-\theta)\lVert g\rVert_{L^2(\mathcal{X})}^2.
  \label{eq:L2-strict-convex-ineq}
\end{equation}

However, in our functional $H$ the quadratic term involves $f'$ rather than
$f$ itself. In general, the map $f\mapsto \lVert f'\rVert_{L^2(\mathcal{X})}^2$
is not strictly convex on $H^1(\mathcal{X})$, since two functions with the same
derivative may differ by an additive constant. On the admissible set $\mathcal{A}$,
this is excluded: if $f,g\in \mathcal{A}$ and $f' = g'$, then $f-g$ is
constant, but both $f$ and $g$ integrate to $1$ on $\mathcal{X}$, so that
\[
  \int_\mathcal{X} (f-g) = 0,
\]
which forces $f=g$. Thus, restricted to $\mathcal{A}$, the map
\begin{equation}
  f \longmapsto \lVert f'\rVert_{L^2(\mathcal{X})}^2
  \label{eq:derivative-quadratic}
\end{equation}
is strictly convex.

On the other hand, the function $t\mapsto t\ln t$ is strictly convex on
$\R_+$, and the passage to the integral preserves strict convexity whenever
the functions differ on a set of positive measure. Hence $-S(f)$ is strictly convex on $\mathcal{A}$.

Therefore $H$ is a strictly convex functional on the convex set $\mathcal{A}$,
so it admits at most one minimizer. Combining the existence established in Step 3 
with the uniqueness shown here, we conclude that $H$ admits a unique minimizer $f^\ast$, 
which completes the proof.
\end{proof}

\begin{remark}[Unbounded State Space Case]
\label{appendix:unbounded_domain}
For the sake of clarity, the existence result (Theorem~\ref{thm:exist-unique-continuous-minimizer}) has been established on a bounded state space. In this setting, the weak convergence of densities implies the convergence of moments.\\
On an unbounded state space, for instance $\mathcal{X} = \R_{+}$, an additional control of the tails is required to conclude that the limiting density $f^{\ast}$ is admissible : uniform integrability of the minimizing sequence $(f_{n})$ is necessary to pass to the limit in the forward and call price constraints.\\
A convenient way to guarantee such uniform integrability is to impose a uniform $(1+\alpha)$-moment bound for some $\alpha>0$. This can be achieved in at least two equivalent ways in our framework:

\begin{enumerate}
  \item \emph{Moment penalization in the objective.}  
  Modify the objective functional by adding a coercive term of the form
  \[
    f \mapsto \lambda \int_{0}^{\infty} x^{1+\alpha} f(x)\,dx,
    \qquad \lambda>0,\ \alpha>0.
  \]
  \item \emph{Moment constraint on the admissible set.}  
  Restrict the admissible set $\mathcal{A}$ by adding the moment condition
  \[
    \int_{0}^{\infty} x^{1+\alpha} f(x)\,dx \leq C,
  \]
  for some constant $C>0$ independent of $f$.
\end{enumerate}

In both cases, the uniform $(1+\alpha)$-moment bound provides uniform integrability via a standard de la Vallée-Poussin argument, ensuring that the limiting density $f^{\ast}$ still matches the constraints on $\R_{+}$.

\end{remark}
\subsubsection{Convergence of the SEDEx Discretization}
\label{subsec:disc-to-cont}

Having established existence and uniqueness in both the continuous and discrete
settings, we now prove that the discrete minimizer converges to its continuous
counterpart as the mesh is refined. Throughout this section we set
$\epsilon_M:=b/M$, which coincides with the mesh size $\Delta s$ but makes the
$M$-dependence explicit. The cell-average projection of an element of $\mathcal A$
shifts the moments by $O(\epsilon_M)$, so we consider the relaxed admissible set
\begin{equation}
\mathcal{A}_M^{\epsilon_M}
:=
\left\{
p\in\Sigma_M \;\middle|\;
\begin{array}{l}
\displaystyle
\left|\sum_{i=1}^M s_i p_i - F_0^T\right|\le \epsilon_M, \\[1.2em]
\displaystyle
e^{-rT}\sum_{i=1}^M (s_i-K_j)_+\,p_i
\in
\left[
C_j^{\mathrm{bid}}-\epsilon_M,\,
C_j^{\mathrm{ask}}+\epsilon_M
\right],
\qquad j=1,\dots,N
\end{array}
\right\}.
\label{eq:AMeps}
\end{equation}
The set $\mathcal A_M^{\epsilon_M}$ relaxes the admissible set
$\mathcal A_M$ of Definition~\ref{def:discrete-admissible-set} by a uniform
tolerance $\epsilon_M$ on the pricing constraints. This tolerance vanishes
as $M\to\infty$ and, at any practically relevant resolution, lies well below
the precision of standard solvers. Since $\mathcal A_M\subset
\mathcal A_M^{\epsilon_M}$, Propositions~\ref{prop:AM-compact-convex}
and~\ref{prop:HM-strict-convex} apply unchanged to $\mathcal A_M^{\epsilon_M}$,
and $H_M$ admits a unique minimizer over it, denoted $p^{\ast,M}$.

To pass to the continuous limit we associate to any $p\in(\R_+)^M$ two
reconstructions on $\overline{\mathcal X}=[0,b]$. Setting $s_0:=0$ and
$I_i:=[s_{i-1},s_i)$, define:
\begin{itemize}
  \item the \emph{piecewise constant interpolation}
  $\tilde f_p(x):=p_i/\epsilon_M$ for $x\in I_i$;
  \item the \emph{piecewise affine interpolation} $\hat f_p\in C(\overline{\mathcal X})$,
  affine on each $[s_i,s_{i+1}]$ with $\hat f_p(s_i)=p_i/\epsilon_M$ for
  $i=1,\dots,M$, and extended by $p_1/\epsilon_M$ on $[0,s_1]$.
\end{itemize}
A direct computation yields the identities
\begin{equation}
\|\hat f_p'\|_{L^2(\mathcal X)}^2 = \frac{1}{\epsilon_M^3}\,\|D^{(1)}p\|_2^2,
\qquad
\int_{\mathcal X}\tilde f_p\ln\tilde f_p\,dx
=\sum_{i=1}^M p_i\ln p_i\;-\;\ln\epsilon_M,
\label{eq:lift-identities}
\end{equation}
so that, plugging both into~\eqref{eq:def-HM},
\begin{equation}
H_M(p)
=\lambda_1\,\|\hat f_p'\|_{L^2(\mathcal X)}^2
+\lambda_2\int_{\mathcal X}\tilde f_p\ln\tilde f_p\,dx
+\lambda_2\,\ln\epsilon_M.
\label{eq:HM-lifted}
\end{equation}
The shift $\lambda_2\ln\epsilon_M$ does not depend on $p$ and is irrelevant to
the minimization, but it must be tracked when comparing $H_M$ with $H$.

\begin{theorem}[Discrete-to-continuous convergence]
\label{thm:disc-to-cont}
Let $f^\ast\in\mathcal A$ be the unique minimizer of $H$ over $\mathcal A$
(\Cref{thm:exist-unique-continuous-minimizer}).
We also consider 
a sequence of grids as in \Cref{def:grid-decision-variable}; the $\ell$-th grid has $M_\ell$ points which form a strictly increasing sequence $\Mcal_{grids}=\{M_\ell, \ell \ge 1\}$.
For a generic $M \in \Mcal_{grids}$, let
$p^{\ast,M}$ be the unique minimizer of $H_M$ over $\mathcal A_M^{\epsilon_M}$ 
and
$\hat f_M:=\hat f_{p^{\ast,M}}$. Then for  $M \in \Mcal_{grids}$, $ M\to\infty$,
\begin{equation}
\hat f_M \longrightarrow f^\ast \quad\text{uniformly on }\mathcal X,
\qquad
\hat f_M \rightharpoonup f^\ast \quad\text{weakly in }H^1(\mathcal X).
\end{equation}
\end{theorem}
\begin{proof}
We proceed in three steps.
\paragraph{Step 1: Comparing the discrete and continuous objectives.}
Define the projection $\Pi_M:\mathcal A\to\R^M$ by
$(\Pi_M f)_i:=\int_{I_i}f(x)\,dx$. For any $f\in\mathcal A$, non-negativity and
$\int f=1$ give $\Pi_M f\in\Sigma_M$. The functions $\varphi_0(x):=x$ and
$\varphi_j(x):=(x-K_j)_+$ are $1$-Lipschitz on $\overline{\mathcal X}$, so
$|\varphi_\ell(s_i)-\varphi_\ell(x)|\le \epsilon_M$ for every $x\in I_i$ and
every $\ell\in\{0,\dots,N\}$. Since $\int_{\mathcal X}f=1$,
\begin{equation}
\Bigl|\sum_{i=1}^M\varphi_\ell(s_i)(\Pi_M f)_i
-\int_{\mathcal X}\varphi_\ell(x)\,f(x)\,dx\Bigr|
=\Bigl|\sum_{i=1}^M\int_{I_i}\bigl(\varphi_\ell(s_i)-\varphi_\ell(x)\bigr)f(x)\,dx\Bigr|
\le \epsilon_M,
\label{eq:lip-bound}
\end{equation}
for $\ell=0,1,\dots,N$. Together with $e^{-rT}\le 1$ and the constraints
defining $\mathcal A$, this yields $\Pi_M f\in \mathcal A_M^{\epsilon_M}$.

Set $\bar f_i:=(\Pi_M f)_i/\epsilon_M$. A change of variable and
Cauchy--Schwarz give
\begin{equation}
(\bar f_{i+1}-\bar f_i)^2
=\Bigl(\tfrac{1}{\epsilon_M}\int_{I_i}\int_0^{\epsilon_M}f'(x+t)\,dt\,dx\Bigr)^2
\le \int_{I_i}\int_0^{\epsilon_M}|f'(x+t)|^2\,dt\,dx,
\label{eq:cell-diff}
\end{equation}
and summing over $i=1,\dots,M-1$ and applying Fubini,
\begin{equation}
\sum_{i=1}^{M-1}(\bar f_{i+1}-\bar f_i)^2
\le \int_0^{\epsilon_M}\!\int_0^{b-\epsilon_M}|f'(x+t)|^2\,dx\,dt
\le \epsilon_M\,\|f'\|_{L^2(\mathcal X)}^2.
\label{eq:proj-deriv}
\end{equation}
Combined with the first identity in~\eqref{eq:lift-identities}, this gives
\begin{equation}
\|\hat f_{\Pi_M f}'\|_{L^2(\mathcal X)}^2
=\frac{1}{\epsilon_M}\sum_{i=1}^{M-1}(\bar f_{i+1}-\bar f_i)^2
\le \|f'\|_{L^2(\mathcal X)}^2.
\label{eq:proj-deriv-final}
\end{equation}
For the entropy, Jensen's inequality applied to $t\mapsto t\ln t$ with the
probability measure $dx/\epsilon_M$ on $I_i$ gives
$\bar f_i\ln\bar f_i\le \tfrac{1}{\epsilon_M}\int_{I_i}f\ln f\,dx$, which
rewrites as $(\Pi_M f)_i\ln(\Pi_M f)_i-(\Pi_M f)_i\ln\epsilon_M\le
\int_{I_i}f\ln f\,dx$. Summing over $i$ and using
$\sum_i(\Pi_M f)_i=1$,
\begin{equation}
\sum_{i=1}^M(\Pi_M f)_i\ln(\Pi_M f)_i
\le \int_{\mathcal X}f\ln f\,dx+\ln\epsilon_M.
\label{eq:proj-entropy}
\end{equation}
Combining~\eqref{eq:proj-deriv-final}--\eqref{eq:proj-entropy}
with~\eqref{eq:def-HM} gives
\begin{equation}
H_M(\Pi_M f)\le H(f)+\lambda_2\ln\epsilon_M.
\label{eq:HM-Pi-bound}
\end{equation}
Choosing $f=f^\ast\in\mathcal A$ and using the optimality of $p^{\ast,M}$ over
$\mathcal A_M^{\epsilon_M}\ni\Pi_M f^\ast$,
\begin{equation}
H_M(p^{\ast,M})\le H(f^\ast)+\lambda_2\ln\epsilon_M.
\label{eq:HM-upper}
\end{equation}
Substituting~\eqref{eq:HM-lifted} for $p=p^{\ast,M}$ into~\eqref{eq:HM-upper}
and cancelling $\lambda_2\ln\epsilon_M$ yields the central inequality
\begin{equation}
\lambda_1\,\|\hat f_M'\|_{L^2(\mathcal X)}^2
+\lambda_2\int_{\mathcal X}\tilde f_M\ln\tilde f_M\,dx
\le H(f^\ast).
\label{eq:central}
\end{equation}

\paragraph{Step 2: Uniform $H^1$ bound and compactness.}
The discrete entropy attains its maximum $\ln M$ on $\Sigma_M$ at the uniform
distribution; hence $-S_M(p^{\ast,M})\ge-\ln M$, and the second identity
in~\eqref{eq:lift-identities} combined with $M\epsilon_M=b$ gives
\begin{equation}
\int_{\mathcal X}\tilde f_M\ln\tilde f_M\,dx
=-S_M(p^{\ast,M})-\ln\epsilon_M
\ge -\ln(M\epsilon_M)=-\ln b.
\label{eq:entropy-lower}
\end{equation}
Plugging~\eqref{eq:entropy-lower} into~\eqref{eq:central} produces the uniform
$H^1$-seminorm bound
\begin{equation}
\|\hat f_M'\|_{L^2(\mathcal X)}^2
\le C_1:=\frac{1}{\lambda_1}\bigl(H(f^\ast)+\lambda_2\ln b\bigr).
\label{eq:H1-seminorm-bd}
\end{equation}

We first establish that $\hat f_M$ and $\tilde f_M$ asymptotically coincide.
Since $\hat f_M-\tilde f_M$ is affine on each $[s_i,s_{i+1}]$ and vanishes at
$s_{i+1}$, its $L^\infty$ norm on that interval equals
$|\bar p_{i+1}^M-\bar p_i^M|$. The affine slope of $\hat f_M$ on
$[s_i,s_{i+1}]$ being $(\bar p_{i+1}^M-\bar p_i^M)/\epsilon_M$, we have
$\sum_{i=1}^{M-1}(\bar p_{i+1}^M-\bar p_i^M)^2/\epsilon_M
=\|\hat f_M'\|_{L^2(\mathcal X)}^2$. Since each nonnegative term is bounded by
the whole sum, using~\eqref{eq:H1-seminorm-bd},
\begin{equation}
\|\hat f_M-\tilde f_M\|_{L^\infty(\mathcal X)}
\le \max_{1\le i\le M-1}|\bar p_{i+1}^M-\bar p_i^M|
\le \sqrt{\epsilon_M\,C_1}
\longrightarrow 0 \quad\text{as } M\to\infty.
\label{eq:tilde-hat-close}
\end{equation}

The triangle inequality, $\|\tilde f_M\|_{L^1(\mathcal X)}=1$
and~\eqref{eq:tilde-hat-close} yield the uniform $L^1$ bound
\begin{equation}
\|\hat f_M\|_{L^1(\mathcal X)}
\le \|\tilde f_M\|_{L^1(\mathcal X)}
   +b\,\|\hat f_M-\tilde f_M\|_{L^\infty(\mathcal X)}
\le 1+b\sqrt{b\,C_1}.
\label{eq:hat-L1-bd}
\end{equation}
Applying the Gagliardo--Nirenberg inequality~\eqref{eq:GN} to $\hat f_M$ in
place of $f_k$, together with~\eqref{eq:H1-seminorm-bd}
and~\eqref{eq:hat-L1-bd}, shows that $(\hat f_M)$ is uniformly bounded in
$H^1(\mathcal X)$.

The same substitution in~\eqref{eq:uniform-bound}, with $\|f_k\|_{L^1}=1$
replaced by the right-hand side of~\eqref{eq:hat-L1-bd}, yields the uniform
$L^\infty$ bound
\begin{equation}
\|\hat f_M\|_{L^\infty(\mathcal X)}
\le \tfrac{1+b\sqrt{b\,C_1}}{b}+\tfrac{2}{3}\,b^{1/2}\sqrt{C_1}.
\label{eq:hat-Linfty-bd}
\end{equation}
Similarly, applying the Cauchy--Schwarz estimate~\eqref{eq:equicontinuity} to
$\hat f_M$ in place of $f_k$ gives the uniform equicontinuity
\begin{equation}
|\hat f_M(x)-\hat f_M(y)|\le \sqrt{C_1}\,|x-y|^{1/2},
\qquad x,y\in\mathcal X,\ M\ge 1.
\label{eq:hat-equicont}
\end{equation}

By the Banach--Alaoglu theorem applied to $H^1(\mathcal X)$, there exist a
subsequence $(M_{\ell_k})_{k\ge 1}\subset\mathcal{M}_{grids}$
(still indexed by $M$, not relabeled) and $\bar f\in H^1(\mathcal X)$ such that
$\hat f_M\rightharpoonup \bar f$ in $H^1(\mathcal X)$, and in particular
$\hat f_M'\rightharpoonup \bar f'$ in $L^2(\mathcal X)$. By the
Arzelà--Ascoli theorem, the $L^\infty$ bound~\eqref{eq:hat-Linfty-bd} and the
equicontinuity~\eqref{eq:hat-equicont} yield, upon further extraction,
$\hat f_M\to\bar f$ uniformly on $\mathcal X$. Combined
with~\eqref{eq:tilde-hat-close}, the triangle inequality gives
$\tilde f_M\to \bar f$ uniformly on $\mathcal X$ as well.

\paragraph{Step 3: Admissibility of $\bar f$ and identification with $f^\ast$.}
Pointwise non\-negativity of $\tilde f_M$ and uniform convergence give
$\bar f\ge 0$ on $\mathcal X$. Since $\mathcal X$ has finite measure, uniform
convergence yields $L^1$-convergence, and
$\int_{\mathcal X}\bar f\,dx=\lim_M\int_{\mathcal X}\tilde f_M\,dx=1$.

For the linear constraints, the Lipschitz estimate~\eqref{eq:lip-bound}
applied to $f=\tilde f_M$ (whose cell averages are $\bar p_i^M$) together with
the uniform convergence $\tilde f_M\to\bar f$ yields, for every
$\ell\in\{0,\dots,N\}$,
\begin{equation}
\sum_{i=1}^M \varphi_\ell(s_i)\,p_i^{\ast,M}
\longrightarrow
\int_{\mathcal X}\varphi_\ell(x)\,\bar f(x)\,dx
\quad\text{as } M\to\infty.
\label{eq:limit-moments}
\end{equation}
Combining~\eqref{eq:limit-moments} with the relaxed constraints
defining~\eqref{eq:AMeps} and $\epsilon_M\to 0$ gives
$\int_{\mathcal X}x\,\bar f\,dx=F_0^T$ and
$e^{-rT}\int_{\mathcal X}(x-K_j)_+\,\bar f\,dx\in
[C_j^{\mathrm{bid}},C_j^{\mathrm{ask}}]$ for every $j=1,\dots,N$. Hence
$\bar f\in\mathcal A$.

It remains to pass to the limit in~\eqref{eq:central}. The weak lower
semicontinuity of $u\mapsto\|u\|_{L^2(\mathcal X)}^2$ and the Fatou-based
lower semicontinuity of $f\mapsto\int_{\mathcal X}f\ln f\,dx$ under pointwise
convergence, both invoked in Step~3 of the proof of
Theorem~\ref{thm:exist-unique-continuous-minimizer} and applied here to the
weak convergence $\hat f_M'\rightharpoonup\bar f'$ and the pointwise
convergence $\tilde f_M\to\bar f$ a.e.\ on $\mathcal X$, give
\begin{equation}
\|\bar f'\|_{L^2(\mathcal X)}^2
\le \liminf_{M\to\infty}\|\hat f_M'\|_{L^2(\mathcal X)}^2,
\qquad
\int_{\mathcal X}\bar f\ln\bar f\,dx
\le \liminf_{M\to\infty}\int_{\mathcal X}\tilde f_M\ln\tilde f_M\,dx.
\label{eq:wlsc-both}
\end{equation}
Combining~\eqref{eq:wlsc-both} with the central inequality~\eqref{eq:central},
\begin{equation}
\begin{split}
H(\bar f)
&=\lambda_1\|\bar f'\|_{L^2(\mathcal X)}^2
+\lambda_2\int_{\mathcal X}\bar f\ln\bar f\,dx \\
&\le \liminf_{M\to\infty}\Bigl[\lambda_1\|\hat f_M'\|_{L^2(\mathcal X)}^2
+\lambda_2\int_{\mathcal X}\tilde f_M\ln\tilde f_M\,dx\Bigr] \\
&\le H(f^\ast).
\end{split}
\label{eq:H-limsup}
\end{equation}
Since $\bar f\in\mathcal A$ and $f^\ast$ is the unique minimizer of $H$ over
$\mathcal A$ (Theorem~\ref{thm:exist-unique-continuous-minimizer}), this forces
$\bar f=f^\ast$. The limit being independent of the extracted subsequence, the whole sequence $(\hat f_{M_\ell})_{\ell\ge 1}$ converges to $f^\ast$ uniformly
on $\mathcal X$ and weakly in $H^1(\mathcal X)$.
\end{proof}

\begin{remark}[Weak convergence of the discrete risk-neutral measure]
\label{rem:measure-convergence}
The argument used to establish~\eqref{eq:limit-moments} extends, by uniform
continuity, to any $\varphi\in C(\overline{\mathcal X})$ in place of
$\varphi_\ell$. Consequently, the atomic discrete probability measure
$\nu_M:=\sum_{i=1}^M p_i^{\ast,M}\delta_{s_i}$ converges weakly, as $M\to\infty$,
to the continuous probability measure $\mu^\ast(dx):=f^\ast(x)\,dx$.
In particular, the discrete call prices
$e^{-rT}\sum_{i=1}^M(s_i-K_j)_+ p_i^{\ast,M}$ converge to their continuous
counterparts $e^{-rT}\int_{\mathcal X}(x-K_j)_+ f^\ast(x)\,dx$ for every
$j=1,\dots,N$.
\end{remark}

\section{Additional numerical results}
\label{appendix:7dte_market_slice}

For completeness, this appendix reports a representative 7DTE slice observed on 2023-04-14. The chain contains 232 quoted strikes. As in the main text, the implied-volatility smile obtained from repriced options provides a satisfactory fit to market quotes.

\begin{figure}[H]
    \centering
    \begin{subfigure}[b]{0.48\textwidth}
        \centering
        \includegraphics[width=\textwidth]{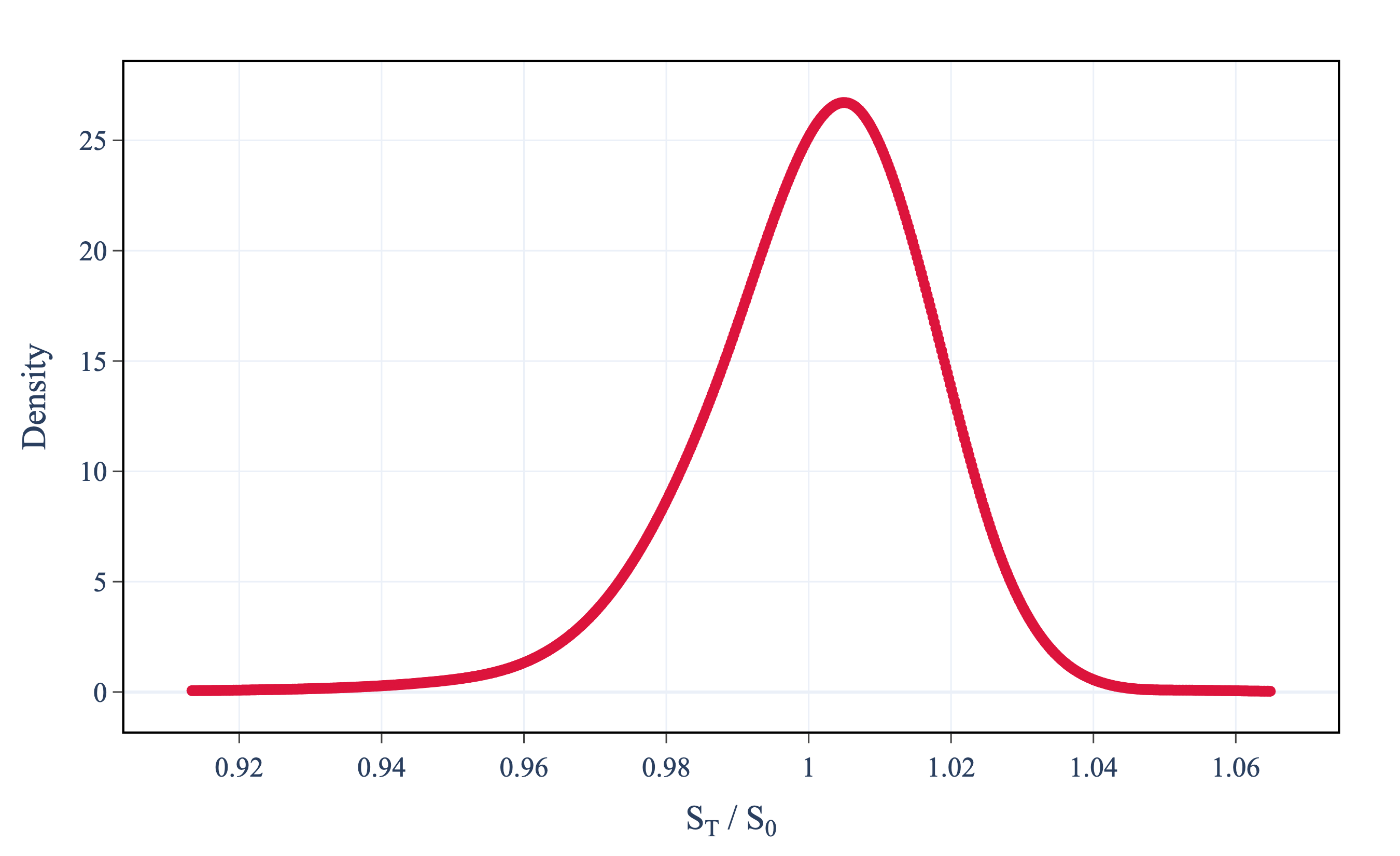}
        \caption*{(a)}
    \end{subfigure}
    \hfill
    \begin{subfigure}[b]{0.48\textwidth}
        \centering
        \includegraphics[width=\textwidth]{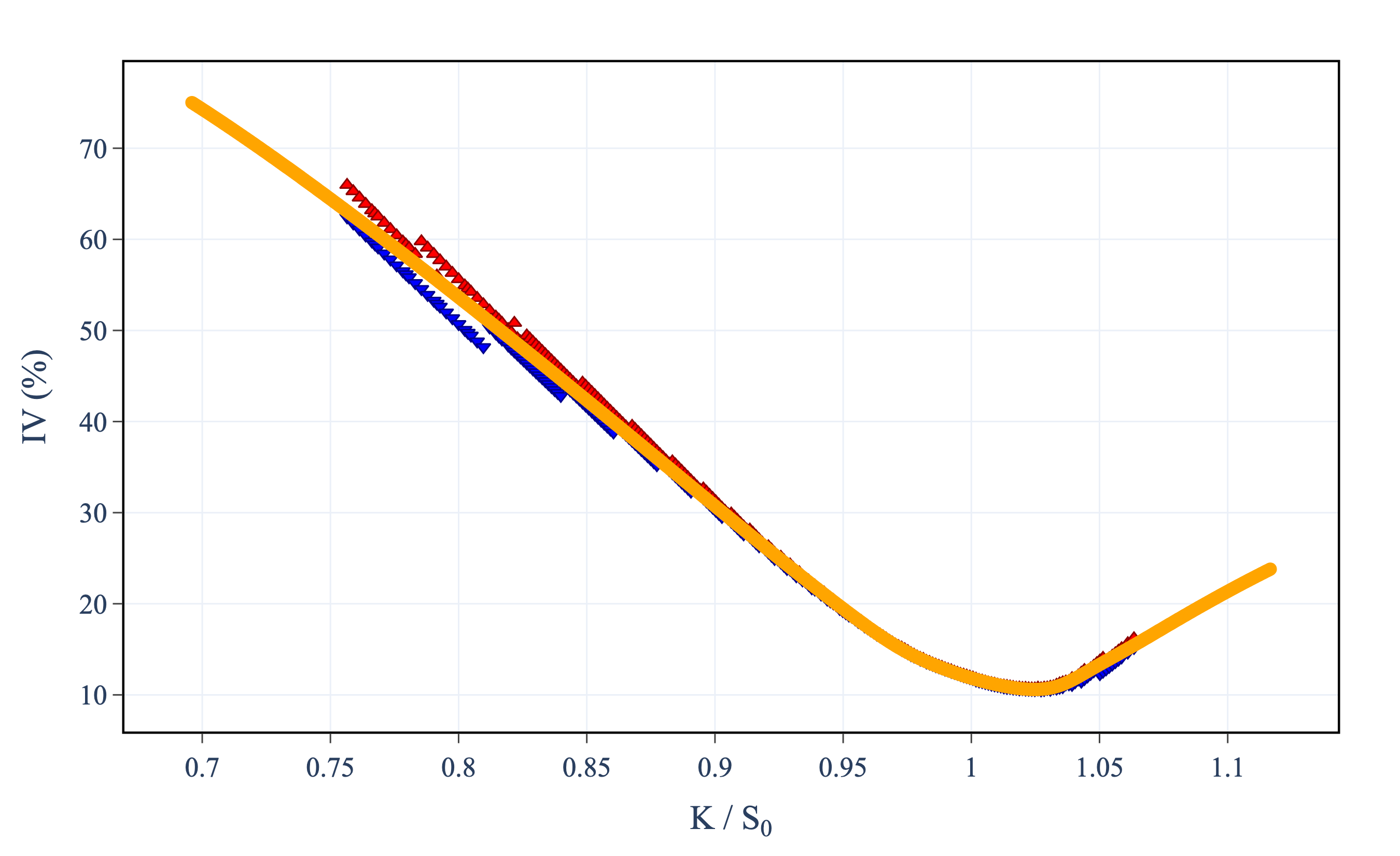}
        \caption*{(b)}
    \end{subfigure}
    \caption{Representative 7DTE market slice on 2023-04-14. 
    \textit{(a)} Risk-neutral density extracted by SEDEx (red dots). \textit{(b)} 
    Market ask (red triangles-up), market bid (blue triangles-down), and SEDEx 
    implied volatilities (orange circles). The bid--ask constraints are satisfied 
    up to a numerical tolerance of $10^{-6}$ vol points.}
    \label{fig:market_7dte_appendix}
\end{figure}

\section{Algorithms}
\subsection{Arbitrage Removal Iterative Executable Strategy (ARIES)}
\begin{algorithm}[htbp]
\caption{Iterative Filtering of Static Arbitrage}
\label{alg:arb_filter}
\begin{algorithmic}[1]
\Require Strikes $(K_i)_{i=1}^N$, quotes $(C_i^{\mathrm{ask}}, C_i^{\mathrm{bid}})$, depth bounds $(Q_i^{\mathrm{ask}}, Q_i^{\mathrm{bid}})$.
\Ensure An arbitrage-free subset of option indices $\mathcal{I}^*$.
\State Initialize the active set $\mathcal{I} \gets \{1,\dots,N\}$.
\Repeat
    \State Solve the linear program \eqref{eq:P3prime} restricted to the active set $i \in \mathcal{I}$.
    \State Extract a \textbf{vertex} optimal solution $x^* = (q^{\mathrm{ask}*}, q^{\mathrm{bid}*}, u^*, \alpha^*)$.
    \If{$x^* \neq 0$} \Comment{Equivalently, $\tilde{\Pi}^*$ is not identically zero}
        \If{$\tilde{\pi}_0(x^*) = 0$} \Comment{Weak-arbitrage case}
            \State Compute $\bar{\lambda}(x^*)$ via \eqref{eq:lambda_bar} and update $x^* \gets \bar{\lambda}(x^*) x^*$.
        \EndIf
        \State Identify the set of binding depth constraints:
        \Statex \qquad $\mathcal{B} \gets \Bigl\{i \in \mathcal{I} : q_i^{\mathrm{ask}*} = Q_i^{\mathrm{ask}} \text{ or } q_i^{\mathrm{bid}*} = Q_i^{\mathrm{bid}} \Bigr\}$.
        \State \Comment{By Proposition~\ref{prop:characterization_proof_version}, $\mathcal{B} \neq \emptyset$.}
        \If{$|\mathcal{B}| > 1$}
            \State For each $i \in \mathcal{B}$, define the saturated depth:
            \Statex \qquad $Q_i^{\mathrm{bind}} = \begin{cases} 
                Q_i^{\mathrm{ask}} & \text{if } q_i^{\mathrm{ask}*} = Q_i^{\mathrm{ask}} \\
                Q_i^{\mathrm{bid}} & \text{if } q_i^{\mathrm{bid}*} = Q_i^{\mathrm{bid}}
            \end{cases}$
            \State Remove the option with the smallest saturated size:
            \Statex \qquad $i^{\dagger} \gets \arg\min_{i \in \mathcal{B}} Q_i^{\mathrm{bind}}$.
        \Else
            \State Let $i^{\dagger}$ be the unique element in $\mathcal{B}$.
        \EndIf
        \State Update the active set: $\mathcal{I} \gets \mathcal{I} \setminus \{i^{\dagger}\}$.
    \EndIf
\Until{$x^* = 0$} \Comment{No static arbitrage remains}
\State \Return $\mathcal{I}^* \gets \mathcal{I}$.
\end{algorithmic}
\end{algorithm}

\begin{paragraph}{Implementation Details}
In our numerical implementation, problem \eqref{eq:P3prime} is solved with the \texttt{HiGHS} optimization suite. \texttt{HiGHS} is an open-source solver for LP/MIP/QP models written in C++, with interfaces in C, R, and Python. We use its interior-point LP solver (\texttt{HiGHS-IPM}) together with the default crossover procedure. This configuration delivered the best practical performance in our experiments. While the IPM typically returns an optimal point on the optimal face (not necessarily a vertex of $\mathcal{C}$), the crossover step computes an optimal basic feasible solution.  
\end{paragraph}
\subsection{Smooth Entropic Density EXtraction (SEDEx)}
\label{subsub:alg_SEDEx}

\begin{algorithm}[H]
\caption{Hybrid Risk-Neutral Density Extraction under Bid-Ask Constraints}
\label{alg:rnd_extraction}
\begin{algorithmic}[1]
\Require Arbitrage-free option subset $\mathcal{I}^*$ (from Algorithm~\ref{alg:arb_filter}) with normalized strikes $K_i/S_0$ and quotes $(C_i^{\mathrm{ask}}/S_0, C_i^{\mathrm{bid}}/S_0)$. Time to expiry $T$. Rates $r, q$ recovered via \eqref{eq:rq_recovery}.
\Ensure The discrete risk-neutral density $(s, p^\ast)$ minimizing $H_M$ on $\mathcal{A}_M$.

\State \textbf{1. Grid Initialization}
\State Set normalized spot $S_0 = 1$ and forward $F = e^{(r-q)T}$.
\State Extract strike bounds $K_{\inf}, K_{\sup}$ via \eqref{eq:strike-bounds}.
\State Compute target mesh size $\Delta s_{\mathrm{tar}} \gets \sigma_{\mathrm{ATM}}\sqrt{2\pi T}\,\delta$ with $\delta=0.5\%$. \Comment{Section~\ref{subsec:mesh_size}}
\State Choose the compatible mesh size $\Delta s$ as in \eqref{eq:compatible-mesh-size}.
\State Compute the compatible upper endpoint $b$ and the number of grid points $M$ via \eqref{eq:upper-bound-support}.
\State Construct the uniform grid $s = (s_1, \dots, s_M)^\top$ where $s_j = j\Delta s$. \Comment{Equation~\eqref{eq:grid_def}}

\State \textbf{2. Matrix Construction}
\State Initialize the payoff matrix $\mathbf{M} \in \mathbb{R}^{|\mathcal{I}^*| \times M}$ with entries:
\Statex \qquad $\mathbf{M}_{ij} \gets (s_j-K_i)_+$
\State Form the combined inequality system for bid-ask constraints \eqref{eq:discrete-call}:
\Statex \qquad $\mathbf{M}^{\mathrm{comb}} \gets e^{-rT} \begin{pmatrix} \mathbf{M} \\ -\mathbf{M} \end{pmatrix}, \qquad C^{\mathrm{comb}} \gets \begin{pmatrix} (C_i^{\mathrm{ask}})_{i\in\mathcal{I}^*} \\ -(C_i^{\mathrm{bid}})_{i\in\mathcal{I}^*} \end{pmatrix}$

\State \textbf{3. Optimization Formulation}
\State Define the decision variable $p \in \mathbb{R}^M$. \Comment{Definition~\ref{def:grid-decision-variable}}
\State Retrieve regularization weights $\lambda_1, \lambda_2$ satisfying ratio \eqref{eq:lambda-ratio-scaling}. \Comment{Section~\ref{subsec:reg_weight}}
\State Define the discrete functional $H_M(p)$ via \eqref{eq:def-HM}:
\Statex \qquad $H_M(p) \gets \frac{\lambda_1}{(\Delta s)^3} \sum_{j=1}^{M-1}(p_{j+1}-p_{j})^2 \;+\; \lambda_2 \sum_{j=1}^{M} p_{j} \ln(p_{j})$

\State \textbf{4. Convex Programming}
\State Solve the constrained minimization problem on the admissible set $\mathcal{A}_M$ (Definition~\ref{def:discrete-admissible-set}):
\Statex \qquad $p^\ast \gets \arg\min_{p \ge 0} H_M(p)$
\Statex \qquad \text{subject to:}
\Statex \qquad $\sum_{j=1}^{M} p_{j} = 1$ \Comment{Simplex constraint \eqref{eq:simplex}}
\Statex \qquad $\sum_{j=1}^{M} s_j p_{j} = F$ \Comment{Forward constraint \eqref{eq:discrete-forward}}
\Statex \qquad $\mathbf{M}^{\mathrm{comb}} p \leq C^{\mathrm{comb}}$ \Comment{No-arbitrage call prices \eqref{eq:discrete-call}}

\State \Return The optimal grid-probability pair $(s, p^\ast)$.
\end{algorithmic}
\end{algorithm}
\paragraph{Implementation Details}
In our numerical implementation, the convex program defined in Algorithm~\ref{alg:rnd_extraction} is solved with \texttt{Clarabel}, an open-source interior-point solver for convex conic optimization. Among the solvers we tested, it delivered the best practical performance in our experiments. One of its main advantages is that it handles convex optimization problems with quadratic objectives directly. It also supports exponential-cone constraints, which makes it possible to represent the entropy component in \eqref{eq:def-HM} in a numerically stable conic form.

\section{On the Selection of Parameters}
\label{subsec:selection-regularization-weights}
\subsection{Detailed Estimation of the Risk-Free Rate and Dividend Yield}
\label{appendix:procedure_rq}
\paragraph{Put--Call parity as an affine restriction in strike.}
Consider European options with maturity $T$ and spot price $S_0$, the absence of arbitrage opportunities implies the standard put--call parity relation:
\begin{equation}
    C(K) - P(K) = S_0 e^{-qT} - K e^{-rT},
\end{equation}
where $C(K)$ and $P(K)$ denote the prices of call and put options with strike $K$. This relationship can be reformulated as an affine function of the strike price. Denoting $Y(K) = C(K) - P(K)$ the value of the synthetic forward position, we get
\begin{equation}
    Y(K) = \beta_0 + \beta_1 K,
\end{equation}
where the coefficients $\beta:=(\beta_0, \beta_1)$ identify the financial parameters of interest:
\begin{equation}
    \beta_0 = S_0 e^{-qT} \quad \text{and} \quad \beta_1 = -e^{-rT}.
\end{equation}

\paragraph{Bid--ask bounds for the synthetic forward.}
Denote by $P_i^{\mathrm{bid}},P_i^{\mathrm{ask}}$ the best bid and ask quotes for a European put option with maturity $T$ and strike $K_i$. The no-arbitrage interval for the synthetic forward position is defined by the executable bounds:
\begin{equation}
\label{eq:synthetic_bounds}
B_i \;:=\; C_{i}^{\mathrm{bid}} - P_{i}^{\mathrm{ask}},
\qquad
A_i \;:=\; C_{i}^{\mathrm{ask}} - P_{i}^{\mathrm{bid}}.
\end{equation}
Let $\mathcal{S}_i \;:=\; A_i - B_i \;\ge 0 $ be the related spread and $Y_i^{\mathrm{mid}}$ its midpoint
\begin{equation}
\label{eq:synthetic_mid}
Y_i^{\mathrm{mid}} \;:=\; \frac{A_i+B_i}{2}=\frac{C_{i}^{\mathrm{ask}} + C_{i}^{\mathrm{bid}} }{2}-\frac{P_{i}^{\mathrm{ask}}+P_{i}^{\mathrm{bid}}}{2}.
\end{equation}
Theoretical consistency requires that the model value $y_i(\beta):=\beta_0+\beta_1 K_i$ lies within $[B_i,A_i]$ for all quoted strikes. 

\paragraph{Weighted Least Squares Estimation.}
We estimate $\beta$ by minimizing the distance to the executable bounds.
Consider the symmetric \textit{distance-to-bounds} loss
\(
\ell_i(\beta):=(y_i(\beta)-B_i)^2 + (y_i(\beta)-A_i)^2.
\)
A direct computation yields the identity
\begin{equation}
\label{eq:mid_equivalence}
\ell_i(\beta)
\;=\;
2\bigl(y_i(\beta)-Y_i^{\mathrm{mid}}\bigr)^2 + \frac{\mathcal{S}_i^2}{2}.
\end{equation}
Since the spread term $\mathcal{S}_i^2/2$ is independent of $\beta$, minimizing $\sum_i \omega_i \ell_i(\beta)$ is equivalent to minimizing the weighted squared error to the midpoints
\begin{equation}
\label{eq:wls_objective}
(\beta_0^*,\beta_1^*)
\;\in\;
\arg\min_{\beta_0, \beta_1}
\sum_{i\in\mathcal K_T}
\omega_i\,
\bigl(\beta_0+\beta_1 K_i - Y_i^{\mathrm{mid}}\bigr)^2,
\end{equation}
where $\mathcal K_T$ is the set of strikes simultaneously available for call and put options and $\omega_i$ are strike-specific weights.
Setting $\omega_i = 1$ for all $i\in\mathcal K_T$ yields the OLS case studied in \cite{van2022risk}.

\paragraph{Weighting Scheme.}
Put--call parity couples a call and a put at the same strike: thus, when the call at $K_i$ is out-of-the-money (OTM), the corresponding put is in-the-money (ITM), and conversely. As a result, parity-based estimates are constrained by the liquidity of the ITM leg, which typically exhibits low open interest, small quoted sizes, and wide bid--ask spreads compared to at-the-money (ATM) and OTM options. For very short-dated options, this liquidity drops off sharply for ITM options, their mid-quotes are therefore noisy proxies for an equilibrium price. By contrast, ATM and near OTM options concentrate trading activity and price discovery. In turn, an unweighted parity regression would overreact to low-information observations and yield unstable rates.

To anchor the regression on the most informative quotes while retaining sufficient strikes for estimating the parameters reliably, we use weights that combine (i) an inverse-spread component and (ii) an at-the-money proximity component:
\begin{equation}
\label{eq:weights_components}
\eta_{\mathrm{spr},i} \;:=\; \frac{1}{\mathcal{S}_i},
\qquad
\eta_{\mathrm{atm},i} \;:=\; \frac{1}{|{K_i}/{S_0}-1|+\upsilon},
\end{equation}
where $\upsilon>0$ is a small scale parameter chosen as the ratio between the lowest strike increment (typically 5 points for S\&P 500) and the spot price. We then normalize each component and combine them through a convex combination:
\begin{equation}
\label{eq:weights}
\omega_i
\;:=\;
(1-\alpha)\,\frac{\eta_{\mathrm{spr},i}}{\max_{j\in\mathcal K_T}\eta_{\mathrm{spr},j}}
\;+\;
\alpha\,\frac{\eta_{\mathrm{atm},i}}{\max_{j\in\mathcal K_T}\eta_{\mathrm{atm},j}},
\qquad \alpha\in[0,1].
\end{equation}
This construction\footnote{In the numerical experiments, we set $\alpha=0.5$.} yields the intended ranking: (i) deep ITM strikes are double-penalized (large $|K_i-S_0|$ and typically large $\mathcal{S}_i$), (ii) far OTM strikes are penalized primarily through moneyness, and (iii) near-the-money strikes receive the highest weight. 

\paragraph{Parameter Recovery and Constraints.}
The optimization \eqref{eq:wls_objective} is constrained by $0<\beta_0\leq S_0$ and $-1 \le \beta_1 < 0$ to guarantee non-negative interest rate and yield. The parameters are recovered as:
\begin{equation}
\label{eq:rq_recovery}
    {r}^* = -\frac{1}{T} \ln(-\hat{\beta}_1^*) \quad \text{and} \quad {q}^* = -\frac{1}{T} \ln\left(\frac{\hat{\beta}_0^*}{S_0}\right).
\end{equation}

Finally, we address the specific case of ultra-short-dated options. As the time to expiration approaches zero, the mapping \eqref{eq:rq_recovery} becomes ill-conditioned: a tiny estimation error in $\beta_1$ can translate into a large annualized $r$ due to the factor $1/T$ (and similarly for $\beta_0$ and $q$). In this regime, discounting effects are negligible relative to bid--ask spreads. Therefore, for maturities shorter than one trading day, we set $r=q=0$ to preserve numerical stability.\\

The hybrid regularization procedure introduced in Section \ref{subsub:methodo_rnd} depends on three key numerical parameters: the regularization weights $\lambda_{1}$ and $\lambda_{2}$, which respectively control quadratic smoothness and entropy-based regularization, and the mesh size $\Delta s$ used for the discretization of the state space. In the next sections, we use a simple \citet{BlackScholes1973} toy model to derive explicit orders of magnitude for these quantities. This provides practical guidelines for the choice of $(\lambda_{1},\lambda_{2})$ and $\Delta s$ in the short-maturity regime. We work under the risk-neutral measure, and assume that the underlying price process $(S_{t})_{t\in[0,T]}$ follows the Black--Scholes dynamics. Without loss of generality we set $S_{0}=1$, and we denote by $\sigma$ the (constant) volatility. Then
\[
S_{T} \sim \text{Lognormal}\bigl(\mu,\tilde{\sigma}^{2}\bigr),
\qquad
\mu = \bigl(r - q - \tfrac{1}{2}\sigma^{2}\bigr)T,
\qquad
\tilde{\sigma}^{2} = \sigma^{2} T,
\]
where $\tilde{\sigma}^{2}$ is the total variance of $\log S_{T}$.

\subsection{Regularization Weights}
\label{subsec:reg_weight}
The density $g_{S_{T}}$ of $S_{T}$ is given by
\begin{equation}
\label{eq:density-lognormal}
g_{S_{T}}(y)
= \frac{1}{y \tilde{\sigma} \sqrt{2\pi}}
\exp\!\left(
  - \frac{(\ln y - \mu)^{2}}{2\tilde{\sigma}^{2}}
\right).
\end{equation}
Differentiating \eqref{eq:density-lognormal} gives
\begin{equation}
\label{eq:derivative-density-lognormal}
g_{S_{T}}'(y)
= - \frac{g_{S_{T}}(y)}{y}
\left(
  1 + \frac{\ln y - \mu}{\tilde{\sigma}^{2}}
\right),
\qquad y>0.
\end{equation}

\paragraph{Order of Magnitude of the Quadratic Term.}
We first compute the squared $L^{2}$-norm of $g_{S_{T}}'$,
\[
\| g_{S_{T}}' \|_{2}^{2}
= \int_{0}^{+\infty} \bigl(g_{S_{T}}'(y)\bigr)^{2} \, dy.
\]

\begin{proposition}
\label{prop:L2-norm-gprime-lognormal}
Let $S_{T}$ be lognormally distributed as in \eqref{eq:density-lognormal}. Then the squared $L^{2}$-norm of its density derivative satisfies
\begin{equation}
\label{eq:L2-gprime-exact}
\bigl\| g_{S_{T}}' \bigr\|_{2}^{2}
=
\frac{\exp\!\left(-3\mu + \tfrac{9}{4}\tilde{\sigma}^{2}\right)}{8\sqrt{\pi}}
\frac{\tilde{\sigma}^{2} + 2}{\tilde{\sigma}^{3}}.
\end{equation}
In particular, as $\tilde{\sigma}\to 0$,
\begin{equation}
\label{eq:L2-gprime-asympt}
\bigl\| g_{S_{T}}' \bigr\|_{2}^{2}
=
\frac{1}{4\sqrt{\pi}}\,\tilde{\sigma}^{-3} + O(\tilde{\sigma}^{-2}).
\end{equation}
\end{proposition}

\begin{proof}
By \eqref{eq:derivative-density-lognormal} we have
\begin{equation}
\label{eq:L2-gprime-integral}
\bigl\| g_{S_{T}}' \bigr\|_{2}^{2}
=
\int_{0}^{+\infty}
\frac{g_{S_{T}}(y)^{2}}{y^{2}}
\left(
  1 + \frac{\ln y - \mu}{\tilde{\sigma}^{2}}
\right)^{2} dy.
\end{equation}
Setting $t = \ln y$, we get
\begin{align}
\bigl\| g_{S_{T}}' \bigr\|_{2}^{2}
&=
\int_{\mathbb{R}}
\frac{g_{S_{T}}(e^{t})^{2}}{e^{t}}
\left(
  1 + \frac{t - \mu}{\tilde{\sigma}^{2}}
\right)^{2} dt
\nonumber\\
&=
\frac{1}{2\pi\tilde{\sigma}^{2}}
\int_{\mathbb{R}}
\exp\!\left(
  -3t - \frac{(t-\mu)^{2}}{\tilde{\sigma}^{2}}
\right)
\left(
  1 + \frac{t - \mu}{\tilde{\sigma}^{2}}
\right)^{2}
dt.
\label{eq:L2-gprime-intermediate}
\end{align}
Defining
\[
m = \mu - \frac{3}{2}\tilde{\sigma}^{2},
\]
a direct computation shows that
\begin{equation}
\label{eq:completion-square}
\frac{(t-\mu)^{2}}{\tilde{\sigma}^{2}} + 3t
=
\frac{(t-m)^{2}}{\tilde{\sigma}^{2}}
+ 3\mu - \frac{9}{4}\tilde{\sigma}^{2},
\end{equation}
so that \eqref{eq:L2-gprime-intermediate} becomes
\begin{equation}
\label{eq:L2-gprime-m}
\bigl\| g_{S_{T}}' \bigr\|_{2}^{2}
=
\exp\!\left(-3\mu + \frac{9}{4}\tilde{\sigma}^{2}\right)
\frac{1}{2\pi\tilde{\sigma}^{2}}
\int_{\mathbb{R}}
\exp\!\left(
  - \frac{(t-m)^{2}}{\tilde{\sigma}^{2}}
\right)
\left(
  1 + \frac{t - m - \tfrac{3}{2}\tilde{\sigma}^{2}}{\tilde{\sigma}^{2}}
\right)^{2}
dt.
\end{equation}
We now rewrite the integral in terms of a Gaussian variable $X \sim \mathcal{N}\!\left(m,\tfrac{\tilde{\sigma}^{2}}{2}\right)$.
From \eqref{eq:L2-gprime-m}, we deduce that
\begin{equation}
\label{eq:L2-gprime-expectation}
\bigl\| g_{S_{T}}' \bigr\|_{2}^{2}
=
\exp\!\left(-3\mu + \frac{9}{4}\tilde{\sigma}^{2}\right)
\frac{1}{2\tilde{\sigma}\sqrt{\pi}}
\,
\mathbb{E}\left[
\left(
  -\frac{1}{2} + \frac{X - m}{\tilde{\sigma}^{2}}
\right)^{2}
\right].
\end{equation}
Using the moments of a Gaussian distribution and simplifying the resulting expression, one obtains
\[
\mathbb{E}\left[
\left(
  -\frac{1}{2} + \frac{X - m}{\tilde{\sigma}^{2}}
\right)^{2}
\right]
=
\frac{\tilde{\sigma}^{2} + 2}{4\tilde{\sigma}^{2}},
\]
so that \eqref{eq:L2-gprime-exact} follows from \eqref{eq:L2-gprime-expectation}. The asymptotic behavior in \eqref{eq:L2-gprime-asympt} as $\tilde{\sigma}\to 0$ follows directly.
\end{proof}

\paragraph{Order of Magnitude of the Entropy Term.}
We use the differential entropy $S(\cdot)$ introduced in \eqref{eq:def-H}. For the lognormal density $g_{S_{T}}$ in \eqref{eq:density-lognormal}, a standard computation yields
\begin{equation}
\label{eq:entropy-lognormal}
S\bigl(g_{S_{T}}\bigr)
=
\frac{1}{2}
+ \frac{1}{2}\ln\bigl(2\pi\tilde{\sigma}^{2}\bigr)
+ \mu.
\end{equation}
In the regime $\tilde{\sigma}\to 0$ we therefore have
\begin{equation}
\label{eq:entropy-asympt}
- S\bigl(g_{S_{T}}\bigr)
\sim -\ln(\tilde{\sigma}),
\qquad
\tilde{\sigma}\to 0.
\end{equation}

\begin{remark}
\label{rem:total-variance-small}
The asymptotic regime $\tilde{\sigma} \to 0$ should be interpreted as a regime of small total variance. This is consistent with the parameter range relevant for short-dated options: typical volatilities are in the range $[15\%,40\%]$ and maturities span from one day to one week, so that $\tilde{\sigma}^{2} = \sigma^{2}T$ is small. In this regime, the approximations \eqref{eq:L2-gprime-asympt} and \eqref{eq:entropy-asympt} provide a reasonable description of the relative contributions of the quadratic and entropy terms.
\end{remark}

\paragraph{Balancing the Quadratic and Entropy Contributions}
In order for the quadratic and entropy parts to contribute at the same order of magnitude in our hybrid procedure, we require
\[
\lambda_{1} \bigl\| g_{S_{T}}' \bigr\|_{2}^{2}
\asymp
-\lambda_{2}S\bigl(g_{S_{T}}\bigr)
\qquad\text{as}\quad \tilde{\sigma}\to 0.
\]
Combining \eqref{eq:L2-gprime-asympt} and \eqref{eq:entropy-asympt}, we obtain the asymptotic relation
\begin{equation}
\label{eq:lambda-ratio-scaling}
\frac{\lambda_{1}}{\lambda_{2}}
\sim
-4\sqrt{\pi}\,\tilde{\sigma}^{3}\,\ln(\tilde{\sigma}),
\qquad
\tilde{\sigma}\to 0.
\end{equation}
This relation provides a practical guideline for choosing $(\lambda_{1},\lambda_{2})$ when the total variance $\tilde{\sigma}^{2}$ is small.

\subsection{Mesh Size}
\label{subsec:mesh_size}
We now turn to the discretized version of the problem, as defined in \eqref{eq:def-HM}. We recall that the continuous domain $\mathcal{X}$ is approximated by a uniform grid with mesh size $\Delta s = b/M$. 
Within this toy model, we relate the choice of $M$ to the  distribution of $S_{T}$. Let $F_{S_{T}}$ denote the distribution function of $S_{T}$ and $F_{S_{T}}^{-1}$ its quantile function:
\begin{equation}
\label{eq:lognormal-quantile}
F_{S_{T}}^{-1}(p)
=
\exp\!\bigl(
  \mu + \tilde{\sigma}\ \Phi^{-1}(p)
\bigr),
\qquad p\in(0,1),
\end{equation}
where $\Phi^{-1}$ is the quantile function of a standard normal distribution.

Fix a small probability increment $\delta \in (0,1)$, and define, for $p\in(\tfrac{\delta}{2}, 1-\tfrac{\delta}{2})$, the relative variation
\begin{equation}
\label{eq:relative-variation-def}
\mathcal{R}(p)
=
\frac{
  F_{S_{T}}^{-1}\bigl(p + \tfrac{\delta}{2}\bigr)
  - F_{S_{T}}^{-1}\bigl(p - \tfrac{\delta}{2}\bigr)
}{
  F_{S_{T}}^{-1}\bigl(p - \tfrac{\delta}{2}\bigr)
}.
\end{equation}
This quantity corresponds to the relative change in the underlying level when the cumulative probability increases from $p-\delta/2$ to $p+\delta/2$.

\begin{lemma}
\label{lem:relative-variation-minimizer}
The function $\mathcal{R}$ defined in \eqref{eq:relative-variation-def} admits a unique global minimizer at
\[
p^{*} = \frac{1}{2}.
\]
In other words, the relative variation is minimal when the probability interval is centered at the median of the lognormal distribution.
\end{lemma}

\begin{proof}
Using the expression \eqref{eq:lognormal-quantile}, we may rewrite $\mathcal{R}(p)$ as
\begin{align}
\mathcal{R}(p)
&=
\frac{
  \exp\!\bigl(\mu + \tilde{\sigma}\,\Phi^{-1}(p + \tfrac{\delta}{2})\bigr)
  - \exp\!\bigl(\mu + \tilde{\sigma}\,\Phi^{-1}(p - \tfrac{\delta}{2})\bigr)
}{
  \exp\!\bigl(\mu + \tilde{\sigma}\,\Phi^{-1}(p - \tfrac{\delta}{2})\bigr)
}
\nonumber\\
&=
\exp\!\left(
  \tilde{\sigma}\bigl[
    \Phi^{-1}\!\left(p + \tfrac{\delta}{2}\right)
    - \Phi^{-1}\!\left(p - \tfrac{\delta}{2}\right)
  \bigr]
\right)
- 1.
\label{eq:R-as-exp}
\end{align}
Since the exponential function is strictly increasing and $\tilde{\sigma}>0$, minimizing $\mathcal{R}(p)$ is equivalent to minimizing
\begin{equation}
\label{eq:r-def}
r(p)
=
\Phi^{-1}\!\left(p + \tfrac{\delta}{2}\right)
-
\Phi^{-1}\!\left(p - \tfrac{\delta}{2}\right),
\qquad
p\in\bigl(\tfrac{\delta}{2}, 1-\tfrac{\delta}{2}\bigr).
\end{equation}
The function $r$ is differentiable on its domain. Using the identity
\[
\bigl(\Phi^{-1}\bigr)'(u)
=
\frac{1}{\varphi\bigl(\Phi^{-1}(u)\bigr)},
\]
where $\varphi$ is the standard normal density, we obtain
\begin{equation}
\label{eq:r-derivative}
r'(p)
=
\frac{1}{\varphi\bigl(\Phi^{-1}(p + \tfrac{\delta}{2})\bigr)}
-
\frac{1}{\varphi\bigl(\Phi^{-1}(p - \tfrac{\delta}{2})\bigr)}.
\end{equation}
The first-order optimality condition $r'(p)=0$ is therefore equivalent to
\begin{equation}
\label{eq:equal-densities}
\varphi\bigl(\Phi^{-1}(p + \tfrac{\delta}{2})\bigr)
=
\varphi\bigl(\Phi^{-1}(p - \tfrac{\delta}{2})\bigr).
\end{equation}
Finally, the function $r$ is strictly decreasing on $(\tfrac{\delta}{2}, 1/2]$ and strictly increasing on $[1/2, 1-\tfrac{\delta}{2})$, which shows that $p^{*}=1/2$ is the unique global minimizer of $r$, and therefore of $\mathcal{R}$.
\end{proof}

We now characterize the minimal relative variation in the regime where the increment $\delta$ is small.

\begin{proposition}
\label{prop:R-min-asympt}
In the setting of Lemma~\ref{lem:relative-variation-minimizer}, the minimal value of $\mathcal{R}$ satisfies
\begin{equation}
\label{eq:R-min-asympt}
\mathcal{R}(p^{*})
\sim
\tilde{\sigma}\,\sqrt{2\pi}\,\delta,
\qquad
\delta \to 0.
\end{equation}
\end{proposition}

\begin{proof}
From \eqref{eq:R-as-exp} and \eqref{eq:r-def} we have
\[
\mathcal{R}(p^{*})
=
\exp\!\bigl(\tilde{\sigma}\,r(p^{*})\bigr) - 1,
\qquad
p^{*} = \frac{1}{2}.
\]
By symmetry,
\[
r(p^{*})
=
\Phi^{-1}\!\left(\tfrac{1}{2} + \tfrac{\delta}{2}\right)
-
\Phi^{-1}\!\left(\tfrac{1}{2} - \tfrac{\delta}{2}\right)
=
2\,\Phi^{-1}\!\left(\tfrac{1}{2} + \tfrac{\delta}{2}\right).
\]
For $\delta$ small we may use a first-order Taylor expansion of the quantile function around $1/2$:
\[
\Phi^{-1}\!\left(\tfrac{1}{2} + \tfrac{\delta}{2}\right)
\sim
\frac{\delta}{2}\sqrt{2\pi},
\qquad
\delta\to 0.
\]
Hence
\[
r(p^{*})
\sim
\sqrt{2\pi}\,\delta,
\qquad
\delta\to 0.
\]
Substituting into $\mathcal{R}(p^{*}) = \exp(\tilde{\sigma} r(p^{*})) - 1$ and using the expansion $\exp(x)-1\sim x$ as $x\to 0$ yields
\[
\mathcal{R}(p^{*})
\sim
\tilde{\sigma}\,\sqrt{2\pi}\,\delta,
\qquad
\delta\to 0,
\]
which is \eqref{eq:R-min-asympt}.
\end{proof}

\begin{remark}
\label{rem:numerical-example}
As an illustration, consider $\delta = 0.5\%$, $\sigma=15\%$ and a one-day maturity $T = 1/365$. Then the total volatility is
\[
\tilde{\sigma} = \sigma\sqrt{T}
\approx 0.0079,
\]
so that \eqref{eq:R-min-asympt} yields a minimal relative variation of order
\[
\mathcal{R}(p^{*})
\approx \tilde{\sigma}\,\sqrt{2\pi}\,\delta
\approx 10^{-4}.
\]
Choosing a uniform grid with step $\Delta s = \mathcal{R}(p^{*})$ leads to a number of grid points $M$ of order $10^{4}$ in \eqref{eq:def-HM}. For the same parameter set, the scaling \eqref{eq:lambda-ratio-scaling} implies that the ratio $\lambda_{1}/\lambda_{2}$ is of order $10^{-5}$. These orders of magnitude are consistent with the numerical experiments reported in \Cref{sec:numerical_results}.
\end{remark}
\end{document}